\documentclass{article}
\input{shortcuts.sty}
\usepackage{cite}
\usepackage{amsthm}

\newtheorem{theorem}{Theorem}

\newtheorem{lemma}{Lemma}
\newtheorem{remark}[theorem]{Remark}
\newtheorem{proposition}[theorem]{Proposition}

\begin{document}

\author{Lorenzo Masoero\thanks{Department of Electrical Engineering and Computer
Science, Massachusetts Institute of Technology, \texttt{lom@mit.edu}} \and Federico Camerlenghi\thanks{Department of Economics, Management and Statistics, University of Milano-Bicocca} \and Stefano Favaro\thanks{Department of Economic and Social Sciences, Mathematics and Statistics, University of Torino} \and Tamara Broderick\thanks{Department of Electrical Engineering and Computer
Science, Massachusetts Institute of Technology}}
\title{More for less: Predicting and maximizing genetic variant discovery via Bayesian nonparametrics}
 \date{}

\maketitle

\begin{abstract}
While the cost of sequencing genomes has decreased dramatically in recent years, this expense often remains non-trivial. Under a fixed budget, then, scientists face a natural trade-off between quantity and quality: spending resources to sequence a greater number of genomes (quantity) or spending resources to sequence genomes with increased accuracy (quality). Our goal is to find the optimal allocation of resources between quantity and quality. Optimizing resource allocation promises to reveal as many new variations in the genome as possible. In this paper, we introduce a Bayesian nonparametric methodology to predict the number of new variants in a follow-up study based on a pilot study. We validate our method on cancer and human genomics data. When experimental conditions are kept constant between the pilot and follow-up, we find that our prediction is competitive with the best existing methods. Unlike current methods, though, our new method allows practitioners to change experimental conditions between the pilot and the follow-up. We demonstrate how this distinction allows our method to be used for more realistic predictions and for optimal allocation of a fixed budget between quality and quantity.
\end{abstract}


\section{Introduction} \label{sec:intro}
New genomics data promise to reveal more of the diversity, or variation, among organisms, and thereby new scientific insights. However, the process of collecting genetic data requires resources, and optimal allocation of these resources is typically a challenging task. Under a fixed budget constraint, there is often a natural trade-off between quality and quantity in genetic experiments. Sequencing genomes at a higher quality reveals more details about individual organisms' genomes but incurs a higher cost. Similarly, sequencing a greater number of genomes reveals more about variation across the population but also costs more to accomplish. It is then critical to understand how to optimally allocate a fixed budget between quality and quantity in genetic experiments, in the service of learning as much as possible from the experiment.

To maximize the amount learned from a genetic experiment, we first need to quantify a notion of ``amount learned''. Scientists use a reference genome for a species of interest in an experiment; a (genetic) \emph{variant} is any difference in an observed genome relative to the reference genome. Variants facilitate understanding of evolution \citep{10002015global, mathieson2017differences}, diversity of organisms \citep{10002015global, sirugo2019missing}, oncology \citep{chakraborty2019somatic}, and disease \citep{cirulli2010uncovering, zuk2014searching, bomba2017impact}. Thus, the number of observed variants is a concrete metric of ``amount learned'' from a genetic experiment. For optimal budget allocation, then, we first predict the number of new variants in the follow-up study under different allocations of budget with respect to quality and quantity; next we choose the experimental setting that maximizes the number of new variants.

Optimal budget allocation supports scientists who face resource constraints. In research on non-human and non-model organisms, small sequencing studies are often conducted under limited budgets \citep{da2016next}. The development of reliable and inexpensive sequencing pipelines is thus an active research area \citep{peterson2012double, souza2017efficiency, aguirre2019optimizing}. Accurate prediction of the number of new variants can also be important for understanding the site of origin of cancers as well as the clonal origin of metastasis \citep{chakraborty2019somatic}. And in precision medicine, accurate estimation of the number of new rare variants can aid effective study design and evaluation of the potential and limitations of genomic datasets \citep{momozawa2020unique, zou2016quantifying}. We detail further potential applications in microbiome research, single-cell sequencing, and wildlife monitoring in \Cref{sec:discussion}.

There exists a rich statistical literature on prediction in a follow-up study, relative to a pilot study, when conditions do not change between the pilot and follow-up. We may think of each organism as belonging to multiple groups, where each group is defined by a variant, and the goal is to discover the number of new groups in a follow-up study. A simpler special case of this formulation occurs when each organism belongs to a single group, which is referred to as a species \citep{good1956number, efron1976estimating, lijoi2007bayesian, orlitsky2016optimal}. In general, as in the context of genetic variation, organisms belong to multiple groups that we refer to as \emph{features}. Researchers have developed a wide range of approaches for predicting the number of new features, often interpreted as amount of new genetic variation, in a follow-up study. These approaches include Bayesian methods \citep{ionita2009estimating}, jackknife-based estimators \citep{gravel2014predicting}, linear programming methods \citep{gravel2014predicting, zou2016quantifying}, and variations on the classical Good-Toulmin estimator \citep{orlitsky2016optimal,chakraborty2019somatic}. To the best of our knowledge, though, no existing work provides predictions when the experimental conditions may change between the pilot and follow-up study. And thus no existing work can be used directly for optimal allocation of a fixed budget in experimental design.

Moreover, while there is existing work in other forms of optimal experimental design, it does not fit our goals here. In pioneering work, \citet{ionita2010optimal} propose how to allocate a fixed budget in a pilot study, before any data is observed. While their method treats every dataset the same, our method allows the different variation patterns in different datasets to inform the best follow-up design. 
Separately, researchers have considered how to best choose samples among a number of subpopulations \citep[e.g.][]{dumitrascu2018gt,camerlenghi2020nonparametric}. In this case, the trade-off is between uncertainty and reward, as in classic multi-armed bandit settings, rather than between quality and quantity.

In the present work, we propose a Bayesian nonparametric methodology to predict the number of new variants to be discovered in a follow-up study given observed data from a pilot study. Critically, our approach works when the experimental conditions change between the pilot and follow-up. We then demonstrate how to apply the proposed methodology for optimal budget allocation in the design of a follow-up study given data available from a pilot study. Here, for prediction, we build on a classic Bayesian nonparametric framework for feature allocations known as the beta-Bernoulli process \citep{hjort1990nonparametric, kim1999nonparametric, thibaux2007hierarchical, teh2009indian,broderick2012beta}. The posterior distributions of all our predicted quantities, such as the number of new variants to be discovered, are available in closed-form expressions. Our corresponding Bayesian estimators are simple, computationally efficient, and scalable to massive datasets. In addition, our Bayesian nonparametric framework captures realistic power-law behaviors in genetic data. We will see that, when the pilot and follow-up studies are constrained to have the same experimental setup as in previous work, our predictions are competitive with the state-of-the-art and superior to a number of recent proposals. Most importantly, though, we demonstrate that our predictions maintain their accuracy when experimental conditions change between the pilot and follow-up. Finally, we give an empirical demonstration of how our predictions can be used for designing the follow-up study with an optimal allocation of a fixed budget between quality and quantity. We validate the proposed methodology on synthetic and real data, with a focus on human genomics. Specifically, we consider the TCGA and MSK-impact datasets \citep{cheng2015memorial}, as well as the recent gnomAD dataset of \citet{karczewski2019variation}.


 \section{Data and modeling assumptions} \label{sec:model}

 Modern high-throughput sequencing technologies allow accurate determination of an organism's genome \citep{reuter2015high}. A reference genome serves as a fixed representative, and variants relative to the reference genome can take many forms, including deletions, inversions, translocations, and insertions; see \citet{taylor2004current} and references therein. In the present work, we do not distinguish between different forms of variants, though in \Cref{sec:discussion} we briefly discuss how our framework could be extended to make this distinction. To establish notation and start building up to our Bayesian nonparametric model, we first assume that the process of observing variants is flawless; we develop a more realistic model for observations in \Cref{sec:seq_err}.

 Suppose there are $J$ variants observed among the $N$ pilot genomes, $0 \le J < +\infty$, with $\aloc_j$ the label of the $j$-th variant in order of appearance. Let $\acount_{n,j}$ equal $1$ if the variant with label $\aloc_j$ is observed for the $n$-th organism; otherwise, let $\acount_{n,j}$ equal $0$. We collect data for the $n$-th organism in $\mcount_{n} := \sum_{j=1}^{J} \acount_{n,j} \delta_{\aloc_j}$, which pairs each variant observation with the corresponding variant label by putting a mass of size $\acount_{n,j}$ at location $\aloc_j$. We use the notation $\mcount_{N_1:N_2}$, where $N_1 \le N_2$, to denote $(\mcount_{N_1}, \mcount_{N_1+1}, \mcount_{N_1+2}, \ldots, \mcount_{N_2})$. Given the observable $\mcount_{1:N}$, we consider a Bayesian approach to predict the number of variants in the follow-up study. Specifically, letting $\mfreq$ be an appropriate latent parameter, we specify a generative model via a likelihood function $\PP(\mcount_{1:N}\, |\, \mfreq)$ and a prior distribution $\PP(\mfreq)$. Technically there is a fixed, and finite, upper bound on the number of possible variants established by the size of any individual genome. But this bound is usually much larger, often by orders of magnitude, than the number of observed variants. Moreover, in practice, we expect that no study of any practical finite size $N$ will reveal all possible variants, simply because some variants are so exceedingly rare. Bayesian nonparametric methods allow us to avoid hard-coding an unwieldy, large finite bound that may cause computational and modeling headaches. In particular, they allow the observed number of variants to be finite for any finite dataset and grow without bound, in such a way that computation typically scales closely with the actual number of variants observed. Formally, we imagine a countable infinity of latent variants, labelled as $\{\aloc_j\}_{j\geq1}$, and we write $\mcount_{n} := \sum_{j\geq1} \acount_{n,j} \delta_{\aloc_j}$; since $\acount_{n,j} = 0$ for all unobserved variants, this equation reduces to the previous definition of $\mcount_n$ above.

 Following existing methods for estimating new-variant cardinality \citep{ionita2009estimating, gravel2014predicting, zou2016quantifying,orlitsky2016optimal,chakraborty2019somatic}, we assume that every variant appears independently of every other variant; that is, $\acount_{n,j}$ is independent of $\acount_{n,k}$ across all $n$ for $j \ne k$. In reality, nearby positions on a genome can be highly correlated; this phenomenon is called linkage disequilibrium. However, our assumption has two principal advantages: (i) it makes our computations much easier; (ii) it is supported by our state-of-the-art empirical results in \Cref{sec:exp}. We also make the milder assumption that organisms are (infinitely) exchangeable; roughly, we assume that the order in which we observe the sample organisms is immaterial for any sample size $N$. Since, for the moment, we assume variant observation is flawless, this assumption presently translates into an exchangeability assumption on the observed data. More precisely, let $[N] := \{1,\ldots,N\}$, and let $\discount_N$ represent a permutation of $[N]$. Then, for the variant with label $\aloc_j$, for any $N$ and any $\discount_N$, we assume $\PP(\acount_{1,j}, \ldots, \acount_{N,j}) = \PP(\acount_{\discount_N(1), j}, \ldots, \acount_{\discount_N(N), j})$. Indeed, if we expected systematic variation among organisms in our population between earlier and later samples, we would find it difficult to predict future data from past data without knowing more about the nature of the variation.

 Exchangeability of $\{\acount_{n,j}\}_{n\geq1}$ implies the existence of a random variable $\afreq_j$, i.e. the variant's proportion, such that the $\acount_{n,j}$ are Bernoulli draws with parameter $\afreq_j$, independently and identically distributed across $n$ \citep{finetti31}. We pair each $\afreq_j$  with its variant's label $\aloc_j$ in a random measure $\mfreq := \sum_{j\geq1} \afreq_j \delta_{\aloc_j}$, and we assume the $\mcount_n$'s are conditionally independent and identically distributed given $\mfreq$. In addition, we make the following modeling assumptions: (i) the conditional distribution of $\mcount_n$ given $\mfreq$ is the distribution of a Bernoulli process ($\BeP$) with parameter $\mfreq$, and we write $\mcount_n\, |\, \mfreq \stackrel{iid}{\sim} \BeP(\mfreq)$; (ii) the prior distribution on $\Theta$ is the law of the three-parameter beta process ($\tBP$) \citep{teh2009indian, broderick2012beta}. In agreement with the assumption of independence for $\{\acount_{n,j}\}_{n\geq1}$ across $j$, we can interpret the three-parameter beta process as a collection of independent priors on the $\afreq_j$ such that it satisfies our goals: (G1) a finite number of observed variants in any finite sample; (G2) a number of observed variants that is unbounded as the number of samples grows. Furthermore, the three-parameter beta process is able to capture power-law behaviours \citep{teh2009indian, broderick2012beta}, which are common in physical processes. The three-parameter beta process is characterized by: (i) a mass parameter $\mass$ that scales the total number of variants observed; (ii) a discount parameter $\discount$ that controls the power-law growth in observed variant cardinality; (iii) a concentration parameter $\conc$ that modulates the frequency of more widespread variants.

 We say that the random measure $\mfreq$ is distributed as a three-parameter beta process, $\mfreq \sim \tBP(\mass, \discount, \conc)$, if $\mfreq ~=~ \sum_{j\geq 1} \theta_j \delta_{\aloc_{j}}$, with $\{\afreq_j\}$ drawn from a Poisson process with rate measure 
  \begin{displaymath} 
  	\ratemeas(\d \afreq)
  		 = \mass \frac{\Gamma(1+\conc)}{\Gamma(1-\discount)\Gamma(\conc+\discount)} \afreq^{-1-\discount} (1-\afreq)^{\conc+\discount-1} \bm{1}(\afreq \in [0,1])\d \afreq,
  \end{displaymath}
 where $\bm{1}(A)$ stands for the indicator function of the event $A$. The $\aloc_{j}$'s serve merely to distinguish the variants, so it is enough to ensure that they are all almost surely distinct. Thus we take $\aloc_{j} \stackrel{iid}{\sim} \Unif[0,1]$. The Poisson point process representation is convenient in our proofs. To meet goals G1 and G2, the three-parameter beta process hyperparameters must satisfy: $\mass > 0$, $\conc > -\discount $, and $\discount \in [0,1)$ \citep{teh2009indian, james2017bayesian, broderick2018posteriors}.


 \section{Predicting the number of new variants} \label{sec:prediction}
 %
 \subsection{Initial proposals for prediction} \label{sec:prediction_naive}

 In \Cref{sec:model} we introduced and motivated a Bayesian nonparametric model consisting of: (i) a Bernoulli process likelihood function, $\mcount_n \mid \mfreq \stackrel{iid}{\sim} \BeP(\mfreq)$, for observed variants conditioned on variants' proportions; (ii) a three-parameter beta process prior, $\mfreq \sim \tBP(\mass, \discount, \conc)$ over variants' proportions. Now we use this model to predict the number, $\countnew{N}{M}$, of new variants in a follow-up study of size $M\geq1$ after an initial pilot study of size $N\geq1$:
 \begin{displaymath}
 \countnew{N}{M} := \sum_{j\geq1} \bm{1}\left( \sum_{n=1}^{N} \acount_{n,j} = 0\right) \bm{1}\left( \sum_{m=1}^{M} \acount_{N+m,j} > 0 \right).
 \end{displaymath}
 We derive the posterior distribution of $\countnew{N}{M}$ given $\mcount_{1:N}$. So the expected value of the posterior distribution is a Bayesian nonparametric estimator of $\countnew{N}{M}$ with respect to a squared loss function. With a slight abuse of notation, for any two random variables $X$ and $Y$ defined on the same probability space we let $X\,|\,Y$ denote the random variable whose distribution coincides with the conditional distribution of $X$ given $Y$. We write $\mathcal{N}(\mu,\rho^{2})$ for a Gaussian random variable with mean $\mu$ and variance $\rho^{2}$, and we let $(a)_{b \uparrow} := \prod_{i=1}^{b} (a + i-1)$ denote the rising factorial.

 \begin{proposition}
 \label{prop:naive_new}
 Let $\mfreq \sim \tBP(\mass, \discount, \conc)$ and $\mcount_n \mid \mfreq \stackrel{iid}{\sim} \BeP(\mfreq)$ for $n =1, \ldots, N$ and $N\geq1$. Then,
 \begin{align} \label{eq:naive_new}
 	\countnew{N}{M} \mid \mcount_{1:N} &\sim \Pois \left\{ \mass \sum_{m=1}^M \frac{(\conc+\discount)_{(N+m-1) \uparrow}}{(\conc+1)_{(N+m-1) \uparrow}} \right\}.
 \end{align}
 \end{proposition}

 From \Cref{prop:naive_new}, the Bayesian nonparametric estimator of $\countnew{N}{M}$ under squared loss is
 \begin{displaymath}
 \pred{N}{M} := E \left( \countnew{N}{M} \mid \mcount_{1:N} \right)=\mass \sum_{m=1}^M \frac{(\conc+\discount)_{(N+m-1) \uparrow}}{(\conc+1)_{(N+m-1) \uparrow}}.
 \end{displaymath}
 $\pred{N}{M}$ predicts the number of new variants in a follow-up study.  In the next result we show that the distribution of $\countnew{N}{M}\mid \mcount_{1:N} $ exhibits almost-sure power-law growth in the sample size $N$ with power determined by the three-parameter beta process hyperparameters. We also characterize asymptotic noise around the posterior predictive mean. See \Cref{sec:app-proofs} for proofs of \Cref{prop:naive_new} and  \Cref{prop:asymptotic}.

 \begin{proposition}
 \label{prop:asymptotic}
 Under the setting of \Cref{prop:naive_new},
 \begin{align}
 	\frac{ \countnew{N}{M} }{M^{\discount}}\, \bigg\vert\, \mcount_{1:N}\, \stackrel{\text{a.s.}}{\longrightarrow}\, \xi \quad \textrm{ as } M \rightarrow +\infty, \label{eq:slln}
 \end{align}
 where $ \xi := \frac{\mass}{\discount}\frac{\Gamma(\conc+1)}{\Gamma(\conc+\discount)}$. The \Cref{eq:slln} limit holds almost surely, conditionally given $\mcount_{1:N}$. Also,
 \begin{align}
 	\sqrt{M^{\discount}} \left( \frac{\countnew{N}{M}}{M^{\discount}} - \xi \right)\, \bigg\vert\, \mcount_{1:N}\, \stackrel{\text{d}}{\longrightarrow}\, \mathcal{N} (0, \xi) \quad \textrm{ as } M \rightarrow +\infty, \label{eq:clt}
 \end{align}
 where the limit in \Cref{eq:clt} holds true in distribution.
 \end{proposition} 

  Besides $\countnew{N}{M}$, researchers may be interested in relatively rare new variants since rare variants are known to play a role in disease predisposition \citep{cirulli2010uncovering, saint2014important, bomba2017impact}. In particular, let $\countnewtimes{N}{M}{r}$ denote the number of new variants that occur exactly $r$ times in the follow-up study, and let $\countnewtimes{N}{M}{\le R}$ denote the number of new variants that occur at most $R$ times in the follow-up study. A suitably chosen small value of $r$ or $R$ encodes a notion of rareness for variants. See \Cref{prop:naive_new_num_times} and \Cref{prop:asymptotic_numtimes} for a characterization of the posterior distributions of $\countnewtimes{N}{M}{r}$ and $\countnewtimes{N}{M}{\le R}$ given $\mcount_{1:N}$.

 Our propositions reveal key attributes of our Bayesian nonparametric estimators. First and foremost, the posterior distribution of $\countnew{N}{M}$ depends on $\mcount_{1:N}$ only via the initial sample size $N$. See \Cref{eq:naive_new_num_times} for similar behavior in the posterior distributions of $\countnewtimes{N}{M}{r}$ and $\countnewtimes{N}{M}{\le R}$. Moreover, from \Cref{prop:asymptotic} we see that the large-$M$ asymptotic behavior of the posterior distribution of $\countnew{N}{M}$ is completely determined by hyperparameters of the three-parameter beta process; see \Cref{prop:asymptotic_numtimes} for similar behavior in the posterior distributions of $\countnewtimes{N}{M}{r}$ and $\countnewtimes{N}{M}{\le R}$. Therefore, learning hyperparameters of the three-parameter beta process from the observed data is critical. In \Cref{sec:empirics}, we  propose an empirical Bayes procedure to this end.

 Our Bayesian nonparametric approach above, like existing approaches for estimating the number of new variants in a follow-up study \citep{ionita2009estimating, gravel2014predicting, zou2016quantifying,orlitsky2016optimal,chakraborty2019somatic}, relies on the assumption that variants are always observed under the same conditions. Moreover, none of these methods account for how improved variant observation quality may incur a larger cost. But conditions may change between pilot and follow-up experiments, and these changes may be informed by an experimental budget. We address these issues below. In \Cref{sec:seq_err}, we show how our general Bayesian nonparametric framework can be adapted to the case where variants are not observed perfectly; in fact, we show how we can adapt to different experimental conditions between the pilot and follow-up. Then, in \Cref{sec:opt_design}, we build on the work of \citet{ionita2010optimal} to show that
 we can optimize for the best conditions, to yield the most variants, in the follow-up. That is, we next consider the challenging problem of optimal allocation of a fixed budget between quality and quantity in genomic experiments: spending resources for sequencing a greater number of genomes (quantity) or spending resources for sequencing with increased accuracy (quality).

 \subsection{Accounting for sequencing errors} \label{sec:seq_err}

 We extend the Bayesian nonparametric estimator introduced in \Cref{sec:prediction_naive} to account for non-trivial sequencing error. In \Cref{sec:prediction_naive} we have assumed that if any organism exhibits a variant, that variant is detected, i.e., $\acount_{n,j} = 1$ for organism $n$. However, in practice, sequencing a genome is a complex and noisy process. Millions of reads of fragments of the same genomic sequence need to be aligned and compared to the reference genome. Every position $j$ of the genome of individual $n$ is read a random number $\depth_{n,j}$ of times. $\depth_{n,j}$ is the (random) sequencing depth of the process. Out of these $\depth_{n,j}$ times, $\depth_{n,j,\text{err}}$ reads give rise to an error, due to technological imperfections, and are discarded. Here, $0 \leq \depth_{n,j,\text{err}} \leq \depth_{n,j}$. The remaining $\depth_{n,j,\text{noerr}} = \depth_{n,j} - \depth_{n,j,\text{err}}$ reads are correctly processed, aligned to the reference genome, and recorded \citep{ionita2010optimal}.  Every error-free read can either agree with the reference genome, or disagree. We let $C_{n,j} \in \{0,1,\dots, \depth_{n,j,\text{noerr}}\}$ denote the number of times that reads are correctly processed and we observe disagreement with the reference genome. Finally, a variant is said to be called whenever some discrepancy criterion, i.e. the variant calling rule, is satisfied.

 Following \citet{ionita2010optimal}, we focus on simple threshold variant calling rules. That is, a variant is called whenever a sufficient number of reads disagree with the reference genome. Given the threshold value $\threshold>0$, variation is declared if the count $C_{n,j}$ exceeds $\threshold$, i.e.\ $\acount_{n,j} = \bm{1}(C_{n,j} \geq \threshold)$. This threshold variant calling rule is a simplification of actual variant callers used in modern genomic pipelines; see e.g.\ \citep{xu2018review} for a review of variant calling algorithms. While simplistic, the threshold rule has the benefit of being easy to interpret; by contrast, state-of-the-art alternatives are much more complex, to the point of being somewhat inscrutable by their users. In fact, understanding how to tailor the variant calling rule to the data-gathering process is itself an active area of research \citep{hwang2015systematic,cornish2015comparison,kumaran2019performance}.

 In setting up our model to account for sequencing error, we make the following additional assumptions. (i) Following standard practice in the genetics literature  \citep[e.g.,][]{lander1988genomic, ionita2010optimal,sampson2011efficient}, we assume that the sequencing depth $\depth_{n,j}$ is a Poisson random variable with parameter $\lambda$, which we refer to as the sequencing quality. (ii) The reads are independent and identically distributed across individuals and positions. (iii) $p_{\text{err}}$ is a fixed probability  of reading error that depends on the sequencing technology. (iv) Conditionally on  $\depth_{n,j}$ total reads, the number of error-free reads $\depth_{n,j,\text{noerr}}$ is a binomial random variable, with $\depth_{n,j}$ as the number of trials and $1-p_{\text{err}}$ as the probability of success in a trial. Under these assumptions (i) -- (iv), as showed in \Cref{proof:poisson_thinning} in \Cref{sec:poisson_thinning}, the probability of obtaining at least $\threshold$ successful reads at any position $j$ for any individual $n$ is
 \begin{align}
 	\probcall(\lambda, \threshold, p_{\text{err}})
 		&:=  \sum_{t \geq \threshold} \frac{e^{-\lambda} \lambda^t}{t!}
 			\sum_{i=\threshold}^t \binom{t}{i} (1-p_{\text{err}})^i p_{\text{err}}^{t-i} =
 		\sum_{t\geq \threshold} \frac{e^{-\lambda(1-p_{\text{err}})}\{\lambda(1-p_{\text{err}})\}^t}{t!}. \label{eq:poisson_thinning}
 \end{align} 
 We still assume $\mfreq \sim \tBP(\mass, \discount, \conc)$ for the prior distribution over variant proportions. As in \Cref{sec:prediction_naive}, we draw whether organism $n$ has variant with proportion $\afreq_j$ according to $\Bern(\afreq_j)$. If the organism does have the variant, we now draw whether we observe the variant according to $\Bern(\probcall)$, with $\probcall=\probcall(\lambda, \threshold, p_{\text{err}})$.  Hence, the probability of declaring the presence of variant $j$ is now given by  $\PP(C_{n,j} ~\geq ~\threshold ~\mid \Theta) ~=~ \theta_j \probcall$.

 Observe that $\probcall$ is modulated by the parameter $\lambda$, which controls the sequencing depth and can be set by the practitioner. \citet{ionita2010optimal} considered a setting with a single study, where that study is yet to be run. In this section, unlike the work of \citet{ionita2010optimal}, we assume that we have access to data from a pilot study when designing a follow-up study. We use subscripts to denote potentially different values of $\probcall$ across experiments. For instance, the practitioner may choose a sequencing depth in the follow-up study that is different from the sequencing depth in the pilot study. Hence we write $\probcallinit = \probcall(\lambda_{\text{pilot}},\threshold,p_{\text{err}})$ for the pilot experiment and $\probcallfollowup = \probcall(\lambda_{\text{follow}},\threshold,p_{\text{err}})$ for the follow-up. Our methods can be immediately extended to the case where there are multiple initial experiments with different $\probcall$ values.
 \begin{proposition}
  \label{prop:noise_new}
 Let $\mfreq \sim \tBP(\mass, \discount, \conc)$,  that is $\mfreq := \sum_{j\geq1} \afreq_j \delta_{\aloc_j}$. Furthermore, let $\mcount_n \mid \mfreq \stackrel{iid}{\sim} \BeP(\mfreq_{\text{pilot}})$, where $\mfreq_{\text{pilot}}:=\sum_{j\geq1} \probcallinit\theta_j\delta_{\psi_j}$, for $n=1,\ldots,N$ and $N\geq1$, and let $\mcount_{N+m} \mid \mfreq \stackrel{iid}{\sim} \BeP(\mfreq_{\text{follow}})$, where $\mfreq_{\text{follow}}:=\sum_{j\geq1} \probcallfollowup\theta_j\delta_{\psi_j}$, for $m=1,\ldots,M$ and $M\geq1$. Then,
 \begin{align} \label{eq:noise_new}
 	\countnew{N}{M} \mid \mcount_{1:N}
 		&\sim \Pois (\gamma),
 \end{align}
 with $\gamma :=  \mass \probcallfollowup \sum_{m=1}^M E\{ (1-\probcallfollowup B)^{m-1} (1-\probcallinit B)^N \}$ and $B \sim \Betadist(1 - \discount, \conc+\discount)$.
 \end{proposition}
 The expected value of the posterior distribution in \Cref{prop:noise_new} provides a Bayesian nonparametric estimator, with respect to a squared loss function, of $\countnew{N}{M}$. Namely, this estimator is
 \begin{displaymath}
 \mass \probcallfollowup \sum_{m=1}^M E\{ (1-\probcallfollowup B)^{m-1} (1-\probcallinit B)^N \},
 \end{displaymath}
 where $B \sim \Betadist(1 - \discount, \conc+\discount)$. This new estimator extends \Cref{sec:prediction_naive} to the case where sequencing error is taken into account. We defer the proof of \Cref{prop:noise_new} to \Cref{sec:app-proofs}.


 \section{Empirics for the prediction} \label{sec:empirics} 

 Our more realistic model of variant observation sets up a prediction framework for the number of new variants in a follow-up experiment. But without further development, we still face the difficulty that our predictor from \Cref{eq:naive_new} does not use any information about the pilot experimental data except its cardinality. Recall that the hyperparameters $\mass, \discount, \conc$ control the behavior of the estimator (\Cref{prop:asymptotic}). So we will induce a dependency on the observed pilot data by fitting these hyperparameter values to the pilot data. One common approach in empirical Bayes is to maximize the probability of the data given the hyperparameters: $\argmax_{\mass, \discount, \conc} \PP(\mcount_{1:N} | \mass, \discount, \conc)$ with $\PP(\mcount_{1:N} | \mass, \discount, \conc) = \int_{\mfreq} \PP(\mcount_{1:N} | \mfreq) \PP(d\mfreq | \mass, \discount, \conc)$. In the case without sequencing errors, this probability can be expressed in closed form as the exchangeable feature probability function (EFPF) \citep{broderick2013feature}. However, with sequencing errors, the integral can be very high-dimensional and expensive to compute with Markov chain Monte Carlo. Moreover, even without sequencing errors, the exchangeable feature probability function for the beta process is a complex function of sums, products, quotients, and exponentiation of gamma functions  \citep[Eq.\ 8]{broderick2013feature}, which we find can lead to numerical instability in the optimization.

 An easier choice is to treat the prediction from our model as a regression function with its own parameters $\mass, \discount, \conc$. We can fit these parameters to the pilot project data by imagining subsets of the true pilot data as mini-pilot projects themselves and directly minimizing error in prediction on the remaining pilot data. In particular, consider index $n \in [N]$ as the size of the imagined mini-pilot. Then, by our earlier definition, $\pred{n}{m}$ is the prediction for the number of new variants in the next $m$ data points given the first $n$ data points. Here we write $\pred{n}{m}(\mass, \discount, \conc)$ to emphasize the hyperparameter dependence. Let $\countnew{n}{m} \mid \mcount_{1:N}$ be the true number of new variants in the next $m$ data points (for $m$ such that $n+m \le N$) given the first $n$ data points. Then we solve
 \begin{equation}
 	\estmass, \estdiscount, \estconc := \argmin_{\substack{\mass, \discount, \conc: \\ \mass >0, \; \discount \in [0,1), \conc > - \discount}}
 		\sum_{m=1}^{N-n} \left\{ \pred{n}{m}(\mass, \discount, \conc) - \left( \countnew{n}{m} \mid \mcount_{1:N} \right) \right\}^2. \label{eq:cost}
 \end{equation}
 We set $n = \lfloor 2/3\times N \rfloor$, a choice that works well across all applications we consider here. To find $\estmass, \estdiscount, \estconc$ we use the differential evolution algorithm \citep{storn1997differential}. We also considered using multiple folds of the pilot study, in the style of cross validation, instead of a single train-test split. In our experiments, we did not observe a noticeable difference between our proposal in \Cref{eq:cost} and this more-involved procedure. We choose to minimize the 2-norm, but \Cref{eq:cost} can be straightforwardly adapted for other standard choices of error (e.g., 1-norm). Finally, we use $\estpred{N}{M} := \pred{N}{M}(\estmass, \estdiscount, \estconc)$ as our estimator for the number of new variants in the follow-up study of size $M$ after observing data from the pilot study of size $N$.


 \section{Sequencing errors and optimal experimental design} \label{sec:opt_design}

 Our goal is to maximize the number of variants we expect to observe under a fixed budget. To see how the budget comes into play, note two cost sources in the follow-up study. (i) It costs more to increase the number of samples $M$ since sequencing each additional sample adds an additional cost. (ii) Likewise, it costs more to increase the quality of each sample, where increasing quality is accomplished by increasing the sequencing quality in the followup, $\seqfollowup$. We might encode the total cost as a function of these settings: $\costfcn(M, \seqfollowup)$. Here, $\costfcn$ is increasing in both of its arguments. Conversely, we expect to discover more variants as either of $M$ or $\seqfollowup$ increases and fewer variants as either quantity decreases. Therefore, we face a trade-off in where to best allocate experimental budget between $M$ and $\seqfollowup$.

 Our framework allows us to precisely quantify and optimize this trade-off. In particular, we now emphasize the dependence of $\estpred{N}{M}$ on $\seqfollowup$, via $\probcallfollowup$, by writing $\estpreddepth{N}{M}{\seqfollowup}$ for $\estpred{N}{M}$ computed with $\seqfollowup$. Since we can compute $\estpreddepth{N}{M}{\seqfollowup}$ across values of $M$ and $\seqfollowup$ using \Cref{eq:noise_new}, we can optimize to find the maximum possible predicted variants under some budget $\budget$. We are interested in the experimental settings under which this maximum is achieved:
 \begin{equation} \label{eq:opt}
 	\argmax_{M, \seqfollowup} \estpreddepth{N}{M}{\seqfollowup} \quad \textrm{ subject to } \quad \costfcn(M, \seqfollowup) \le \budget.
 \end{equation}
 To the best of our knowledge, no previous methods \citep{ionita2009estimating, ionita2010optimal,gravel2014predicting, zou2016quantifying,orlitsky2016optimal,chakraborty2019somatic} have been designed or modified to predict variants under different experimental conditions in a follow-up study given results from a pilot. We believe the Bayesian nonparametric framework we adopt here allows particularly straightforward handling of different sequencing depths, and more generally different experimental setups. Notably, \citet{ionita2010optimal} consider experimental design, but only for a single future study, without observing any pilot data. Given its Bayesian grounding, their associated estimator might be adapted to our pilot and follow-up framework using similar techniques to those we introduce above. But we will see in \Cref{sec:exp} that the quality of their estimator is much worse than that of our method; any corresponding experimental design would therefore suffer. We suspect our gains are due to the flexibility of the Bayesian nonparametric framework and ability to capture power laws in the data. 

 Note that practitioners might instead be interested in maximizing the number of new rare variants in the follow-up study, i.e. variants that appear at most $R$ times in the follow-up sample. In this case, we can still apply empirical Bayes estimates of hyperparameters $\estmass, \estdiscount, \estconc$ obtained via \Cref{eq:cost}. In particular, let $\estpredtimesdepth{N}{M}{\le R}{\seqfollowup}$ be the Bayesian nonparametric estimator of the number of new rare variants, with hyperparameter values set to $\estmass, \estdiscount, \estconc$ and follow-up sequencing quality set to $\seqfollowup$. Then, to maximize the number of new rare variants, we solve the optimization problem in \Cref{eq:opt} with $\estpredtimesdepth{N}{M}{\le R}{\seqfollowup}$ in place of $\estpreddepth{N}{M}{\seqfollowup}$. We highlight that we still suggest learning the hyperparameters $\estmass, \estdiscount, \estconc$ via the original optimization problem, with the predictor $\estpreddepth{N}{M}{\seqfollowup}$ of all new variants. We make this recommendation since rare variants may be sparser in the pilot study and thereby provide less information about these hyperparameters.


 \section{Experiments} \label{sec:exp}

 \subsection{Experimental setup}
 \label{sec:exp_setup}

 We evaluate our methods on both synthetic and real data. Code is available at \url{https://bitbucket.org/masoero/moreforless_bayesiandiscovery/src/master}. For real data, we use human cancer genomics datasets. In cancer genomics, rare variants may be useful in developing effective clinical procedures and understanding cancer biology, and researchers have recognized the importance of appropriate sequencing depth in the data-gathering process \citep{griffith2015optimizing,rashkin2017optimal}. Following the setup of \citet{chakraborty2019somatic}, we consider the Cancer Genome Atlas (TCGA), a large and publicly available cancer genomics dataset. It contains somatic mutations from $N = 10{,}295$ patients and spans $33$ different cancer types. See \Cref{sec:app_TCGA_msk} for more details on the data. In what follows, we show that our method produces accurate predictions when the sequencing depth is kept constant (\Cref{sec:exp_noerror}); we show it is the only method that can produce accurate predictions under changing conditions (\Cref{sec:exp_different}) and the only method that can inform optimal design of experiments (\Cref{sec:exp_opt_d}).  In \Cref{sec:app_TCGA_msk} we  report additional cancer genomics results, including with the MSK-impact database, a targeted sequencing study also used by  \citet{chakraborty2019somatic}.

 In \Cref{sec:app-exp_gnomAD}, we report results for the Genome Aggregation Database \citep{karczewski2019variation}, a recent extension of the Exome Aggregation Consortium data set \citep{lek2016analysis} and the largest publicly available human genomic dataset.  We include additional experiments on synthetic data in \Cref{sec:app-exp_additional}, to illustrate when and why different methods may fail. 

 \subsection{Prediction with no sequencing errors} \label{sec:exp_noerror}

 \begin{figure}[t!]
       \centering \includegraphics[width=\textwidth,height=\textheight,keepaspectratio]{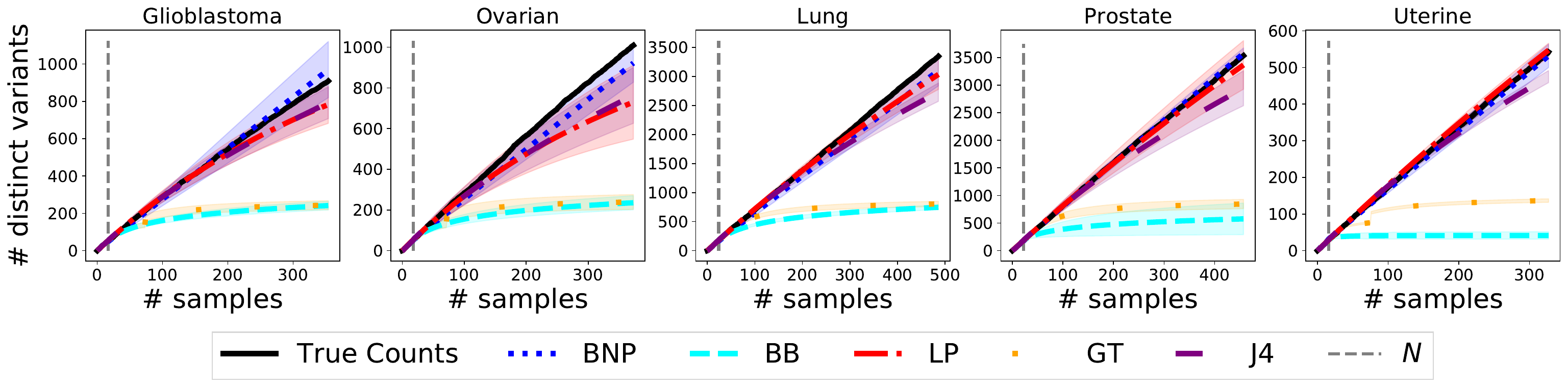}
 \caption{TCGA dataset: predicting the number of new variants. The true number of distinct variants (black) is compared to: our method (blue, BNP); Bayesian parametric (cyan, BB); linear program (red, LP); Good-Toulmin (orange, GT); 4th order jackknife (purple, J4). Shaded regions represent one standard deviation.}
 \label{fig:prediction_all}
 \end{figure}
 \begin{figure}[t!]
       \centering \includegraphics[width=\textwidth,height=\textheight,keepaspectratio]{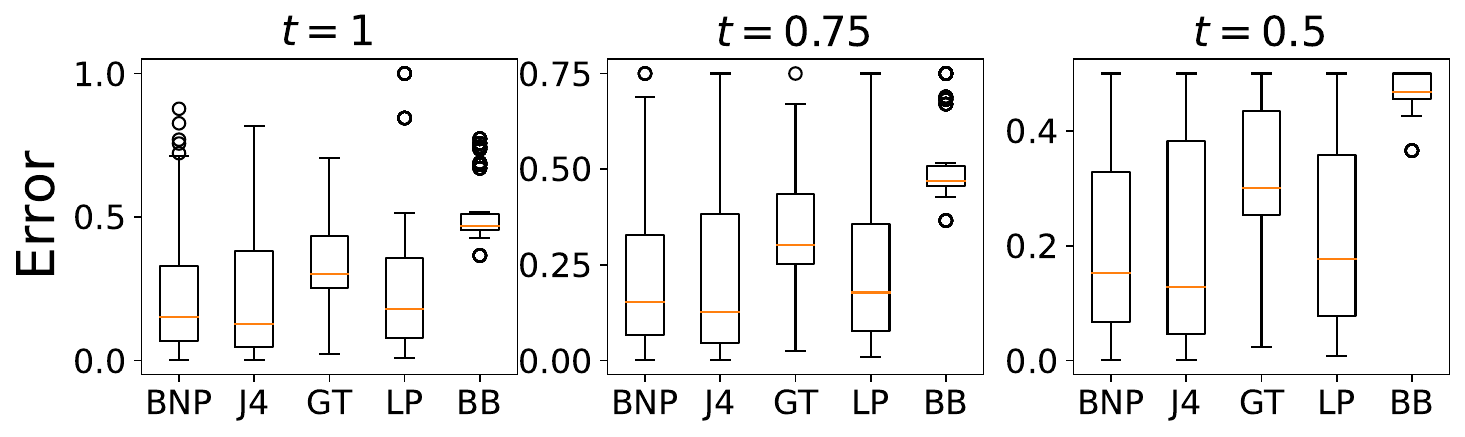}
 \caption{Trimmed percent prediction error on the TCGA data across all $33$ cancer types and $20$ folds, for different trimming thresholds $t$ (\Cref{eq:loss_ape}). We compare our method (BNP) to Bayesian parametric (BB), linear program (LP), Good-Toulmin (GT), 4-th order jackknife (J4).}
 \label{fig:prediction_same_conditions_boxplot}
 \end{figure}

 Researchers have developed several approaches for predicting the number of new variants in a follow-up study under the assumption of perfect recovery of variants: e.g., parametric Bayesian methods \citep{ionita2009estimating}, linear programming methods \citep{gravel2014predicting, zou2016quantifying}, a harmonic jackknife \citep{gravel2014predicting}, and a smoothed version of the classic Good-Toulmin estimator \citep{chakraborty2019somatic}. To assess the prediction error under constant sequencing conditions, we focus on the TCGA dataset. We partition the samples in the dataset into $33$ datasets according to cancer-type annotation of each patient.  For each cancer type, we predict the number of new variants that will be observed in a follow-up sample given a pilot sample. To do so, we use an approach akin to cross validation; namely we treat each cancer type as a dataset. We divide the dataset into 20 folds of equal size; we consider each fold in turn as a pilot study and treat the remaining folds as the follow-up. A smaller number of folds corresponds to a larger pilot study. All methods improve when the pilot study is increased substantially in size, i.e.\ when there is more information in the pilot. We find that the choice of 20 folds creates a challenging scenario with a small amount of pilot information. Nonetheless, our method, the harmonic jackknife, and the linear program all still perform well in these conditions.

 We follow \citet{zou2016quantifying} to visualize our results for five cancer types in \Cref{fig:prediction_all}; namely, we plot the number of total predicted variants, averaged across folds, as a function of total data points (pilot plus follow-up). A vertical dashed line marks the pilot size; non-trivial predictions are to the right of this line. Shaded regions indicate one empirical standard deviation, measured across the folds.
 \Cref{fig:prediction_all} demonstrates that our predictor matches the true number of variants much more closely than the parametric Bayesian method and smoothed Good-Toulmin estimator. In this case without sequencing error, our method has roughly the same performance as the harmonic jackknife and linear programming.

 To more directly compare performance of the methods across all 33 cancer types, we calculate the error of each method across all types and all folds within each type; see \Cref{fig:prediction_same_conditions_boxplot}. More precisely, for each of the $33$ cancer types and for each of the $20$ folds we compute the trimmed absolute percentage prediction error incurred by the five methods at the largest possible extrapolation value. See \Cref{eq:loss_ape} in \Cref{sec:app_TCGA_msk}. In \Cref{fig:prediction_same_conditions_boxplot}, we summarize these $33*20 = 660$ error values for each method in a boxplot. Lower errors are better. We find that our Bayesian nonparametric methods performs similarly to the linear programming method and to the harmonic jackknife. Our method outperforms the smoothed Good-Toulmin estimator and the parametric Bayesian approach. In \Cref{sec:app_TCGA_msk}, we also follow \citet{chakraborty2019somatic} and run an experiment with an entirely separate pilot and follow-up study. In terms of comparison among estimators, these additional experiments lead to similar conclusions. 

 We performed additional experiments to better understand how our method compares to existing methods. In \Cref{sec:app-exp_additional_bnp} we run both the Bayesian parametric approach and our method on data simulated (a) under the parametric Bayesian model used by \citet{ionita2009estimating} and (b) under our own 3-parameter beta process model. We find that the approach of \citet{ionita2009estimating} works well with data simulated from their model but poorly with the three-parameter beta process data. Our results suggest that the parametric Bayesian method \citep{ionita2009estimating} struggles with data exhibiting power laws, which we expect in real life.

 While the method of \citet{zou2016quantifying} performs well in our experiments above, we found serious numerical issues in other cases. In particular, \citet{zou2016quantifying} exploits a linear programming approach to estimate rare variant proportions; the authors approximate proportions of common variants with the corresponding empirical frequencies. The authors define a variant as ``rare'' if it has frequency less than $\kappa/100$, for a user-defined threshold $\kappa \in (0,100)$, interpreted as a percent. In practice, we found that the output of the algorithm is very sensitive to the choice of $\kappa$; see \Cref{sec:app-exp_additional_zou}. The authors suggest $\kappa = 1$ as a default setting, but we observed numerical instability and poor predictive performance for this value. This observation holds especially when the pilot size $N$ is small, which we believe to be a particular case of interest in designing experiments for further data collection (i.e., for the follow-up study). For instance, we expect the small-$N$ case to arise in the study of non-model organisms \citep{russell2017non}. In \Cref{fig:prediction_all}, we chose $\kappa = 20$, which led to convergence of the optimization algorithm in all cases. We explore other values of $\kappa$ in \Cref{sec:app-exp_additional_zou}. Beyond these issues, we sometimes found that the method of \citet{zou2016quantifying} failed to converge. While the convergence issue did not arise for our experiments in this section, it did arise for another analysis of the TCGA data; see \Cref{sec:app_tcga_genes}.

 The Good-Toulmin method used in \citet{chakraborty2019somatic} performs poorly in our experiments above (\Cref{fig:prediction_all} and \Cref{fig:prediction_same_conditions_boxplot}), as well as in our further real-data experiments in \Cref{sec:app-exp_gnomAD}. However, we find that this method seems competitive with the best alternative on other cancer genomics data; see \Cref{sec:app_tcga_genes}. Further understanding of the variable performance of this estimator would be an important first step before any potential future use. By contrast, we find that jackknife \citet{gravel2014predicting}, with an appropriate hyperparameter calibration, performs well across all of our experiments when conditions are kept constant between the pilot and follow-up.

 \subsection{Prediction under different experimental conditions} \label{sec:exp_different}

 \begin{figure}[t!]
       \centering \includegraphics[width=\textwidth,height=\textheight,keepaspectratio]{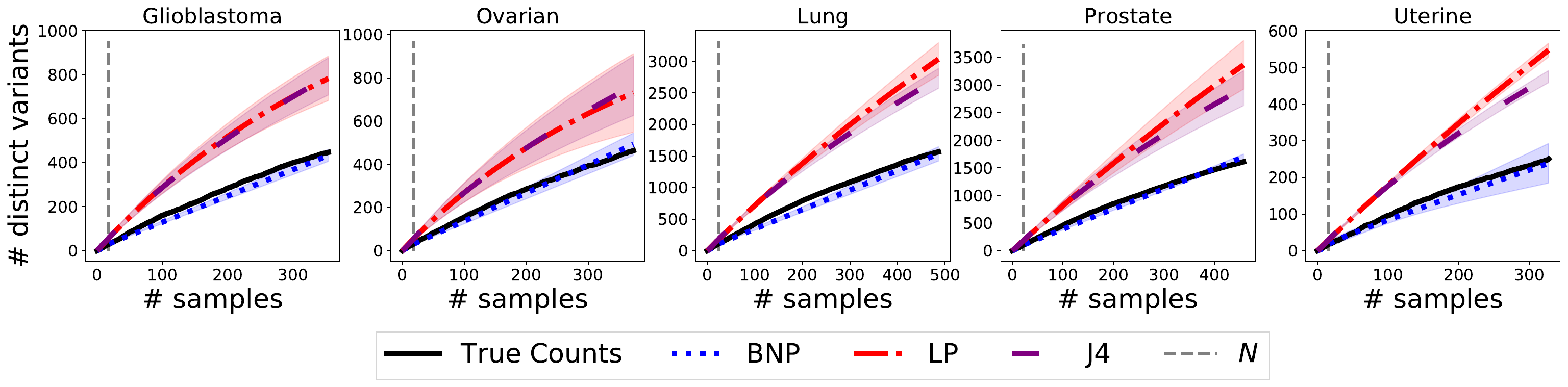}
 \caption{TCGA dataset: predicting the number of new variants. The true number of distinct variants (black) is compared with:
 our method (blue, BNP); linear program (red, LP); 4th order jackknife (purple, J4). Shaded bands represent one standard deviation.}
 \label{fig:pred_exp_d}
 \end{figure}

 We now turn to the case where there may be sequencing errors in the pilot study, in the follow-up study, or both. Moreover, the sequencing quality may differ between the pilot study and the follow-up study. To the best of our knowledge, no existing methods work in this case. We believe that the parametric Bayesian method of \citet{ionita2009estimating}, the smoothed Good-Toulmin estimator of \citet{chakraborty2019somatic}, and the linear programming method of \citet{zou2016quantifying} could all be adapted to take sequencing errors into account. However, we have seen that the parametric Bayesian and Good-Toulmin methods already struggle when there are no sequencing errors. And the linear programming method suffers from numerical instability when the training sample size is small (the case of most interest). While the harmonic jackknife of \citep{gravel2014predicting} performs well when there are no sequencing errors, we do not think it will be straightforward to adapt it to the case where sequencing quality may change between the pilot and follow-up. 

 In \Cref{fig:pred_exp_d} we see that there is indeed a noticeable difference in the number of observed variants when the experimental conditions change between the pilot and follow-up. In particular, we consider a pilot sequencing quality $\seqinit=100$ and a follow-up sequencing quality $\seqfollowup~=~50$. We use a fixed threshold $T=45$, a realistic coverage value in human genomic experiments \citep{karczewski2019variation}, and the same five cancer types as in \Cref{fig:prediction_all}. To represent this change between studies, we use the TCGA data as in \Cref{sec:exp_noerror} but apply additional thinning to simulate imperfect observation due to sequencing depth; see \Cref{sec:app_TCGA_details} for additional details. Since the harmonic jackknife is not able to use information about the changing sequencing depth, we expect our Bayesian nonparametric method to deliver superior predictive performance when sequencing quality changes. This behavior is exactly what we see in \Cref{fig:pred_exp_d}.

 \subsection{Designing experiments to maximize the number of observed variants} \label{sec:exp_opt_d}

 \begin{figure}[t!]
       \centering \includegraphics[width=\textwidth,height=\textheight,keepaspectratio]{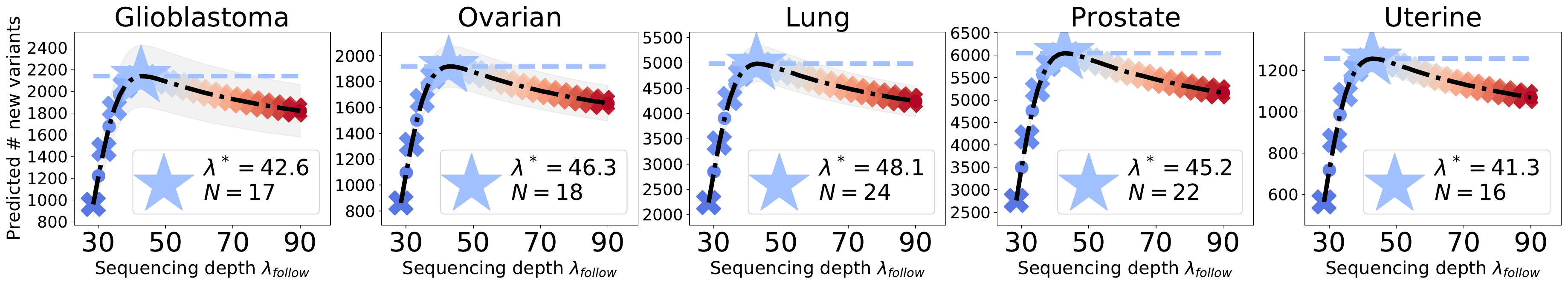}
     \caption{TCGA dataset: designing an experiment to maximize the number of new variants in a follow-up study.}
 \label{fig:opt_design}
 \end{figure}

 We show that our method can be used for experimental design in practice. Our procedure consists of three steps. (i) Given the pilot data and sequencing quality $\seqinit$, we minimize \Cref{eq:cost} to estimate the parameters $\conc, \discount, \mass$. (ii) We consider a range of values of the follow-up sequencing quality $\seqfollowup$; for each $\seqfollowup$, we choose the maximum follow-up size $M$ that stays within our budget $\budget$, and we use the learned values of $\conc, \discount, \mass$ to predict the number of new variants in each case. (iii) We choose the settings of $\seqfollowup$ and $M$ that maximize the number of new variants. We illustrate this procedure in \Cref{fig:opt_design}. In our experiments, we set the cost function $\costfcn(M,\seqfollowup) = M \log (\seqfollowup)$ as in \citet{ionita2010optimal}. For every cancer type, we retain $5\%$ of the observations as a pilot study. We set a budget $\budget$ such that we can sample at full depth, i.e. coverage of 100x, only half of the total remaining $95\%$ of the observations. We set variant calling rule threshold $\threshold = 45$, error $p_{\text{err}} = 0.01$, and $\seqinit = 100$. We run the procedure over all folds. We plot the predicted number of observed variants in the follow-up by maximizing $M$ under the budget $\budget$ and quality $\seqfollowup$, $\hat{P}_N^{(M, \seqfollowup)}$; the shaded region in \Cref{fig:opt_design} illustrates one standard deviation. We see a trade-off in quality and quantity. Namely, maximizing quantity $M$ leads to very small values of $\seqfollowup$ to maintain the budget $\budget$. With sufficiently low quality, though, fewer variants are discovered. Conversely, when $\seqfollowup$ is set very high, we require a very small $M$ to maintain the budget $\budget$, and not many variants are discovered. Intermediate values of $\seqfollowup$ and $M$ serve to maximize the number of variants discovered under a fixed budget.
 \section{Discussion} \label{sec:discussion}

 We have presented a Bayesian nonparametric method for predicting the number of variants in a follow-up study using information from a pilot study. Our method works even when the follow-up study has different experimental conditions from the pilot study, and can be used for optimal design of the follow-up study.

 Though our experiments here focus on rare variants from bulk studies in human genetics, we briefly describe further potential applications to emphasize the generality of our framework.  First, in microbiome research, there is an increasing interest in (i) devising low-cost pipelines for efficient sequencing \citep{rajan2019phylogenetic, sanders2019optimizing}, as well as (ii) defining best-practice protocols for data collection processes \citep{hillmann2018evaluating, bharti2021current}. Indeed, scientists have already expressed an interest in optimal allocation of a budget given information from a pilot experiment \citep{zaheer2018impact, pereira2019impact}. Second, in single-cell sequencing, scientists are interested in reliably estimating important gene properties. In this case, there exists a vast and growing literature that highlights the importance of establishing the optimal trade-off between the quality (sequencing depth) of the experiment, and the number of cells to be sequenced. See, for example,  \citet{bacher2016design, li2018accurate, zhang2020determining}. 
 Third, it is becoming common practice to use modern, non-invasive approaches for surveying wildlife populations, such as camera-traps \citep{tarugara2019cost, welbourne2020camera}. Accurate estimation of the living population and timely adoption of preventive measures are crucial for the survival of endangered species \citep{johansson2020identification}. But conservation groups often face a limited budget. These groups might benefit from trading off equipment density and quality.

 While the present paper has focused on data that can be represented as collections of binary features (e.g.\ variants and non-variants), our method may be extended to the case in which the observations are vectors of counts, as well as the case in which there exist multiple categories for each feature (e.g.\ different types of variants). In particular, by means of the Bayesian nonparametric conjugacy framework of \citet{james2017bayesian, broderick2018posteriors}, we may extend our method to use a categorical (or multinomial) likelihood process with a conjugate Bayesian nonparametric prior for the now-multiple frequencies per variant location. Our Bayesian nonparametric method may also be easily extended to accommodate multiple different pilot studies. For the latter extension, we would still generate variant proportions according to the three-parameter beta process; we would then generate variants in each pilot study according to different damped Bernoulli processes. The ultimate effect would be to introduce more distinct, but workable, Bernoulli terms in \Cref{eq:noise_new}. Moreover, in this work we have focused on threshold variant calling rules, which are a simplification of state-of-the art variant callers \citep{xu2018review}. Extending our framework to encompass more realistic variant calling rules is an interesting future research direction. An important practical challenge in this case will be even specifying a formula or series of formulas to describe how popular variant callers work.


 \section*{Acknowledgments} \label{sec:acknowledgments}

 The authors are grateful to the Editor, the Associate Editor, and two anonymous Referees for their comments, corrections, and suggestions, which have greatly improved the paper. The results shown in the present paper are in whole or part based upon data generated by the TCGA Research Network. The authors thank Boyu Ren, Joshua Schraiber, Michael Hoffman, and Brian Trippe for useful discussions and comments. The authors are also grateful to Boyu Ren for help working with the gnomAD dataset. Federico Camerlenghi and Stefano Favaro received funding from the European Research Council under the European Union's Horizon 2020 research and innovation programme under grant agreement No 817257, and the Italian Ministry of Education, University and Research, ``Dipartimenti di Eccellenza" grant 2018--2022. Lorenzo Masoero and Tamara Broderick were supported in part by the DARPA I2O LwLL program, the CSAIL-MSR Trustworthy AI Initiative, an NSF CAREER Award, a Sloan Research Fellowship, and ONR.

\bibliographystyle{abbrvnat}
\bibliography{ref.bib}

\appendix

\section*{Appendix}
This document contains the supplementary material for ``\emph{More for less: predicting and maximizing genomic variant discovery via Bayesian nonparametrics}''. In \Cref{sec:app-proofs} we present the proofs of the results presented in \Cref{sec:model}. We next provide detail about the competing methods we considered. In \Cref{sec:app-ionita} the Bayesian parametric estimator of  \citet{ionita2009estimating}, in \Cref{sec:app-zou} the linear program proposed by \citet{zou2016quantifying}, in \Cref{sec:app_jack} the Jackknife estimator used in \citet{gravel2014predicting}, and in \Cref{sec:app-gt} the Good-Toulmin estimator used in \citet{chakraborty2019somatic}. We conclude providing additional experimental results. In \Cref{sec:app_TCGA_msk} we present additional detail about the data used in \citet{chakraborty2019somatic}, and considered in the analysis in the main text. In \Cref{sec:app-exp_gnomAD} we report results for the gnomAD project \citep{karczewski2019variation}, an extension of the datasets previously considered in \citet{gravel2014predicting, zou2016quantifying}. We conclude with extensive experiments on simulated data in \Cref{sec:app-exp_additional}.


\renewcommand{\theequation}{A.\arabic{equation}}
\setcounter{equation}{0} 

\section{Additional results and proofs} \label{sec:app-proofs}

\subsection*{Proof of \Cref{prop:naive_new}}

\begin{proof} 

By construction, the variant frequencies $\{\afreq_j\}$ are formed from a Poisson point process with rate measure 
 \begin{align} 
 	\ratemeas(\d \afreq)
 		 &= \mass \frac{\Gamma(1+\conc)}{\Gamma(1-\discount)\Gamma(\conc+\discount)} \afreq^{-1-\discount} (1-\afreq)^{\conc+\discount-1} \bm{1}_{[0,1]}(\afreq)\d \afreq .
			\label{eq:ratemeas}
 \end{align}

Recall that a variant with frequency $\afreq_j$ appears in organism $n$ with Bernoulli probability $\theta_j$, independently across $n$. Therefore, the collection of variant frequencies whose corresponding variants have not yet appeared after $N$ organisms comes from a thinned Poisson point process relative to the original Poisson point process generating the $\{\afreq_j\}$; the thinned process has rate measure $\ratemeas(\d \afreq) \cdot \Bern(0 | \afreq)^N$ and is independent of the collection of frequencies that did appear in the first $N$ organisms.
Similarly, the collection of variant frequencies corresponding to variants that did not appear in the first $N$ organisms but then did appear in the first follow-up organism comes from a thinned Poisson point process with rate measure $\ratemeas(\d \afreq) \cdot \Bern(0 | \afreq)^N \cdot \Bern(1 | \afreq)$ and is independent of the collection of frequencies that did not appear in the first $N+1$ organisms. Recursively, for $m \ge 1$, the collection of variant frequencies corresponding to variants that did not appear in the first $N+m-1$ organisms but then did appear in the $m$th follow-up organism comes from a thinned Poisson point process with rate measure
\begin{align*}
	\lefteqn{\ratemeas(\d \afreq) \Bern(0 | \afreq)^{N+m-1} \Bern(1 | \afreq)} \\
		&= \mass \frac{\Gamma(1+\conc)}{\Gamma(1-\discount)\Gamma(\conc+\discount)} \afreq^{-1-\discount + 1} (1-\afreq)^{ \conc+\discount-1 + N+m-1} \bm{1}_{[0,1]}(\afreq)\d \afreq \\
		&= \mass \frac{\Gamma(1+\conc)}{\Gamma(1-\discount)\Gamma(\conc+\discount)}
			\cdot \frac{\Gamma(1 - \discount) \Gamma(\conc+\discount-1 + N+m)}{ \Gamma(\conc + N+m)} \\
		& \quad {} \cdot \Betadist(\afreq \mid 1 - \discount, \conc+\discount-1 + N+m) \d \afreq \\
		&= \mass \frac{ (\conc+\discount)_{(N+m-1) \uparrow} }{ (1+\conc)_{(N+m-1)\uparrow } } \Betadist(\afreq \mid 1 - \discount, \conc+\discount-1 + N+m) \d \afreq. 
\end{align*}
Finally, we observe that the number of points in a Poisson point process is Poisson distributed with mean equal to the integral of its rate measure. Each of these Poisson point processes is independent, and the sum of independent Poissons is Poisson with mean equal to the sum of the means. So, since $\countnew{N}{M}$ is the sum of points in these $M$ Poisson point processes with $m \in [M]$, 
we have $\countnew{N}{M}$ is Poisson with mean
\begin{align*}
	\sum_{m=1}^{M}& \int_0^1 \mass \frac{ (\conc+\discount)_{(N+m-1) \uparrow} }{ (1+\conc)_{(N+m-1) \uparrow} } \Betadist( \afreq | 1 - \discount, \conc+\discount-1 + N+m) \d \afreq \\
	 	&= \sum_{m=1}^{M} \mass \frac{ (\conc+\discount)_{(N+m-1) \uparrow} }{ (1+\conc)_{(N+m-1) \uparrow} },
\end{align*}
as was to be shown.

\end{proof}

\subsection*{Proof of \Cref{prop:asymptotic}}

In the following we make use of the $O$ notation, indeed we will write $f(x)=O(g(x))$ to mean that the ratio $|f(x)/g(x)|$
is a bounded function of the variable $x$; we also write $f(x)=o(g(x))$ as $x\to x_0$ (little $o$ notation) to mean that $\lim_{x \to x_0}f(x)/g(x)=0$.
A preliminary result is needed.
\begin{lemma} \label{lem:asymptotic_Gamma_sum}
For any $c >0$, $N \geq 1$ and $\discount \in (0,1)$ we have that
\begin{equation} \label{eq:asymptotic_Gamma_sum}
\frac{1}{M^\discount}\sum_{m=1}^M \frac{\Gamma (\conc+N+m-1+\discount)}{\Gamma (\conc+N+m)} = \frac{1}{\discount}+ O (M^{-\discount})
\end{equation}
is  satisfied as $M$ grows to $+\infty$.
\end{lemma}

\begin{proof}
As in the proof of \citet[Lemma 2]{berti2015central}, we know that for any $x >0$, 
\[
    \frac{\Gamma (x+\discount)}{\Gamma (x+1)} = x^{\discount-1} (1+g (x)),
\]
where $g : (0, +\infty) \to R$ is such that $\sup_{x \geq 0} |g (x)x|< +\infty $. Putting $x= \conc+N+m-1$, where $m \geq 1$ while $\conc$ and $N$ are fixed constants, this very last condition on $g$ is equivalent to
$\sup_{y \geq \conc+N-1} |g (y)y|< +\infty$. Hence there exists $K >0$ such that $ |g(y)| \leq K/y$ for any $y \geq \conc+N-1$. With this in mind we focus on the left hand side of \Cref{eq:asymptotic_Gamma_sum} in the statement of \Cref{lem:asymptotic_Gamma_sum}:
\begin{align} 
    \frac{1}{M^\discount}&\sum_{m=1}^M \frac{\Gamma (\conc+N+m-1+\discount)}{\Gamma (\conc+N+m)}  \\
    &= \frac{1}{M^\discount}\sum_{m=1}^M \nonumber
 (\conc+N+m-1)^{\discount-1} (1+g (\conc+N+m-1)) \label{eq:sum_2_terms_1}\\
 &= \frac{1}{M^\discount}\sum_{m=1}^M (\conc+N+m-1)^{\discount-1} \\ &+ \frac{1}{M^\discount}\sum_{m=1}^M (\conc+N+m-1)^{\discount-1} g (\conc+N+m-1). \label{eq:sum_2_terms_2}
\end{align}
As for the sum of \Cref{eq:sum_2_terms_1} note that the following inequalities hold true
\begin{align} \label{eq:lim1}
    \frac{(\conc+N+M)^\discount -(\conc+N)^\discount}{\discount M^\discount} &=	\int_1^{M+1} \frac{(\conc+N+m-1)^{\discount-1}}{M^\discount} \d m   \nonumber \\
	&\le  \sum_{m=1}^M \frac{(\conc+N+m-1)^{\discount-1}}{M^\discount}\nonumber \\
	& \le 	\int_0^{M} \frac{(\conc+N+m-1)^{\discount-1}}{M^\discount} \d m \nonumber\\
	&=  \frac{(\conc+N+M-1)^\discount-(\conc+N-1)^\discount}{\discount M^\discount},
\end{align}
where we have used the fact that $(\conc+N+m-1)^{\discount-1}$ is decreasing in $m$, and used the corresponding integrals to bound the sum. We can use an asymptotic expansion of the upper and the lower bound in \Cref{eq:lim1} to get
\begin{align*}
    \frac{1}{\discount} &\left( \frac{\discount (\conc+N)}{M} + o \Big( \frac{1}{M} \Big)  -\frac{(\conc+N)^\discount}{M^\discount} \right) \leq  \sum_{m=1}^M \frac{(\conc+N+m-1)^{\discount-1}}{M^\discount} -\frac{1}{\discount}\\ 
    & \leq \frac{1}{\discount} \left( \frac{\discount (\conc+N-1)}{M} + o \Big( \frac{1}{M} \Big)  -\frac{(\conc+N-1)^\discount}{M^\discount} \right),
\end{align*}
which entails  that
\begin{align}
\label{eq_tesi1}
\sum_{m=1}^M \frac{(\conc+N+m-1)^{\discount-1}}{M^\discount} =\frac{1}{\discount}+ O \Big(\frac{1}{M^\discount}\Big).
\end{align}
As for \Cref{eq:sum_2_terms_2}, we exploit the properties of $g$ to get
\begin{align*}
	\Big| \frac{1}{M^\discount} & \sum_{m=1}^M (\conc+N+m-1)^{\discount-1} g (\conc+N+m-1) \Big|  \\
	&\le \frac{K}{M^\discount}\sum_{m=1}^M  (\conc+N+m-1)^{\discount-2} \\
	&\leq \frac{K}{M^\discount} \int_{\conc+N-1}^{\conc+N+M-1} \frac{1}{x^{2-\discount}} \d x \\
	&=\frac{K}{M^\discount (1-\discount)}  \left\{ \frac{1}{(\conc+N-1)^{1-\discount}}  - \frac{1}{(\conc+N+M-1)^{1-\discount}} \right\}
\end{align*}
The last inequality implies that
\begin{align}
\label{eq:lim2}
\Big|\frac{1}{M^\discount}\sum_{m=1}^M (\conc+N+m-1)^{\discount-1} g (\conc+N+m-1)\Big|=  O \Big( \frac{1}{M^\discount}\Big). 
\end{align}
Putting \Cref{eq:lim1} and \Cref{eq:lim2} in \Cref{eq:sum_2_terms_1} and \Cref{eq:sum_2_terms_2} the thesis follows. 
\end{proof}

If $X$ is a real valued random element, we denote by $\Phi_X (t) = E (e^{itX})$ its characteristic function, where $i$ is the imaginary unit. We also assume that all the random variables are defined on a probability space $(\Omega, \mathcal{A}, \PP)$, and we denote by ${\rm pr}_N$ the probability ${\rm pr}$ given $X_{1:N}$;  $E_N$ and $\text{var}_N$ will stand for the expected valued and the variance given $X_{1:N}$, respectively.
\begin{proof}[of \Cref{prop:asymptotic}]
We start by showing the strong law of large numbers of \Cref{eq:slln} in the main text.
From \Cref{lem:asymptotic_Gamma_sum} we deduce that

\begin{align}
\label{eq:EK_m_inf}
    \frac{E_N ( \countnew{N}{M} )}{M^\discount} = \frac{\mass}{M^\discount} \frac{ \Gamma (\conc+1)}{\Gamma (\conc+\discount)} 
\sum_{m=1}^{M} \frac{\Gamma (\conc+\discount +N+ m-1)}{\Gamma (\conc+N+m)} \to \frac{\mass \Gamma (\conc+1)}{\discount \Gamma (\conc+\discount)}
\end{align}
as $M \to +\infty$. We observe that $\countnew{N}{M} = H_N^{(1)}+ \cdots +H_N^{(M)}$, where 
$H_N^{(m)}$ are independent Poisson random variables with mean 
\[
    \frac{\mass (\conc+\discount)_{(N+m-1) \uparrow}}{(\conc+1)_{(N+m-1) \uparrow}},
\]
for $m=1, \ldots , M$ and $M$ is arbitrary large. $H_N^{(m)}$ is the number of new variants that have been observed in the $(N+m)$-th individual, conditionally on the first $N$ individuals.
As a consequence we may write
\begin{align*}
    \frac{\countnew{N}{M} - E_N \left(\countnew{N}{M}\right)}{M^\discount} = \frac{H_N^{(1)}-E_N (H_N^{(1)}) +\cdots  + H_N^{(M)} - E_N (H_N^{(M)})}{M^\discount} .
\end{align*}
The Kronecker's lemma \citep[Lemma IV.3.2]{Shiryaev:1995} implies  that 
\begin{align*}
\lim_{M\to +\infty }\frac{\countnew{N}{M}- E_N \left(\countnew{N}{M}\right)}{M^\discount} = 0 \quad {\rm pr}_N-\text{almost surely},
\end{align*}
provided that the following condition is satisfied 
\begin{align}
\label{eq:cond_slln}
\sum_{m=1}^{+\infty } \frac{\text{var}_N (H_N^{(m)} )}{m^{2 \discount}} < +\infty.
\end{align}
This may be easily verified as follows:
\begin{align*}
    \sum_{m=1}^{+\infty } \frac{\text{var}_N (H_N^{(m)} )}{m^{2 \discount}} & = \sum_{m=1}^{+\infty } \frac{\mass}{m^{2\discount}} \frac{( \conc +\discount)_{N+m-1 \uparrow}}{(\conc +1)_{N+m-1 \uparrow}} \\
    & = \mass  \frac{\Gamma (\conc +1)}{\Gamma ( \conc +\discount)} \sum_{m=1}^{+\infty } \left\{ \frac{1}{m^{2\discount}}
\frac{\Gamma (\conc+\discount +N+m-1)}{\Gamma (\conc+N+m)} \right\} < +\infty.
\end{align*}
The series turns out to be convergent because the following asymptotic relation holds true:
\[
    \frac{\Gamma (\conc+\discount +N+m-1)}{\Gamma (\conc+N+m)} \sim \frac{\mass}{m^{1+\discount}} \frac{ \Gamma (\conc +1)}{\Gamma (\conc+\discount)}.
\]
Hence \Cref{eq:cond_slln} is satisfied, so we conclude that
\begin{align*}
    \lim_{M\to +\infty }\frac{\countnew{N}{M}- E_N \left(\countnew{N}{M}\right)}{M^\discount} = 0 \quad \text{almost surely,}
\end{align*}
which is equivalent to the thesis thanks to \Cref{eq:EK_m_inf}.

\par We now prove the central limit theorem stated in \Cref{eq:clt} in the main text. We prove the result using the convergence of characteristic functions.
We use the fact that the posterior distribution of $\countnew{N}{M}$ is Poisson to evaluate the characteristic function \textit{a posteriori}: for convenience, let $\centeredcountnew{N}{M} := \sqrt{M^{\sigma}} \left( \frac{\countnew{N}{M}}{M^{\sigma}} -\xi \right)$, where we recall that $\xi$ is defined as
\begin{align*}
	\xi:=\frac{\alpha}{\sigma} \frac{\Gamma(c+1)}{\Gamma(c+\sigma)}.
\end{align*}
Then, 
\begin{align*}
    &\Phi_{\centeredcountnew{N}{M} \mid X_{1:N} } (t) 
      =E_N \left[ \exp \left\{ it  \centeredcountnew{N}{M} \right\}\right]\\
    & = \exp \left\{  -it\xi \sqrt{M^{\sigma}} +  \mass  (e^{it/\sqrt{M^{\sigma}}} -1) \sum_{m=1}^M \frac{(c+\discount)_{N+m-1 \uparrow}}{(\conc+1)_{N+m-1 \uparrow}}  \right\}\\
    & = \exp \left\{  -it\xi \sqrt{M^{\sigma}} +   (e^{it/\sqrt{M^{\sigma}}} -1) \frac{\mass \Gamma (\conc+1)}{\Gamma (c+\discount)}
    \sum_{m=1}^M \frac{\Gamma (c+N+m-1+\discount)}{\Gamma(c+N+m)} \right\}.
    \end{align*}
    
We now use  \Cref{lem:asymptotic_Gamma_sum} and the asymptotic expansion of the exponential function to get

    \begin{align*}
    &\Phi_{\centeredcountnew{N}{M} \big\vert X_{1:N}} (t)  = \\
    &=\exp \left\{  -it\xi \sqrt{M^{\sigma}} +  \frac{\mass\Gamma (\conc+1)}{\Gamma (c+\discount)}\left(\frac{it}{\sqrt{M^{\sigma}}} -\frac{t^2}{2M^{\sigma}} +O(M^{-\frac{3}{2}\sigma})\right)  \left( \frac{M^{\sigma}}{\discount}+ O (1)\right)\right\}\\
    &  = \exp \left\{  -it\xi \sqrt{M^{\sigma}} +  \frac{\mass}{\discount} \frac{\Gamma (\conc+1)}{\Gamma (c+\discount)} \left(it\sqrt{M^{\sigma}} -\frac{t^2}{2} +O\left(\sqrt{M^{\sigma}}\right)\right) \right\}\\
    &= \exp \left\{  -it\xi \sqrt{M^{\sigma}} +   \xi(it\sqrt{M^{\sigma}} -\frac{t^2}{2} +O\left(\sqrt{M^{\sigma}}\right)) \right\} \\
    &= \exp \left\{     -\frac{\xi t^2}{2} +O\left(\sqrt{M^{\sigma}}\right) \right\},
\end{align*}
where in the penultimate line we substituted 
\[
    \xi = \frac{\mass}{\discount} \frac{\Gamma (\conc+1)}{\Gamma (c+\discount)}.
\]
Therefore, as $M$ grows to infinity, we get
\begin{align*}
    \Phi_{\sqrt{M^{\sigma}} \left( \frac{\countnew{N}{M}}{M^{\sigma}} -\xi \right) \big \vert X_{1:N}} (t) 
    \longrightarrow  \exp \left\{  - \frac{\xi t^2}{2} \right\},
\end{align*}
and the thesis follows.
\end{proof}

\subsection*{Proof of \Cref{prop:noise_new}}

\begin{proof}
To see the almost sure finiteness of the Poisson parameter and hence of the random variables $\countnew{N}{M}$ and of $ \countnewtimes{N}{M}{r}$, and $\countnewtimes{N}{M}{\le R}$, note that the parameter constraints for the three-parameter beta process are specifically constructed so that $\afreq \nu(\d \afreq)$ is a proper beta distribution; see the end of \cref{sec:model} and \citet{james2017bayesian, broderick2018posteriors}.
The $\afreq$ factor will arise from $\Bern(1 \mid \probcallfollowup \afreq)$.

The exact form of the Poisson parameter in \Cref{eq:noise_new} 
arises by following the same thinning argument as in the proof of \Cref{prop:naive_new}. To see the beta representation,
\begin{align*}
	\lefteqn{ \Bern(1 \mid \probcallfollowup \afreq) \Bern(0 \mid \probcallfollowup)^{m-1} \Bern(0 \mid \probcallinit)^{N} \nu(\d \afreq) } \\
		&= \mass \frac{\Gamma(1+\conc)}{\Gamma(1-\discount)\Gamma(\conc+\discount)} \afreq^{-1-\discount} ( (1-\afreq)^{ \conc+\discount-1} \times \\
		&\times (\probcallfollowup \afreq) (1- \probcallfollowup \afreq)^{m-1} (1- \probcallinit \afreq)^{N} \bm{1}_{[0,1]}(\afreq) \d \afreq \\
		&= \mass \frac{\Gamma(1+\conc)}{\Gamma(1-\discount)\Gamma(\conc+\discount)} 
			\probcallfollowup
			(1- \probcallfollowup \afreq)^{m-1} (1- \probcallinit \afreq)^{N}
			\cdot \frac{\Gamma(1 - \discount) \Gamma(\conc+\discount)}{ \Gamma(\conc + 1)} \\
		& \quad {} \cdot \Betadist(\afreq \mid 1 - \discount, \conc+\discount) \d \afreq \\
		&= \mass \probcallfollowup \Betadist(\afreq \mid 1 - \discount, \conc+\discount) \d \afreq.  
\end{align*}
The exact form of the Poisson parameter $\gamma_r$ in \Cref{eq:noisenewtimes}   
arises by following the same thinning argument as in the proof of \Cref{prop:naive_new_num_times}. To see the beta representation, 
\begin{align*}
	\lefteqn{ \Bern(1 \mid \probcallfollowup \afreq)^{r} \Bern(0 \mid \probcallfollowup \afreq)^{M-r} \Bern(0 \mid \probcallinit \afreq)^{N} \nu(\d \afreq) } \\
		&= \mass \frac{\Gamma(1+\conc)}{\Gamma(1-\discount)\Gamma(\conc+\discount)} \afreq^{-1-\discount} ( (1-\afreq)^{ \conc+\discount-1} \times  \\
		&\times (\probcallfollowup \afreq)^{r} (1- \probcallfollowup \afreq)^{M-r} (1- \probcallinit \afreq)^{N} \bm{1}_{[0,1]}(\afreq) \d \afreq \\
		&= \mass \frac{\Gamma(1+\conc)}{\Gamma(1-\discount)\Gamma(\conc+\discount)}
			\probcallfollowup^{r}
			(1- \probcallfollowup \afreq)^{M-r} (1- \probcallinit \afreq)^{N}
			\cdot \frac{\Gamma(r - \discount) \Gamma(\conc+\discount)}{ \Gamma(\conc + r)} \\
		& \quad {} \cdot \Betadist(\afreq \mid r - \discount, \conc+\discount) \d \afreq \\
		&= \mass \probcallfollowup^{r} \frac{ (1+c)_{(r-1)\uparrow} }{ (1-\discount)_{(r-1)\uparrow} } \cdot \Betadist(\afreq \mid r - \discount, \conc+\discount) \d \afreq.
\end{align*}
\end{proof}

\subsection{Number of new rare variants}

\begin{proposition}[Number of new rare variants] \label{prop:naive_new_num_times}
	Assume the model in \Cref{prop:naive_new}.
	Let $\countnewtimes{N}{M}{r}$ represent the number of new variants that occur $r$ times in a follow-up sample of size $M$ after a preliminary study of size $N$. I.e., we count the variants that do not occur in the preliminary $N$ samples but then occur $r$ times in the follow-up $M$ samples. Let $\countnewtimes{N}{M}{\le R}$ similarly represent the number of new variants that occur at most $R$ times. Here $r, R \in [M]$. Define
	\begin{displaymath} \label{eq:countnewtimes}
		\countnewtimes{N}{M}{r} := \sum_{j=1}^{\infty} \bm{1}\left( \sum_{n=1}^{N} \acount_{n,j} = 0\right) \bm{1}\left( \sum_{m=1}^{M} \acount_{N+m,j} = r \right).
	\end{displaymath}
	Then
	\begin{align} \label{eq:naive_new_num_times}
		\countnewtimes{N}{M}{r} \mid \mcount_{1:N} \sim \Pois ( \lambda_r ),
		\end{align}
	for $\lambda_r := \mass \binom{M}{r} \frac{(1-\discount)_{(r-1) \uparrow} (\conc+\discount)_{(N+M-r) \uparrow} }{(\conc+1)_{(N+M-1)\uparrow} }$. Moreover, for
	\begin{align*}
			\countnewtimes{N}{M}{\le R} := \sum_{r=1}^{R} \countnewtimes{N}{M}{r},
	\end{align*}
	it holds 
	\begin{align*}
		\countnewtimes{N}{M}{\le R} \mid \mcount_{1:N} &\sim \Pois \left(\sum_{r=1}^{R} \lambda_r \right).
	\end{align*}
	
	\end{proposition}

Just like for $\countnew{N}{M}$, the Bayesian nonparametric predictors of $\countnewtimes{N}{M}{r}$ and $\countnewtimes{N}{M}{\le R}$ correspond to the expected values of $\countnewtimes{N}{M}{r} \mid \mcount_{1:N} $ and $\countnewtimes{N}{M}{\le R} \mid \mcount_{1:N} $, respectively, i.e. the parameters of the posterior predictive Poisson distributions displayed in \eqref{eq:naive_new_num_times}. Similarly to Proposition \eqref{prop:asymptotic}, the large $M$ asymptotic behaviour of $\countnewtimes{N}{M}{r} \mid \mcount_{1:N} $ and $\countnewtimes{N}{M}{\le R} \mid \mcount_{1:N} $ display very specific power law behavior almost surely.

\begin{proof} 

Analogous to the proof of \Cref{prop:naive_new}, we consider the Poisson point process $\{\afreq_j\}$ and thin it to those frequencies corresponding to variants chosen no times in the preliminary $N$ samples and chosen exactly $r$ times out of the follow-up $M$ samples. The probability of being chosen to be thinned, then, is $\Bern(0 \mid \afreq)^N \cdot \binom{M}{r} \cdot \Bern(0 \mid \afreq)^{M-r} \cdot \Bern(1 \mid \afreq)^{r}$. The thinned process therefore has rate measure
\begin{align*}
	\lefteqn{ \ratemeas(\d \afreq) \cdot \Bern(0 \mid \afreq)^N \cdot \binom{M}{r} \cdot \Bern(0 \mid \afreq)^{M-r} \cdot \Bern(1 \mid \afreq)^{r} } \\
		&= \mass \frac{\Gamma(1+\conc)}{\Gamma(1-\discount)\Gamma(\conc+\discount)}
			\afreq^{-1-\discount +r} (1-\afreq)^{\conc+\discount-1 + N + M - r} \d \afreq \\
		&= \mass \frac{\Gamma(1+\conc)}{\Gamma(1-\discount)\Gamma(\conc+\discount)}
			\frac{ \Gamma(r - \discount) \Gamma(\conc + \discount + N + M - r) }{ \Gamma(\conc + N + M) } \times \\
			&\quad\quad \times\Betadist(r - \discount, \conc + \discount + N + M - r) \d \afreq \\
		&= \mass \frac{ (1-\discount)_{(r-1)\uparrow} (\conc + \discount)_{(N+M-r)\uparrow} }{ (1+c)_{(N+M-1)\uparrow} } \Betadist(r - \discount, \conc + \discount + N + M - r) \d \afreq.
\end{align*}
Since $\countnewtimes{N}{M}{r}$ counts the thinned atoms, it has Poisson distribution with mean equal to the integral of the rate measure, i.e.\ mean equal to 
\begin{align} \label{eq:naive_mean_num_times}
    \mass \frac{ (1-\discount)_{(r-1)\uparrow} (\conc + \discount)_{(N+M-r)\uparrow} }{ (1+\conc)_{(N+M-1)\uparrow} },
\end{align}
as was to be shown.

The distribution of $\countnewtimes{N}{M}{\le R}$ follows immediately from the observation that the $\countnewtimes{N}{M}{r}$ are independent Poisson random variables, where the independence is inherited from the independent thinned Poisson point processes.

\end{proof}
\subsection{Asymptotics for number of new rare variants}

\begin{proposition}[Asymptotics for number of new rare variants] \label{prop:asymptotic_numtimes}
Under the setting of \Cref{prop:naive_new_num_times},
\begin{align}
	\frac{\countnewtimes{N}{M}{r}}{M^{\discount}} \bigg\vert \mcount_{1:N}
		\,\stackrel{\text{a.s.}}{\longrightarrow}\,  \xi_r
			\; \textrm{ as } M \rightarrow \infty,
			\intertext{and}
			\frac{\countnewtimes{N}{M}{\le R}}{M^{\discount}} \bigg\vert \mcount_{1:N}
		\,\stackrel{\text{a.s.}}{\longrightarrow}\, \sum_{r=1}^{R} \xi_r
			\; \textrm{ as } M \rightarrow \infty,
			\label{eq:slln_r} 
\end{align}
where $\xi_r :=  \frac{\mass}{r!} (1-\discount)_{(r-1) \uparrow} \frac{\Gamma(\conc+1)}{\Gamma(\conc+\discount)}$, and the limits hold true almost surely, conditionally  given $\mcount_{1:N}$. In addition,
\begin{align}
	\sqrt{M^{\discount}} \left( \frac{\countnewtimes{N}{M}{r}}{M^{\discount}} - \xi_r \right) \bigg\vert \mcount_{1:N}
		&\,\stackrel{\text{d}}{\longrightarrow}\,  \mathcal{N}\left(0, \xi_r \right) \quad \textrm{ as } M \rightarrow \infty
			\label{eq:clt_r} 
			\intertext{and}
	\sqrt{M^{\discount}} \left( \frac{\countnewtimes{N}{M}{\le R}}{M^{\discount}} - \sum_{r=1}^{R} \xi_r \right) \bigg\vert \mcount_{1:N}
		&\,\stackrel{\text{d}}{\longrightarrow}\,  \mathcal{N}\left(0, \sum_{r=1}^{R} \xi_r \right) \quad \textrm{ as } M \rightarrow \infty.\label{eq:clt_r_sum}
\end{align}
The limits displayed in \Cref{eq:clt_r} and \Cref{eq:clt_r_sum} hold in distribution  conditionally given $\mcount_{1:N}$.
\end{proposition} 
\begin{proof}

We start by proving \Cref{eq:slln_r} in the main text,  but in order to do this we have to define some other statistics:
\begin{equation}
\label{eq:cumulative}
\countnew{N}{M,\le R} := \sum_{r=1}^R  \countnew{N}{M,r} \quad \text{and} \quad
\countnew{N}{M,\ge R} := \sum_{r=R}^M  \countnew{N}{M,r}
\end{equation}
which have to be respectively interpreted as the number of new genomic variants observed at most $R$ times and the number of new genomic variants observed at least $R$ times. 
Our strategy is the following: we prove that $\countnew{N}{M, \ge R}/M^\discount$ converges almost surely to a constant and then we use the relation 
\begin{align} \label{eq:M_difference}
 \countnew{N}{M,R} =\countnew{N}{M, \ge R}-\countnew{N}{M, \ge R+1}
\end{align}
to prove the convergence of $\countnew{N}{M,R}$.\\
We evaluate the fist moment of $\countnew{N}{M,\ge R}/M^\discount$ a posteriori: for notation purpose, let $\centeredcountnew{N}{M}$
\begin{align} 
	\frac{1}{M^{\discount} }E_N \left( \countnew{N}{M, \ge R} \right) &= \frac{E_N [\countnew{N}{M}]}{M^{\discount}}  -\sum_{r=1}^{R-1} \frac{E_N \left( \countnew{N}{M,r} \right)}{M^{\discount}} \nonumber\\
	&\longrightarrow \frac{\mass \Gamma (\conc+1)}{\discount \Gamma (c+\discount)}  -
\mass \frac{\Gamma (\conc+1)}{\Gamma (c+\discount)} \sum_{r=1}^{R-1} \frac{(1-\discount)_{R-1 \uparrow}}{R!},\label{eq:M_bar}
\end{align}
as $M \to \infty$, where we have used \Cref{eq:E_M} and \Cref{eq:EK_m_inf}. It then follows that 
\[
    E_N [\countnew{N}{M,r}] \asymp  c_1 M^{\discount}
\]
for some positive constant $c_1 >0$. 
Besides for the variance of $\countnew{N}{M, \ge R}$ we get
\begin{align*}
	&\text{var}_N \left(\countnew{N}{M, \ge R}\right)  = \text{var}_N\left(  \countnew{N}{M} - \countnew{N}{M,\le R-1} \right) \\
	&= E_N \left\{ \countnew{N}{M}  - \sum_{r=1}^{R-1} \left(\countnew{N}{M,r}\right)  - E_N \left(\countnew{N}{M}\right) + \sum_{r=1}^{R-1} E_N(\countnew{N}{M,r})  \right\}^2\\
	& \leq E_N \left\{ \left|\countnew{N}{M} - E_N \left(\countnew{N}{M}\right) \right|    +   \sum_{r=1}^{R-1} \left| \countnew{N}{M,r}   - E_N \left( \countnew{N}{M, r}\right) \right| \right\}^2\\
	&\leq R \cdot E_N \left\{ \left|\countnew{N}{M} - E_N \left(\countnew{N}{M}\right) \right|^2  + \sum_{r=1}^{R-1} \left| \countnew{N}{M}  - E_N\left( \countnew{N}{M,r}\right)\right|^2    \right\}
\end{align*}
where the last inequality follows by a simple application of the discrete version of the H\"older's inequality. Then, using also the fact that  we get that $\countnew{N}{M}$ and  $\countnew{N}{M}$ are Poisson random variable a posteriori, we obtain:
\begin{align*}
	\text{var}_N (\countnew{N}{M, \ge R}) & \leq r  \left( \text{var}_N (\countnew{N}{M})  + \sum_{r=1}^{R-1} \text{var}_N (\countnew{N}{M,r}) \right) \\ 
	& = R   \left\{ E_N [ \countnew{N}{M} ]  + \sum_{r=1}^{R-1} E_N \left( \countnew{N}{M,r} \right) \right\} \asymp c_2 M^{\discount}
\end{align*}
where $c_2>0 $ is a positive constant. From all the previous considerations and by an application of the Markov inequality we obtain that for any $\varepsilon > 0$
\begin{align}
\PP_N \left\{ \Big| \frac{\countnew{N}{M, \ge R}}{E_n \left(\countnew{N}{M, \ge R}\right)} -1 \Big| \geq \varepsilon \right\}&
	\leq \frac{\text{var}_N\left(\countnew{N}{M, \ge R}\right)}{\varepsilon^2  \left\{ E_N \left(\countnew{N}{M, \ge R}\right) \right\}^2}\nonumber\\
	&\lesssim \frac{c_2 M^{\discount}}{\varepsilon^2 (c_1 M^{\discount})^2} \asymp \frac{1}{M^{\discount}}
\end{align}
hence we can conclude that the ratio  
\[
    \frac{\countnew{N}{M, \ge R}}{E_N \left(\countnew{N}{M, \ge R}\right)}
\]
converges
in probability to $1$. Besides if we choose the subsequence $M_k:=k^{2/\discount}$, as $k=1, 2, \ldots $, an application of the first Borel-Cantelli lemma leads us to state that the ratio converges to $1$ almost surely. Since 
$\countnew{N}{M, \ge R}$ is an increasing process as $M$ increases, for any $M$ in the interval 
$\{\lfloor m_k \rfloor, \ldots,  \lfloor m_{k+1} \rfloor \}$ we have that
\[
\countnew{N}{\lfloor m_k \rfloor, \ge R} \le \countnew{N}{M, \ge R} \le \countnew{N}{\lfloor m_{k+1} \rfloor, \ge R}
\]
where $\lfloor x \rfloor$ denotes the integer part of $x$. Hence we also have that
\[
\frac{\countnew{N}{\lfloor m_k \rfloor, \ge R}}{E_N\left(\countnew{N}{\lfloor m_{k+1} \rfloor, \ge R}\right)} 
 \leq \frac{\countnew{N}{M, \ge R}}{E_N\left(\countnew{N}{M, \ge R}\right)}
 \leq  \frac{\countnew{N}{\lfloor m_{k+1} \rfloor, \ge R}}{E_N\left(\countnew{N}{\lfloor m_{k} \rfloor, \ge R}\right)}.
\]
Leveraging the fact that the lower and upper bound of the central term converge to $1$ as $k\to\infty$, 
\[
    \frac{\countnew{N}{M, \ge R}}{E_N[\countnew{N}{M, \ge R} ]} \to 1,
\]
in an almost sure sense as $M \to +\infty$. In other words, using \Cref{eq:M_bar}, 
we have just proved that
\begin{equation}
\label{eq:M_bar_slln}
\frac{\countnew{N}{M, \ge R}}{M^{\discount}} \longrightarrow
\frac{\mass \Gamma (\conc+1)}{\discount \Gamma (c+\discount)} -
\mass \frac{\Gamma (\conc+1)}{\Gamma (c+\discount)}\sum_{r=1}^{R-1} \left(\frac{(1-\discount)_{r-1 \uparrow}}{r!} \right),
\end{equation}
$\PP_N$-almost-surely as $M \to \infty$.

The thesis now follows by \Cref{eq:M_bar_slln} and the \Cref{eq:M_difference}, indeed:
\begin{equation*}
 \frac{\countnew{N}{M,R}}{M^{\discount}} =\frac{\countnew{N}{M, \ge R}}{M^{\discount}}-
 \frac{\countnew{N}{M, \le R+1}}{M^{\discount}}  \to \mass \frac{\Gamma (\conc+1)}{\Gamma (c+\discount)} \frac{ (1-\discount)_{r-1 \uparrow} }{r! },
\end{equation*}
$P_N$-almost surely as $M$ grows.

\par To prove \Cref{eq:clt_r} in the main text, one has to prove that the characteristic functions converge, more precisely
\[
    \Phi_{\sqrt{M^{\sigma}}\left( \frac{\countnew{N}{M,R}}{M^{\sigma}} - \xi_R \right) \bigg\vert X_{1:N} } (t)
    \longrightarrow \exp\left(-\frac{t^2\xi_R}{2} \right) \quad \text{for any } t \in R,
\]
as $M$ goes to infinity.

First of all observe that the (posterior) expectation of 
$\countnew{N}{M,R}$ is such that
\begin{align}
	E_N \left(\countnew{N}{M,R}\right) &= \mass \binom{M}{R} \frac{(1-\discount)_{R-1 \uparrow} (c+\discount)_{N+M-R}}{(\conc+1)_{N+M-1 \uparrow}} \nonumber\\
	&= \left(\frac{\Gamma (M +1)}{\Gamma (M-R+1)}  \frac{\Gamma (\conc+\discount+N+M-R)}{\Gamma (\conc+N+M)} \right) \times \\
	&\times  \mass \frac{(1-\discount)_{R-1 \uparrow}}{\Gamma(R+1)} \frac{\Gamma (\conc+1)}{\Gamma (\conc+\discount)}\nonumber \\
	& = M^{\discount} \left( 1+o (M^{-1}) \right) \mass \frac{(1-\discount)_{R-1 \uparrow}}{\Gamma(R+1)} \frac{\Gamma (\conc+1)}{ \Gamma (c+\discount)}, \label{eq:E_M}
\end{align}
where we have used the asymptotic expansion of ratios of gamma functions given by \citet{tricomi1951asymptotic}.
 Let $\centeredcountnew{N}{M,R}:= \sqrt{M^{\sigma}}\left( \frac{\countnew{N}{M,R}}{M^{\sigma}} - \xi_R \right)$. Using the expansion given in \Cref{eq:E_M} it is easy to see that
\begin{align*}
    &\Phi_{\centeredcountnew{N}{M,R} \bigg\vert X_{1:N} } (t)= \\
    &=\exp \left\{  -it \sqrt{M^{\sigma}} \xi_R + M^{\sigma} \xi_R \left( 1+o (M^{-1}) \right) \left( \frac{it}{\sqrt{M^{\sigma}}} +
    \frac{t^2}{2M^{\sigma}} + o (M^{-1+\sigma}) \right) \right\}\\
    & =  \exp  \left\{ -\frac{t^2 \xi_R}{2} + o(1)\right\},
\end{align*}
therefore the thesis follows.
\end{proof}

\subsection{Number of new rare variants in presence of noise}

Similarly, for $\countnewtimes{N}{M}{r}$ and $\countnewtimes{N}{M}{\le R}$ defined in \Cref{eq:countnewtimes} with $r, R \in [M]$, we have that these quantities are almost surely finite with respective distributions 
\begin{align} \label{eq:noisenewtimes}
	\countnewtimes{N}{M}{r} \mid \mcount_{1:N} \sim \Pois \left( \gamma_r \right), \quad
	\countnewtimes{N}{M}{\le R} \mid \mcount_{1:N} 	\sim \Pois \left( \sum_{r=1}^{R} \gamma_r \right),
\end{align}
where
\begin{align*}
\gamma_r
		:= \binom{M}{r} \int_{\afreq=0}^{1} &\Bern(1 \mid \probcallfollowup \afreq)^{r} \Bern(0 \mid \probcallfollowup \afreq)^{M-r}\\ 
		&\Bern(0 \mid \probcallinit \afreq)^{N} \nu(\d \afreq) \\
		= \mass \binom{M}{r} \probcallfollowup^{r} &\frac{ (1+c)_{(r-1)\uparrow} }{ (1-\discount)_{(r-1)\uparrow} } E_B\{ (1-\probcallfollowup B)^{M-r} (1-\probcallinit B)^N \},
\end{align*}
for $B \sim \Betadist(\afreq \mid r - \discount, \conc+\discount)$.

\subsection{Proof of equality in \Cref{eq:poisson_thinning}} \label{sec:poisson_thinning}
To show that $\phi(\lambda,T,p_{err})$ is the right tail of a Poisson distribution, we recur to the Binomial thinning of Poisson random variables.

\begin{proposition}[Binomial thinning of Poisson random variables] \label{proof:poisson_thinning}
    Let $N \sim \Pois(\lambda)$. Let $X_1,\ldots,X_n \sim \Bern(q)$ independently and identically distributed. Then, $S_N:=X_1 + \ldots X_N \sim \Pois(\lambda q)$, and 
    \begin{align}
        \Pr(X_N \ge T) = \sum_{t\ge T}\frac{e^{-\lambda q}\lambda^q}{t!}.
    \end{align}
\end{proposition}
\begin{proof}
    Let $S_n \sim \Binom(n, q)$ be a binomial random variable with success probability $q$ and $n$ draws. The moment generating function of the binomial distribution is
    \[
        E\left[t^{S_n}\right] = (1-q+qt)^n,
    \]
    while the moment generating function of the Poisson distribution with parameter $\lambda>0$ is
    \[
        E\left[t^N\right] = \sum_{k\ge 0} \frac{(\lambda t)^k}{k!}e^{-\lambda} = \exp\{\lambda t -\lambda\}.
    \]
    Hence,
    \begin{align}
        E\left[t^{S_N}\right] &= \sum_{n\ge 0} \frac{E\left[t^{S_n}\right]\lambda^n e^{-\lambda}}{n!} = \sum_{n\ge 0} \frac{(1-q+qt)^n\lambda^n e^{-\lambda}}{n!} \nonumber\\
        &= \exp\left\{\lambda(1-q+qt)-\lambda \right\}= \exp\left\{  \lambda q t -\lambda q \right\},
    \end{align}
    which implies $S_N \sim \Pois(\lambda q)$.
\end{proof}

In light of this proposition, the equality in \Cref{eq:poisson_thinning} follows. 


\renewcommand{\theequation}{B.\arabic{equation}}
\setcounter{equation}{0} 

\section{Bayesian prediction with the Beta-Bernoulli product model} \label{sec:app-ionita}

We here review the approach proposed by \citet{ionita2009estimating}. The authors consider the same problem of genomic variation described in \Cref{sec:model}. 
\par \citet{ionita2009estimating} assume that there exists a finite, albeit unknown, number of loci at which genomic variation can be observed. We denote such quantity with the letter $K$. Given a pilot study $X = X_{1:N}$ with $J$ distinct variants, we can obtain the site-frequency-spectrum (or fingerprint) of the sample, 
\begin{align} \label{eq:sfs}
    \bm{f}_{N} = [f_{N,1} \ldots, f_{N,J}] \quad\text{with} \quad f_{N,j} = \sum_{\ell=1}^J\bm{1}\left(\sum_{n= 1}^N x_{n,\ell} = j\right),
\end{align}
so that $f_{N,1}$ counts the number of variants observed only once among the $N$ samples, $f_{N,2}$ the number of variants observed in exactly two samples etc.
The input data  $X_{1:N}$ is here viewed as a binary matrix, $X_{1:N} \in \{0,1\}^{N \times J}$, in which all positions at which variation is not observed are discarded, and the order of the columns is immaterial. This binary matrix is modeled via a parametric beta-Bernoulli model: the authors assume that there exists a fixed, unknown number $K < \infty$ of loci at which variation can be observed. For each $j\in [K]$, they assume that there exists an associated variant, labelled by index $j$, displayed by any observation (row) with probability $\theta_j \in [0,1]$. The frequencies $\theta_j$, $j = 1, \ldots, K$ are distributed according to a beta distribution with parameters $a$, $b$, i.e.
\begin{align*}
	\bm{\theta} = \begin{bmatrix}\theta_1 & \dots &\theta_{K} \end{bmatrix},  \quad \text{ with }\quad \theta_j \sim \text{Beta}(a,b) \; \forall j,
\end{align*}
independently and identically distributed. Conditionally on $\bm{\theta}$, 
\begin{align*}
	X_n = \begin{bmatrix} x_{n,1} & \dots & x_{n,K} \end{bmatrix}, \quad \text{ with }\quad x_{n,j} {\sim} \text{Bernoulli}(\theta_j).
\end{align*}
Therefore, the columns of the matrix $X_{1:N}$ are independently and identically distributed, while the rows are made of independent, but not identically distributed entries. Under this  model, the number of counts of each variant is binomially distributed, conditionally on the latent frequency of such variant, i.e.
\begin{align*}
    z_{N,j} \mid \theta_j:=\sum_{i=1}^n x_{n,j} \mid \theta_j \sim \text{Binomial}(N,\theta_j).
\end{align*} 
Recalling that $f_{N,j} = \sum_{\ell=1}^{J} \bm{1}(z_{N,\ell} = j)$ is the number of variants which appear exactly $j$ times among the first $N$ samples, and letting $g(x;a,b)$ be the  density function of a beta random variable with parameters $a,b$ evaluated at $x$, 
\[
    g(x;a,b) = \frac{x^{a-1}(1-x)^{b-1}}{\bm{B}(a,b)} \bm{1}_{[0,1]}(x),
\]
with $B(a,b) = \int_0^1 x^{a-1}(1-x)^{b-1} \d x = \Gamma(a)\Gamma(b)/\Gamma(a+b)$, then 
the probability that exactly $j$ of the $N$ individuals show variation at a given site is given by
\begin{align}
   p_{N,j} &=  \int_0^1 \binom{N}{j} \theta^k (1-\theta)^{N-j}g(\theta;a,b)\d \theta \nonumber \\
   &=\binom{N}{j} \int_0^1 \frac{\theta^{N+a-1}(1-\theta)^{N-j+b-1}}{\bm{B}(a,b)} \d\theta = \binom{N}{j} \frac{(a)_{j\uparrow} (b)_{N-j \uparrow}}{(a+b)_{N\uparrow}}. \label{eq:beta_probs}
\end{align}
Because we can't observe more than $N$ variants in $N$ trials, and since we don't know anything about variants which are yet to be observed, probabilities of \Cref{eq:beta_probs} are then normalized as follows:
\begin{align*}
    \lambda_{N,j} = \frac{p_{N,j}}{\sum_{\ell=1}^N p_{N,\ell}} = \frac{\binom{N}{j}(a)_{j\uparrow} (b)_{N-j \uparrow}}{\sum_{\ell=1}^N\binom{N}{\ell}(a)_{\ell\uparrow}(b)_{N-\ell\uparrow}},
\end{align*}
for all $\ell = 1,\ldots, N$.
It follows that the log likelihood for the observed data $X_{1:N}$ is given by
\begin{align*}
	\ell_{a,b}^{\text{BBPM}}(X_{1:N}) =  \log \left(\prod_{j=1}^N \lambda_{N,j}^{f_{N,j}} \right) = \sum_{j=1}^N f_{N,j} \log(\lambda_{N, j}).
\end{align*}

Notice that the expected number of variants appearing exactly once in a sample of $N$ observations can be computed in closed form,

\begin{align*}
    \eta_{N,1} &:= E[f_{N,1}] = E\left\{\sum_{j=1}^K \binom{N}{1} \ind\left(\sum_{n=1}^N x_{n,j} = 1 \right) \right\} \\
    &= KN \int_{[0,1]} \frac{\theta^{a-1}(1-\theta)^{b-1}}{\bm{B}(a,b)} \theta (1-\theta)^{N-1} \d \theta \\
    &= KN \frac{\bm{B}(a+1, N+b-1)}{\bm{B}(a,b)} = \frac{aKN}{N+b-1}\frac{\bm{B}(a, N+b)}{\bm{B}(a,b)},
\end{align*}
where we used independence of the variants, linearity of the expectation operator and the properties of the beta function.
Letting $M = tN$ be the number of additional samples to be observed, we can compute the expected number of hitherto unseen variants, to be observed in additional $M$ samples after $N$ samples have been collected as

\begin{align}
	\Delta_N(M) &= E\left\{\sum_{j=1}^{K} \bm{1} \left(\sum_{m = 1}^{M} x_{m,j}>0\right)\ind\left(\sum_{n=1}^N x_{n,j} = 0\right) \right\} \nonumber\\
			&= \frac{K}{\bm{B}(a,b)} \int_{[0,1]} (1 - (1-\theta)^{(t+1)N}) - (1-(1-\theta)^N) \theta^{a-1}(1-\theta)^{b-1} \d\theta \nonumber\\
			&= \frac{K}{\bm{B}(a,b)} \int_{[0,1]} \left\{ (1-\theta)^N - (1-\theta)^{(t+1)N}\right\} \theta^{a-1}(1-\theta)^{b-1} \d\theta \nonumber \\
			&= K \frac{\bm{B}(a,N+b)}{\bm{B}(a,b)}  -  K\frac{\bm{B}(a,N(t+1)+b)}{\bm{B}(a,b)}   \nonumber
\end{align}
Now, noting that 
\begin{align*}
    K\frac{\bm{B}(a,N+b)}{\bm{B}(a,b)} =  \frac{\eta_{N,1}}{a} \frac{N+b-1}{N}
\end{align*}
and
\begin{align*}
    K\frac{\bm{B}(a,N(t+1) + b)}{\bm{B}(a,b)} = \frac{\eta_{N,1}}{a} \frac{N+b-1}{N} \frac{\bm{B}(a,N(t+1)+b)}{\bm{B}(a,N+b)},
\end{align*}

it follows that

\begin{align} \label{est:ionita}
	\Delta_N(M) &= \frac{\eta_{N,1}}{a}\frac{N+b-1}{N}\left\{1-\frac{\bm{B}(a, N(t+1)+b)}{\bm{B}(a,N+b)} \right\}.
\end{align}
Importantly, $\Delta_N(M)$ depends on $K$ only via $\eta_{N,1}$. To use the estimator $\Delta_N(M)$, \citet{ionita2009estimating} substitute $\eta_{N,1}$ with its empirical counterpart $f_{N,1}$, the number of variants which have been observed exactly once in the sample $X_{1:N}$. The parameters $a,b$ are found via maximization of the log-likelihood of the model, 
\begin{align*}
	\{a^*, b^*\} = \argmax_{a>0, b>0} \left\{\ell_{a,b}^{\text{BBPM}}(X_{1:N})\right\}
\end{align*}
\begin{remark}
    The estimator obtained in \Cref{est:ionita} crucially relies on the empirical frequency of variants observed once among the first $N$ draws, $f_{N,1}$. For example, if a dataset had $f_{N,1} = 0$, $\Delta_N(M) = 0$ for every $M>0$.
\end{remark}

\renewcommand{\theequation}{C.\arabic{equation}}
\setcounter{equation}{0} 

\section{Linear program to estimate the frequencies of frequencies} \label{sec:app-zou}

\citet{zou2016quantifying} assume, in the same way as \citet{ionita2009estimating}, that there exists a finite albeit unknown number of sites at which variants can be observed. They formalize the problem of hitherto unseen variants prediction as that of recovering the distribution of frequencies of all the genetic variants in the population, including those variants which have not yet been observed.
\par They assume that each possible variant in a sample is independent of the other variants, and that the $j$-th variant appears with a given probability $\theta_j$ conditionally independently and identically distributed across all the individuals observed - i.e.\ the $\theta_j$ are parameters of independent Bernoulli random variables $x_{n,j}$ for all $n \geq 1$ and $j$. Therefore the pilot study $X_{1:N}$ is modeled by a collection of independent Bernoulli random variables, which are also identically distributed  along each column, and the sum $z_{N,j} \mid \theta_j := \sum_{n=1}^N x_{n,j} \mid \theta_j \sim \Binom(N, \theta_j)$. From the frequencies $z_{N,1},\dots,z_{N,J}$ of the $J$ variants observed among the first $N$ samples, it is possible to compute the fingerprint of the sample, $\bm{f}_N$. Given the fingerprint, the goal is to recover the population's histogram, which is a map quantifying, for every $\theta \in[0,1]$, the number of variants such that $\theta_j = \theta$. Formally, learn a map $h$ from the distribution of frequencies to integers
\begin{align}\label{eq:hist}
	h:(0,1] \to \mathbb{N} \cup \{0\}
\end{align}
Because for $N$ large enough the empirical frequencies associated to common variants should be well approximated by their empirical counterpart, \citet{zou2016quantifying} only consider the problem of estimating the histogram from the truncated fingerprint $\bm{f}_N^{(\kappa)} = \{ f_{N,j} \;:\; j/N \le 100 \times \kappa\}$. In their analysis, the authors only consider $\kappa = 1$, i.e.\ they consider ``common'' variants all those variants that appear in more than $1\%$ of the sample elements. Moreover, rather than learning a continuous function as described by \Cref{eq:hist}, they solve a discretized version of the problem. They fix a discretization factor $\delta \ge 1$, and then set up a linear program in which the goal is to correctly estimate  the population histogram associated to the frequencies in the set $\SS ~=~ \left\{\frac{1}{1000 N}, \delta \frac{1}{1000 N}, \dots, \delta^i \frac{1}{1000 N}, \dots, \kappa \right\}$. The value $\delta$, given $\kappa$, determines how many frequencies are going to be estimated in $(0,\kappa]$: the lower $\delta$, the finer the discretization. The authors suggest using $\delta = 1.05$. In our experiments, we set $\delta = 1.01$, for which we find the method to produce better results, at the cost of a small additional computational effort. Finally, the problem of recovering the histogram is solved through the following optimization:
\begin{align*}
	\min_{h(\theta), \theta \in \SS}\sum_{j: j \leq  N\kappa} \frac{1}{1+f_{N,j}} \left|f_{N,j} - \sum_{\theta \in \SS} h(\theta) \Binom(N,\theta,j) \right| \\
	\intertext{subject to} h(\theta)\geq 0, \sum_{\theta \in \SS} h(\theta) \leq K, \sum_{\theta \in \SS} \theta \cdot h(\theta ) + \sum_{j: j > N\kappa}^{J} \frac{j}{N} f_{N,j} = \frac{J}{N},
\end{align*}
where $K$ is an upper bound on the total number of variants, and  $\Binom(N,\theta,j)$ is the probability that a Binomial draw with bias $\theta$ and $N$ rounds is equal to $j$.
\par Given the histogram $\hat{h}$ which solves the linear program above, one can obtain an estimate of the number of unique variants at any sample size $M$ using
\begin{align*} 
	V(\hat{h}, M) = \sum_{\theta :\hat{h}(\theta )>0} \hat{h}(\theta) (1-(1-\theta)^M).
\end{align*}
Following \citet{zou2016quantifying}, we refer to this estimator as the ``unseenEST'' estimator.

\renewcommand{\theequation}{D.\arabic{equation}}
\setcounter{equation}{0} 

\section{Jackknife estimators} \label{sec:app_jack}

Jackknife estimators for predicting the number of hitherto unseen species were first introduced by in the capture-recapture literature by \citet{burnham1978estimation}. 

Given $ X_{1:N} \stackrel{iid}{\sim} F(\psi)$ for some distribution $F$ and some parameter $\psi$, let $ \hpsn = \hpsn(X_{1:N})$ be an estimator of $\psi$ with the property that
\begin{align} 
    E [\hpsn] = \psi + \frac{a_1}{N} + \frac{a_2}{N^2} + \dots,  \label{eq:jack_bias}
\end{align}
for fixed constants $a_1, a_2, \dots$. Without loss of generality assume $\hpsn$ to be symmetric in its inputs $X_{1:N}$, and denote with $\II \subset [N]$ a subset of given size $p$, let $\hps_{N - p, \II}$ be the estimate obtained by dropping the observations whose indices are in $\II$. Similarly, let
\begin{align} \label{eq:jack_corr}
    \hps_{N}^{(p)} = \binom{N}{p}^{-1}\sum_{\II:|\II|=p} \hps_{N-p,\II}
\end{align}
The idea of the Jackknife estimator is that, if the assumption of  \Cref{eq:jack_bias} holds, we can improve over $\hpsn$ by using a correction originating from \Cref{eq:jack_corr}. The $p$-th order Jackknife estimator is defined as
\begin{align}
    \hps_N^{J_p} = \frac{1}{p}\sum_{\ell=0}^{p} \left\{(-1)^{\ell} \binom{p}{\ell} (N-\ell)^{p} \hps_N^{(\ell)} \right\}. \label{eq:jack}
\end{align}
Under the assumption of \Cref{eq:jack_bias}, the estimator of \Cref{eq:jack} has bias approaching zero polynomially fast in the correction order, $\text{Bias}(\hps_N^{J_p}) \sim N^{-p-1}$.

\subsection{An estimator for the population size} 

\citet{burnham1978estimation} introduced a nonparametric procedure to estimate the total number of animals present in a closed population when capture-recapture data is available. Assume that there is a fixed, but unknown number $K$ of total species.  Over the course of $N$ repeated observational experiments, $J \le K$ distinct species are observed. 

Let $X_{1:N}$ be the collection of available data, in which $X_n = [x_{n,1},\ldots,x_{n,J}]$, with $x_{n,j} = 1$ if species $j$ has been observed on the $n$-th experiment, and $0$ otherwise. Moreover, assume that each species $j \in [K]$ has a fixed, but unknown probability $\theta_j \in (0,1]$ of being observed.

Notice that while \citet{burnham1978estimation} developed the estimator having in mind a fixed and finite population of animals, we can also think of each sample $X_n$ as a genomic sequence characterized by the presence or absence of genetic variants at different sites.

 The nonparametric MLE for the total support size $K$ is given by $\hat{K}^{\MLE}(X_{1:N}) = \hat{K}^{\MLE}_N = J$. Clearly $ J\leq K$, therefore $J$ is a biased estimate for $K$. If one assumes, in a similar spirit to \Cref{eq:jack_bias}, that
\begin{align}
    E[\hat{K}^{\MLE}_N] = K +\frac{a_1}{N} + \frac{a_2}{N^2} + \dots,
\end{align}
then one could use the jackknife estimator of \Cref{eq:jack} to estimate $K$. This requires computing $\hps_N^{(\ell)}$ for $\ell = 1,\dots, p$, which are linear functions of the observed fingerprint $\bm{f}_N$. 

\textbf{The case $p = 1$}: We outline the approach for $p=1$. Let $q_{N,n}$ be the number of animals which have been observed only once out of the $N$ trials, exactly on the $n$-th,
\[
    q_{N,n} = \sum_{j\geq 1} \bm{1}(x_{n,j}=1) \bm{1}\left(\sum_{n'\neq n} x_{n',j} = 0 \right)
\]
Then, because $q_{N,1} + \ldots + q_{N,N} = f_{N,1}$ by construction,
\begin{align}
    \hat{K}_{N}^{(1, \setminus n)} = J - q_{N,n} \quad \text{and} \quad \hat{K}_N^{(1)} = \frac{1}{N}\sum_{n=1}^N \hat{K}_{N}^{(1, \setminus n)} = J - \frac{f_{N,1}}{N}.
\end{align}
Therefore, the order 1 jackknife estimator for the total population size is obtained by plugging in $\hat{\psi}_N^{(0)} = J$ and $\hat{\psi}_N^{(1)}=J - \frac{f_{N,1}}{N}$ in \Cref{eq:jack}: 
\begin{align} \label{jack_1}
    \hat{K}_N^{J_1} = J + \frac{N-1}{N} f_{N,1}
\end{align}

\textbf{The case for general $p$}: For any $p\leq N$, it always holds that 
\begin{align}
    \hat{K}_N^{(p)} = J - \binom{N}{p}^{-1} \sum_{\ell=1}^p \binom{N-\ell}{p-\ell} f_{N,\ell}
\end{align}
This formula allows to obtain the general Jackknife estimator of order $p$,
which is a linear function of the observed number of species $J$ and correction terms which depend on the fingerprint $\bm{f}_N$,
\[
    \hat{K}_{N}^{J_p} = \sum_{\ell=1}^p a_{N,\ell}^{(p)} f_{N,\ell}.
\]

\subsection{Estimators for the number of hitherto unseen genomic variants}

Taking inspiration from the approach of \citet{burnham1978estimation}, \citet{gravel2011demographic} and \citet{gravel2014predicting} developed Jackknife estimators for the number of hitherto genomic variants which are going to be observed in $M$ additional samples given $N$ initial ones. Let $V(N)$ denote the total number of variants observed in $N$ samples, and let $\Delta(N+M, N):= H_{N+M-1} - H_{N-1} = \sum_{\ell = N}^{M+N-1} 1/\ell$, where 
\begin{align*}
    H_N ~=~ 1 + 1/2 + \ldots + 1/N
\end{align*}
is the $N$-th harmonic number. To derive their estimators, the authors use the assumption that for a given order $p\geq 1$ the total number of variants present in $N+M$ samples can be estimated as follows:
\begin{align} \label{eq:jack_gravel}
    \hat{V}_N^{(M)} =  V(N) + \sum_{\ell  =1}^p a_{N,\ell}^{(p)} \Delta(N+M, N)^\ell, 
\end{align}
where $\bm{a}^{(p)}_N = [a_{N,1}^{(p)},\ldots,a_{N,p}^{(p)}]$ are constants which depend on the initial sample size $N$, on the order $p$ and on the fingerprint of the sample $\bm{f}_N$. This assumption is exact in the case of a constant size and neutrally evolving population (\citet{gravel2011demographic}). For a given order $p$ the unknown coefficients are obtained by solving the following system of equations:
\begin{align}
    \hat{V}_N^{(M)} = \hat{V}_{N-1}^{(M)} = \ldots = \hat{V}_{N-p}^{(M)}.
\end{align}
Equating $\hat{V}_N^{(M)}$ to $\hat{V}_{N-j}^{(M)}$ using \Cref{eq:jack_gravel} for $j = 1,\ldots, p$, we obtain a system of $p-1$ equations of the form
\begin{align}
    V(N) - V(N-j) = \sum_{\ell=1}^p a_{N,\ell}^{(p)} (\Delta(N+M, N-j)^\ell - \Delta(N+M, N)^\ell).
\end{align}
Using the additional equality
\begin{align}
    V(N) - V(N-\ell) = \sum_{j = 1}^{\ell} \frac{\binom{\ell}{j}}{\binom{N}{j}} f_{N,j}.
\end{align}
we can solve for $a_{N,\ell}^{(p)}$ and express these in terms of $N, \Delta(N+M, N)$ and the fingerprint $\bm{f}_N$, and the final estimator is a linear function of the fingerprint $\bm{f}_N$.

\subsection{Choice of the jackknife order} \label{sec:choice_jack}

As pointed out in \citet{burnham1978estimation}, the optimal order $p$ of the jackknife estimator heavily depends on the data under consideration. It is therefore desirable to obtain a procedure which uses the data to guide the choice of such order. \citet{burnham1978estimation} phrase this decision problem as a sequential hypothesis test, in which one keeps increasing the order of the jackknife until the data suggests that the drop in bias obtained by increase the jackknife order is exceeded by the gain in variance. Precisely, for $p=1,2,\ldots$ one sequentially performs the following test:
\begin{align}
    H_{0,p} :\; E(\hat{K}_N^{J_{p+1}} - \hat{K}_N^{J_p}) = 0 \quad\text{versus}\quad H_{a,p}:\;E(\hat{K}_N^{J_{p+1}} - \hat{K}_N^{J_p}) \neq 0.
\end{align}
If $H_{0,p}$ is rejected, this has to be interpreted as evidence that the bias reduction provided by the $p+1$-th order (with respect to the $p$-th) is larger than the associated increase in variance, and $p+1$-th order should be preferred to the $p$-th order \citep{burnham1978estimation}. The first order $p$ for which the test fails to reject the null hypothesis is chosen as the jackknife order.

The test relies on the following observation: the difference between two jackknife estimators of different orders $p+1$ and $p$ is given by
\begin{align}
    \hat{K}_N^{J_{p+1}} - \hat{K}_N^{J_p} = \sum_{\ell=1}^{p+1} \tilde{a}_{N,\ell}^{(p+1,p)} f_{N,p},
\end{align}
again a linear combination of the fingerprint. Because the conditional distribution of the fingerprint is independent of $K$ given $J$, the minimum variance estimator of the conditional variance is given by
\begin{align}
    \text{est var}(\hat{K}_N^{J_{p+1}} - \hat{K}_N^{J_p} \mid J) = \frac{J}{J-1}\left\{ \sum_{\ell=1}^p (\tilde{a}_{N,\ell}^{(p+1,p)})^2 f_{N,\ell} \frac{(\hat{K}_N^{J_{p+1}} - \hat{K}_N^{J_p})^2}{J} \right\}.
\end{align}
Under $H_{0,p}$, the test statistic
\begin{align}
    T_p = \frac{\hat{K}_N^{J_{p+1}} - \hat{K}_N^{J_p}}{\sqrt{\text{est var}(\hat{K}_N^{J_{p+1}} - \hat{K}_N^{J_p} \mid J)}}
\end{align}
is approximately normally distributed.

For a given extrapolation size $M$, we can apply the same procedure to the estimators derived in \citet{gravel2011demographic} and \citet{gravel2014predicting}, which are again linear combinations of the fingerprint.

\renewcommand{\theequation}{E.\arabic{equation}}
\setcounter{equation}{0} 

\section{Good-Toulmin estimators} \label{sec:app-gt}

In recent work, \citet{chakraborty2019somatic} used the classic smoothed Good-Toulmin estimator \citep{good1956number, efron1976estimating, orlitsky2016optimal} in the context of rare variants prediction. Under the same sampling model assumed by \citet{zou2016quantifying}, this method allows to predict the number of additional variants that will be observed in $M$ additional samples by using the formula

\begin{align}
    \Delta_N(M) \mid X_{1:N} = \begin{cases} \sum_{r=1}^{\infty} (-1)^{r+1}\left(\frac{M}{N}\right)^r f_r &\mbox{ if } M/N \le 1 \\ \sum_{r=1}^{\infty} (-1)^{r+1}\left(\frac{M}{N}\right)^r f_r P(M, N, r) &\mbox{ if } M/N>1 \end{cases},
\end{align}
where
\begin{align}
    P(M, N, r) = \Pr(\Binom(\kappa(M,N)), \theta(M,N) \ge r)
\end{align}
where the smoothing parameters $\kappa$ and $\theta$ can take two different forms: either
\begin{align} \label{eq:GT_1}
    \kappa(M,N) = \lfloor 0.5 \log_2((M^2/N)/(M/N-1))\rfloor \quad \text{and} \quad \theta(M,N) = 1/(M/N+1) 
\end{align}
or
\begin{align} \label{eq:GT_2}
    \kappa(M,N) = \lfloor 0.5 \log_3((M^2/N)/(M/N-1))\rfloor \quad \text{and} \quad \theta(M,N) = 2/(M/N+2).
\end{align}


\renewcommand{\theequation}{F.\arabic{equation}}
\setcounter{equation}{0} 

\section{Additional details and experiments on the TCGA and MSK-impact dataset} \label{sec:app_TCGA_msk}

\subsection{Details about the experimental setup} \label{sec:app_TCGA_details}

The TCGA and the MSK-impact datasets are two publicly available cancer genomics datasets, containing somatic variants from $N=10{,}275$ and $N=9{,}091$ samples respectively. In both datasets, for each patient-id, we have access to a list of recorded variants, together with (i) the gene at which the variant was observed, (ii) and the type of cancer the patient was diagnosed with. The TCGA dataset contains variants from the whole exome (variants are recorded across a total of $G=19{,}441$ genes) for 33 different cancer types. The MSK-impact is a hybridization capture-based NGS clinical assay that is capable of detecting mutations in all exons and selected introns and promoter mutations in $G_1=412$ cancer-associated genes \citep{chakraborty2019somatic} across a finer classification of 329 different cancer types. 

We thank an anonymous reviewer for pointing out that, as in \citet{chakraborty2019somatic}, this data might suffer from normal cell contamination and tumor heterogeneity. We believe future work might address this issue by introducing an additional parameter variable describing what fraction of each sample belongs to the tumor, but this extension is beyond the scope of the present work.

From the data, it is natural to obtain a binary encoding of variation as described in \Cref{sec:model} either for specific cancer types (i.e.\, across all genes, restrict our attention to patients who got diagnosed with the same cancer type) or to specific genes. That is to say, use the machinery of \Cref{sec:prediction} to either predict (A) how many new variants we are going to observe from new samples that have been classified with a specific cancer type or (B) how many new variants we are going to observe from new samples in a specific gene (for any cancer type).

When compared across tumors, viceversa, the TCGA and MSK-impact, even when restricting the attention to the same targeted genes, can show substantially different behavior (see \Cref{fig:diff_cancers}). Therefore, we don't try to predict the rate of growth of new variants observed for cancer types using, e.g.\ the TCGA dataset as a pilot study and the MSK-impact as a follow-up study.

\begin{figure}
      \centering \includegraphics[width=\textwidth,height=\textheight,keepaspectratio]{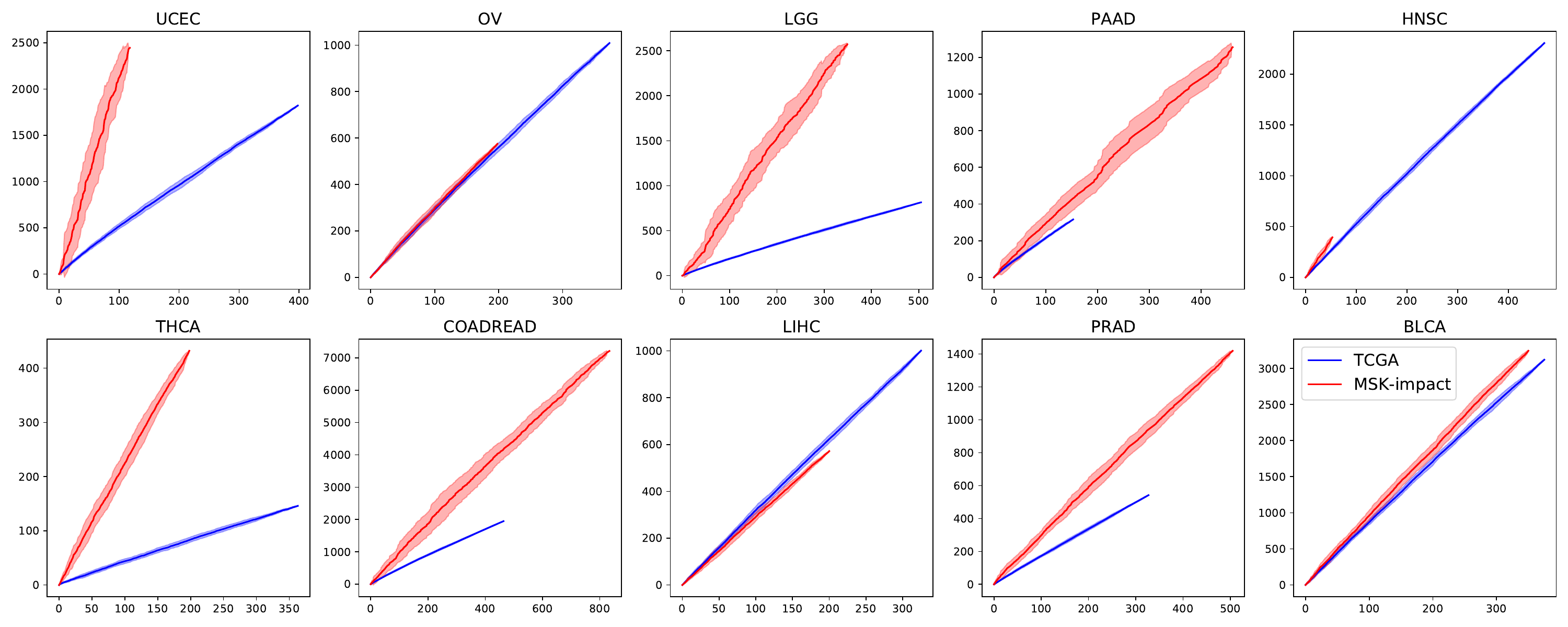}
      \caption{For ten different cancer subtypes, we plot the rate of growth of the number of distinct variants as a function of the sample size for TCGA and MSK-impact. We shuffle the observations according to 20 different permutations, the solid line represents the average number observed across all permutations, and the shaded regions represent one standard deviation above/below. For both datasets, we remove hypermutated outliers - we drop those samples that show more than 3 times the median number of variants observed in the tumor subtype dataset.}
\label{fig:diff_cancers}
\end{figure}

Across genes, the TCGA and MSK-impact are relatively similar in terms of number of variants observed per gene (see \Cref{fig:tcga_impact_comparison} and 
\Cref{fig:tcga_impact_comparison_2}). 

 \begin{figure}
      \centering \includegraphics[width=\textwidth,height=\textheight,keepaspectratio]{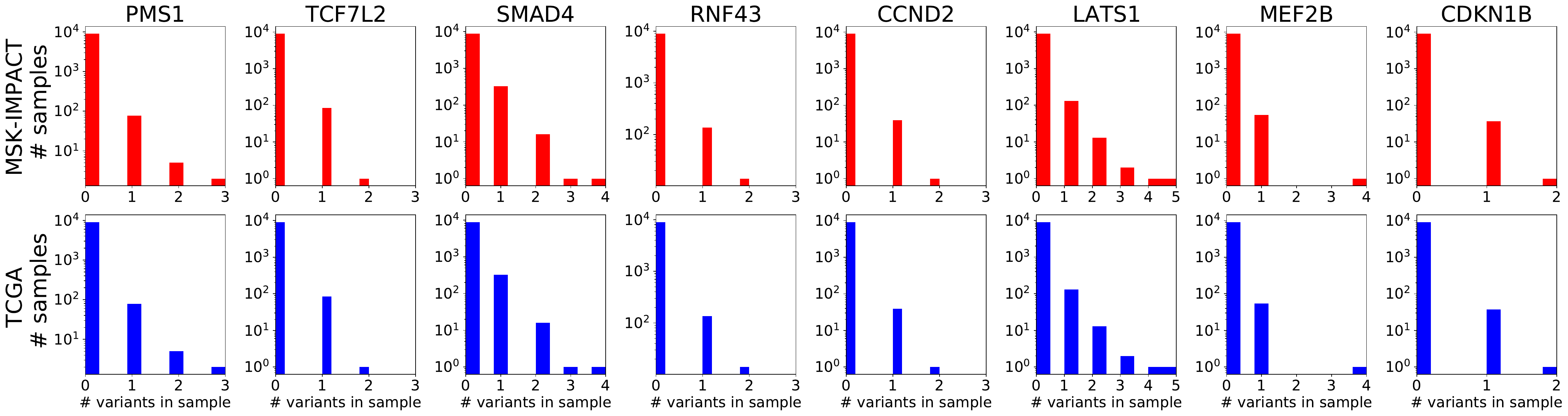}
      \caption{Comparing the number of variants observed for a given gene across the samples (top row: IMPACT, bottom row: TCGA; different columns are different genes).}
\label{fig:tcga_impact_comparison}
\end{figure}

 \begin{figure}
      \centering \includegraphics[width=\textwidth,height=\textheight,keepaspectratio]{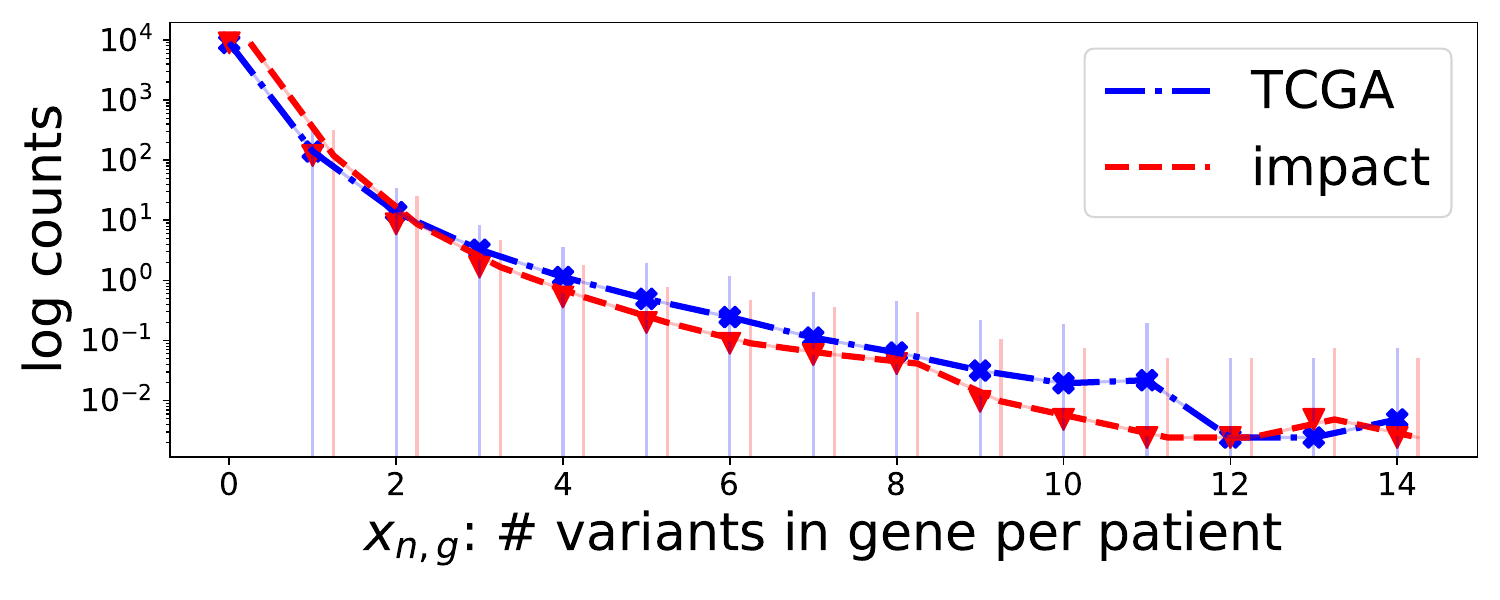}
      \caption{For every patient $n$ and every gene $g$, we let $x_{n,g}$ be the number of variants recorded for patient $n$ in gene $g$. This number ranges from $0$ to $14$ in both TCGA and MSK-impact. We plot the histogram of the $
      {x_{n,g}
      }$ for both the MSK-impact (red) and the TCGA (blue), $y$-axis in log-scale.}
\label{fig:tcga_impact_comparison_2}
\end{figure}

\par For the experiments under changing experimental conditions, i.e.\ in the setting in which samples are noisy, we  perform further thinning to generate the data. That is to say, given $K$ variants across samples, each with empirical frequency  $\hat{\theta}_k$, $k = 1,\ldots,K$ and for a given choice of $\threshold$, sampling error $p_{\text{err}}$ and sequencing quality $\lambda$, we obtain the associated probability $\phi$ that at least $\threshold$ successful reads are obtained at any position $k$ for any individual $n$, i.e.
\begin{align*}
    \phi(\lambda, \threshold, p_{\text{err}}):= \sum_{t\geq\threshold} \frac{1}{t!} e^{-\lambda(1-p_{\text{err}})} \{\lambda(1-p_{\text{err}})\}^t.
\end{align*}
Then, an individual observation $X_n = [x_{n,1},\ldots,x_{n,K}]$ is obtained by independently sampling Bernoulli random variables, 
\[
    x_{n,k} \mid \hat{\theta}_k, \phi(\lambda, \threshold, p_{\text{err}}) \sim \Bern(\hat{\theta}_j\phi(\lambda, \threshold, p_{\text{err}})).
\]

\subsection{Prediction across genes with the same experimental conditions} \label{sec:app_tcga_genes}

We replicate the setup of \citet{chakraborty2019somatic} and use the TCGA dataset as a pilot study and the MSK-impact as a follow-up study. We restrict our attention to the 412 targeted genes in the MSK-impact.
For each targeted gene, in a similar way as to what done for the experiments in \Cref{sec:exp}, we create ten different folds of the data, by sampling (without replacement) for each fold a random subset of $80\%$ of the data. We train on these folds of the TCGA dataset our BNP predictor, as well as the Good-Toulmin method used in \citet{chakraborty2019somatic} and the Jackknife estimators proposed by \citet{gravel2014predicting} to predict both (1) the expected number of new variants in a single new sample in MSK-impact and (2) the total number of new variants we expect to see in a total cohort of the same size as the MSK-impact ($M=9{,}091$ samples). The linear programming of \citet{zou2016quantifying} failed to provide reliable prediction, especially for those genes in which very few observations showed variation (i.e.\ those genes with few ``active'' patients). We therefore we excluded it in our analysis. As shown in \Cref{fig:tcga_impact_replicating_chakra_fig_4}, the three methods considered (Good-Toulmin, fourth order Jackknife and our BNP predictor) performed similarly. 

To quantify the predictive performance of the different methods we used the following setup: for gene $g$, let $Z^{\text{TCGA}, g}\in\{0,1\}^{N\times K}$ be the binary matrix of variation in gene $g$ in the TCGA and similarly $Z^{\text{MSK}, g}\in\{0,1\}^{M\times K}$ the corresponding matrix for the MSK-impact. Let $p_{1,g}$ be, across all patients in the MSK-impact, the average observed number of new variants displayed by one individual in gene $g$ which are not displayed by any patient in the TCGA dataset in gene $g$:
\[
    p_{1,g}:= \frac{1}{M}\sum_{m} \left[\sum_k 1(Z^{\text{MSK}, g}_{m,k}=1)1\left\{\sum_n \left(Z^{\text{TCGA}, g}_{n,k}=0\right)\right\}\right].
\] 
Similarly, let $p_{M,g}$ be the total number of new variants displayed in the MSK-impact that were not present in the TCGA dataset for gene $g$:
\[
    p_{M,g} = \sum_k \left\{1\left(\sum_m Z_{m,k}>0\right) 1\left(\sum_n Z_{n,k}=0\right)\right\}.
\]
For a given method $\mathcal{A}\in\{\text{BNP}, \text{GT}, \text{J4}\}$ (our Bayesian nonparametric method, the Good-Toulmin method used in \citet{chakraborty2019somatic} and the fourth order Jackknife proposed in \citet{gravel2014predicting}), let $\hat{p}_{1,g}^\mathcal{A}$ and $\hat{p}_{M,g}^{\mathcal{A}}$ be the prediction of $p_{1,g}$ and $p_{M,g}$ respectively. Given a threshold $t$ and method $\mathcal{A}$, for $m\in\{1,M\}$
\begin{align}
    \ell(t; g, m) = \min\left\{ \left|\frac{\hat{p}_{m,g} - p_{m,g}}{p_{m,g}} \right|, t\right\}, \label{eq:loss_ape}
\end{align}
and, summing over all genes $g$,
\begin{align}\label{eq:APE}
    \ell(t; m) = \sum_g \ell(t; g, m).
\end{align}
For $\alpha=1$ and $t = \infty$, \Cref{eq:APE} quantifies the absolute average percentage error. Setting $t<\infty$, instead, implies a Huber loss, in which we truncate the (percentage) loss in case it is larger than the threshold $t$. In practice, in our experiments, for every gene $g$ we have ten different folds of the data, which lead, for every predictor, to ten different prediction values. For every gene, and for every prediction method, we retain the median predicted value across the folds, and compute the error (\Cref{eq:APE}) relative to that fold. 
\begin{figure}
      \centering \includegraphics[width=\textwidth,height=\textheight,keepaspectratio]{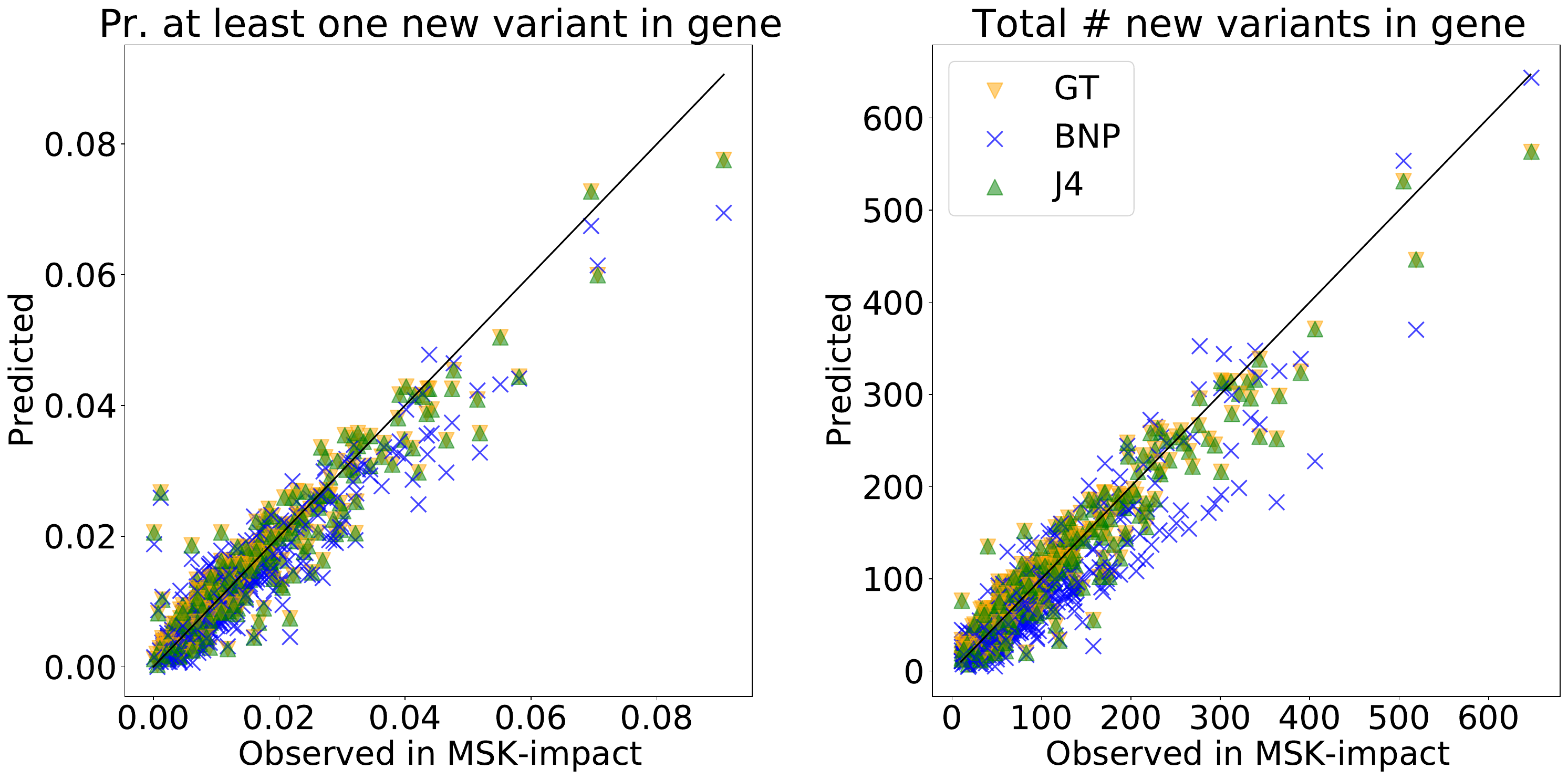}
      \caption{For every targeted gene present in the MSK-impact dataset, we use samples from the TCGA dataset to predict the expected number of new variants in that gene (1, left) from a new sample and (2, right) in a new cohort of the same size as the MSK-impact.}
\label{fig:tcga_impact_replicating_chakra_fig_4}
\end{figure}

\begin{figure}
      \centering \includegraphics[width=\textwidth,height=\textheight,keepaspectratio]{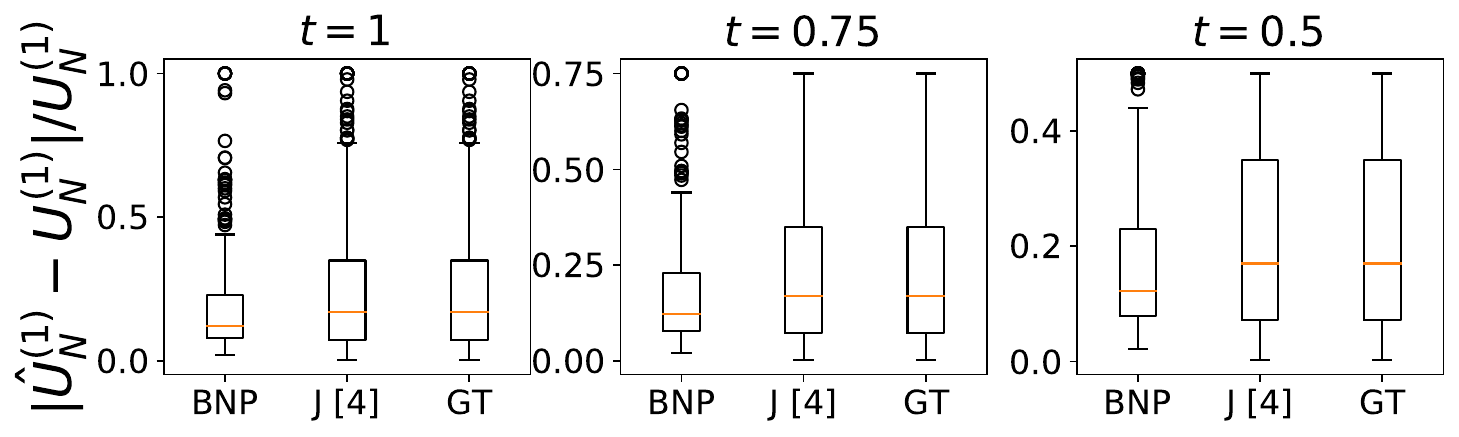}
\caption{For the one-step-ahead prediction problem, we use the TCGA dataset (on each targeted gene) to predict the expected number of new variants in one additional sample in the MSK-impact dataset using our BNP predictor, the (smoothed) Good-Toulmin predictor, as well as the fourth order jackknife. We report, for each fold in the data and for each gene, the loss introduced in \Cref{eq:APE} for $t \in \{1, 0.75, 0.5 \}$ through a boxplot.}
\end{figure}

\begin{figure}
      \centering \includegraphics[width=\textwidth,height=\textheight,keepaspectratio]{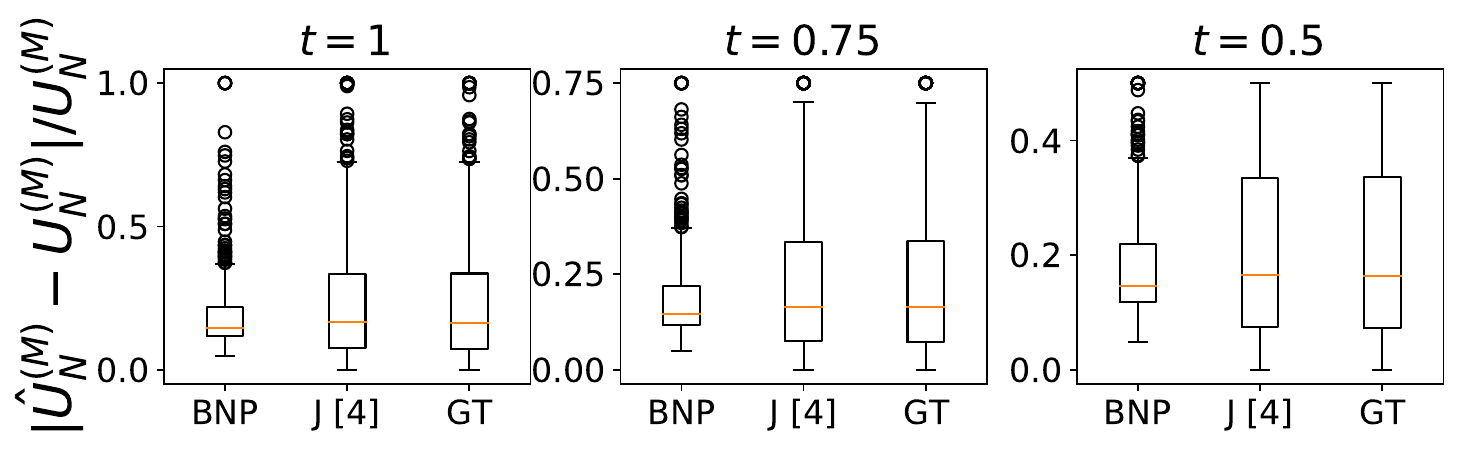}
\caption{For the many-steps-ahead prediction problem, we use the TCGA dataset (on each targeted gene) to predict the expected number of new variants in a new dataset of the size of the MSK-impact dataset using our BNP predictor, the (smoothed) Good-Toulmin predictor, as well as the fourth order jackknife.We report, for each fold in the data and for each gene, the loss introduced in \Cref{eq:APE} for $t \in \{1, 0.75, 0.5 \}$ through a boxplot.}
\end{figure}

\subsection{Predictions across genes with different experimental conditions}

We then move to prediction under changing experimental conditions. Given the original MSK-impact cohort (sampled at an average of 480x depth), we generate pseudo-observations by subsampling data at 100x. To do so, we assume that  each patient-locus pair for which variation is observed, $X_{n,k} = 1$, in the final dataset is generated as follows: first a Poisson random variable $ Y_{n,k} \sim \Pois (\lambda)$ is drawn, and then $X_{n,k}=1(Y_{n,k}\ge T)$ is kept. We pick $\lambda = 100$ and $T=90$ in our experiment. On the new dataset, again, our predictor outperforms alternative methods (see \Cref{fig:genes_changing}). The explanation follows the one given in \Cref{sec:exp}: our method can adapt to the changing conditions whereas competing methods cannot.

\begin{figure}
      \centering \includegraphics[width=\textwidth,height=\textheight,keepaspectratio]{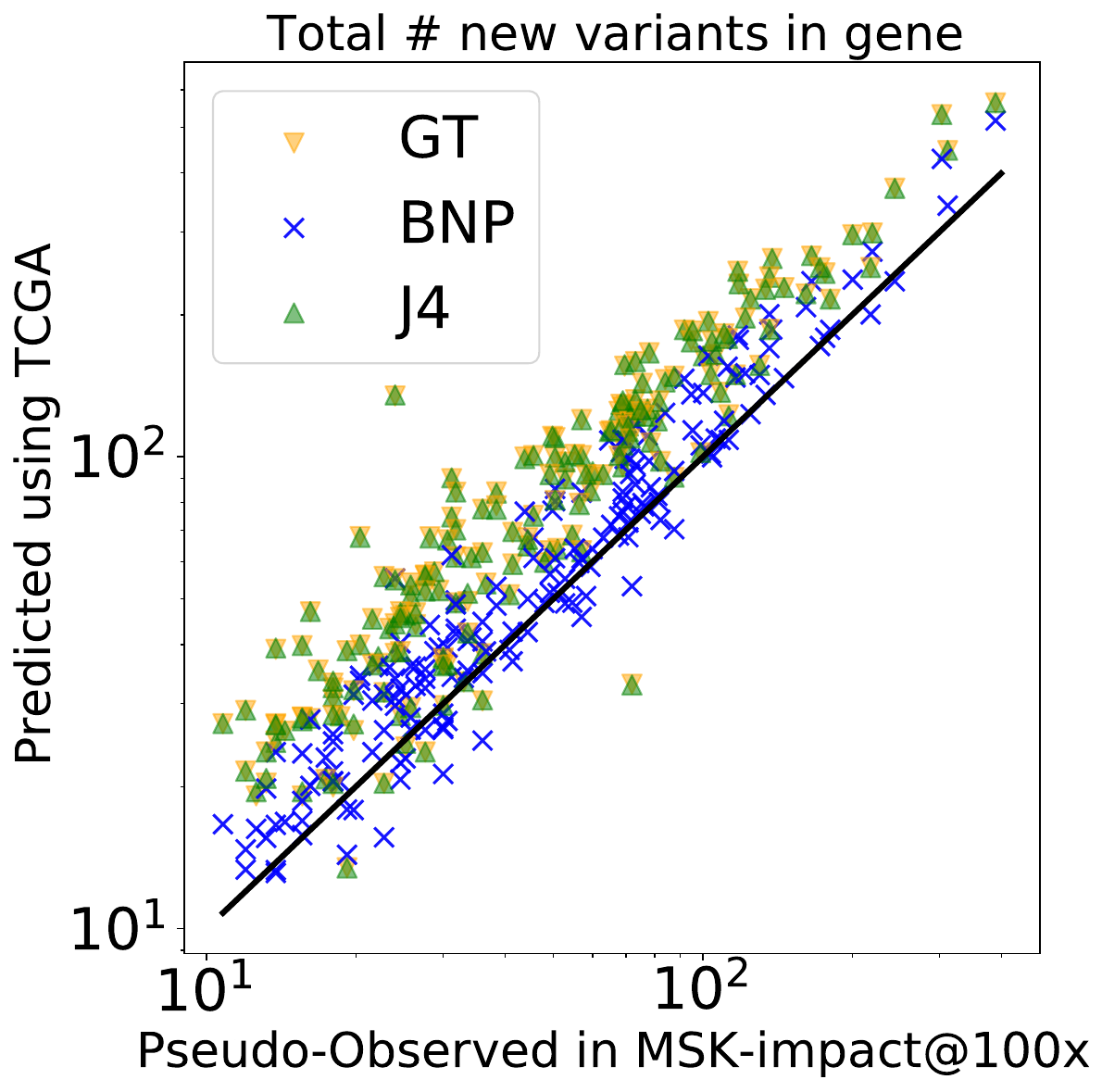}
      \caption{For every targeted gene present in the MSK-impact dataset, we use samples from the TCGA dataset to predict the expected number of new variants in that gene in a new cohort of the same size as the MSK-impact, but now assuming the MSK is sampled at a different sequencing depth than the TCGA.}
\label{fig:genes_changing}
\end{figure}

\subsection{Optimal design of experiments}
Last, we use prediction under changing conditions to inform the (optimal) design of a follow up study. We choose again the cost function $C(m,\lambda) = m\log \lambda$, to sample $m$ new observations at depth $\lambda$. We fix a budget which allows us to sample at full depth ($\lambda=480$) only $M'=M/2$ observations, half of the MSK-impact sample size.

We find that the same trade-off observed in \Cref{sec:exp_opt_d} is also present here. Across genes, we can find a configuration of the sequencing depth ($\lambda = 62$) that leads to a median gain of $6$ additional new variants discovered (per gene), that is an average increase of $10.56\%$ with respect to the number of variants we would have discovered if we had used the full-sequencing depth available alternative under the same budget constraint and cost function.

\begin{figure}
      \centering \includegraphics[width=\textwidth,height=\textheight,keepaspectratio]{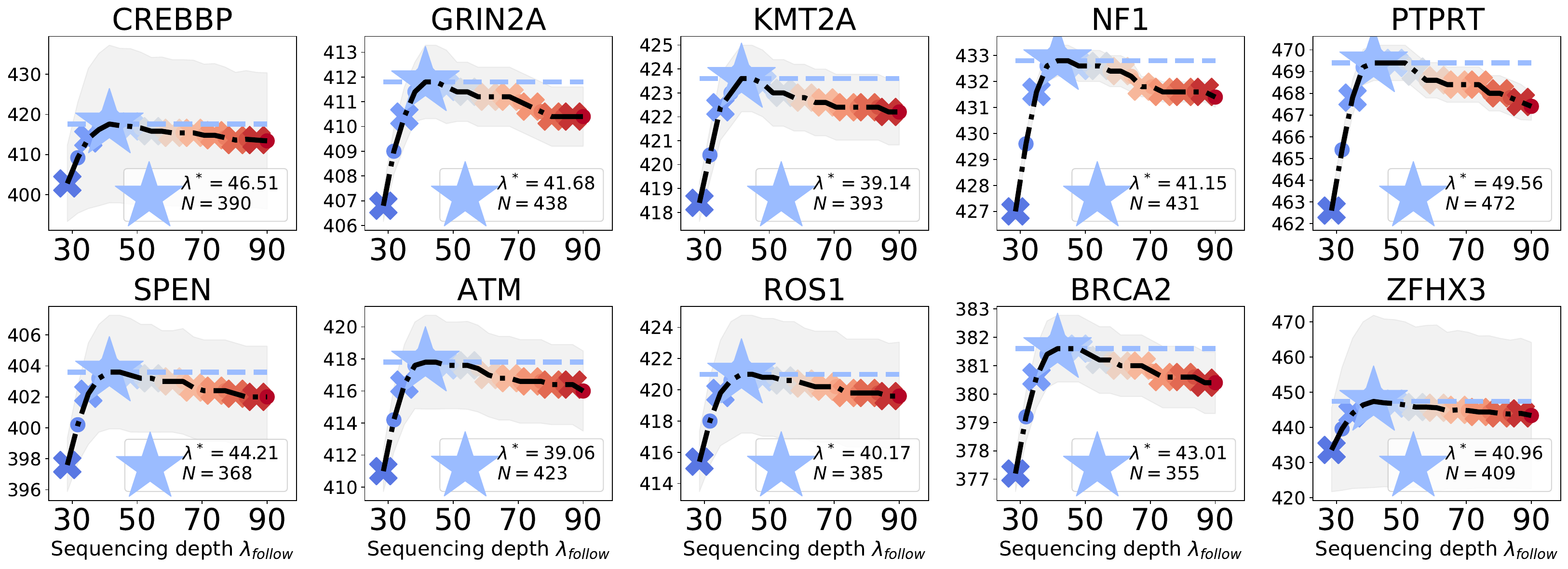}
      \caption{For every targeted gene present in the MSK-impact dataset, we use samples from the TCGA dataset to find the optimal configuration of the sequencing depth in the follow up study so maximize the number of new variants discovered. Here we display results for the 10 genes with the largest number of variants.}
\label{fig:genes_opt}
\end{figure}


\renewcommand{\theequation}{G.\arabic{equation}}
\setcounter{equation}{0} 

\section{Additional experimental results: GnomAD data} \label{sec:app-exp_gnomAD}

\subsection{Experimental setup}
\label{sec:exp_setup_gnomAD}

In order to run our experiments, we use data from the gnomAD (genome aggregation dataset) discovery project \citep{karczewski2019variation}, the largest and most comprehensive publicly available human genome dataset. This dataset contains 125{'}748 exomes sequences (i.e.\, protein-coding regions of the genome), from 8 main populations (African American, Latino, Ashkenazi Jewish, East Asian, Finnish, Non-Finnish European, South Asian, Other\footnote{The ``Other'' subgroup contains all ``\emph{individuals were classified as "other" if they did not unambiguously cluster with the major populations in a principal component analysis}''}). Sample size varies widely across sub populations, e.g.\ the ``Other'' subgroup counts only 3{'}070 observations, while ``Non-Finnish European'' contains 56{'}885 individuals. Moreover, some of these main populations are further split into additional sub populations, e.g.\ ``Non-Finnish European'' contains the ``Bulgarian'', ``Estonian'', ``Northern European'', ``Southern European'', ``Swedish'', ``Other European'' sub populations, while the ``East Asian'' sub population is further split into the ``Korean'', ``Japanese'' and ``Other East Asian'' sub populations (see \citet{karczewski2019variation} for additional details). We ran our analysis on all populations and sub populations.

Because for privacy reasons not all individual sequences are accessible, in order to run our analysis we generate synthetic data which closely resembles the true data as follows. For every subpopulation with $N$ individuals and every position $j=1,\ldots,K$ in the exome, we have access to the total number of individuals $N_j$ showing variation at position $j$. We compute the empirical frequency of variation at site $j$, $\hat{\theta}_j := N_j/N$ for all $j=1,\ldots,K$. Our data is then generated by sampling independent Bernoulli random vectors $X_1,\ldots,X_N$, with $X_n = [x_{n,1},\ldots,x_{n,K}]$. The entries in the vector are independent Bernoulli random variables, $x_{n,j} \sim \Bern(\hat{\theta}_j)$. 

\par For the prediction experiments under changing conditions, i.e.\ in the setting in which samples are noisy, we  perform further thinning to generate the data. That is to say, given the empirical frequencies $\{\hat{\theta}_j\}$, $j = 1,\ldots,K$ and for a given choice of $\threshold$, sampling error $p_{\text{err}}$ and sequencing quality $\lambda$, we obtain the associated probability $\phi$ that at least $\threshold$ successful reads are obtained at any position $j$ for any individual $n$, i.e.
\begin{align*}
    \phi(\lambda, \threshold, p_{\text{err}}):= \sum_{t\geq\threshold} \frac{1}{t!} e^{-\lambda(1-p_{\text{err}})} \{\lambda(1-p_{\text{err}})\}^t.
\end{align*}
Then, an individual observation $X_n = [x_{n,1},\ldots,x_{n,K}]$ is obtained by independently sampling Bernoulli random variables, 
\[
    x_{n,j} \mid \hat{\theta}_j, \phi(\lambda, \threshold, p_{\text{err}}) \sim \Bern(\hat{\theta}_j\phi(\lambda, \threshold, p_{\text{err}})).
\]

\subsection{Prediction with no sequencing errors} \label{sec:exp_noerror_gnomAD}

\begin{figure}[t!]
      \centering \includegraphics[width=\textwidth,height=\textheight,keepaspectratio]{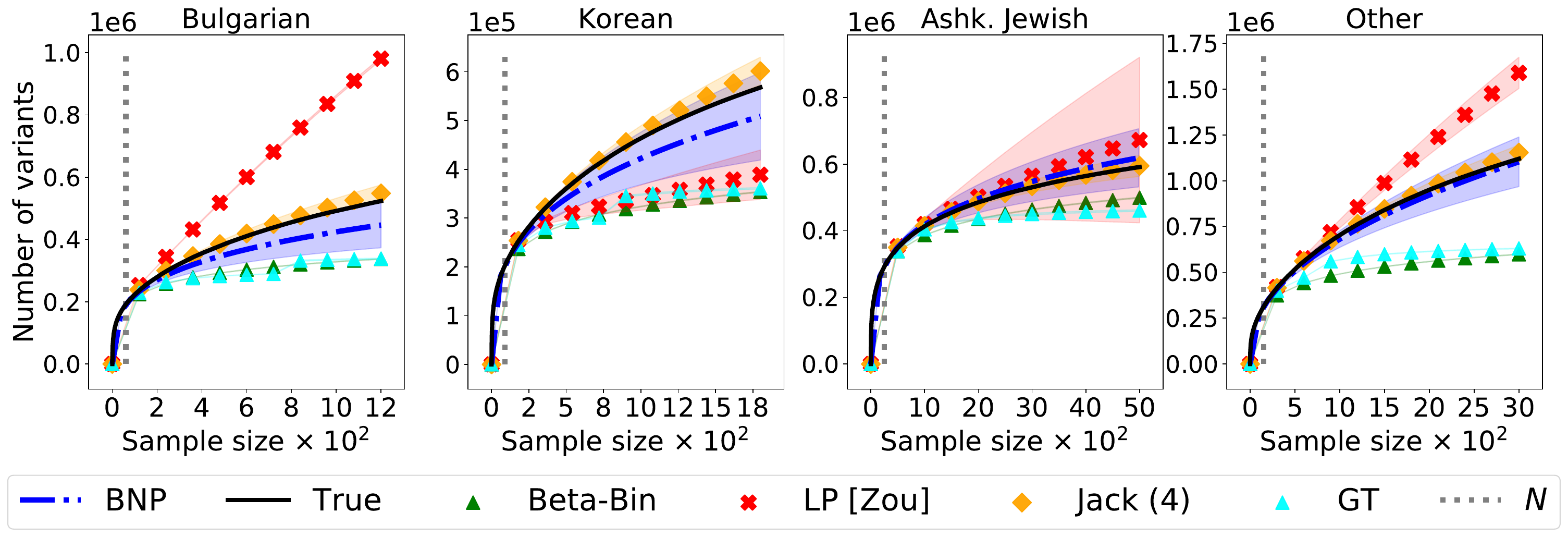}
\caption{Predicting the number of new variants in a follow-up study under constant experimental conditions. 
The solid black line displays the true number of distinct variants (vertical axis) as the sample size increases (horizontal axis). The dotted vertical line indicates the pilot study sample size.
Lines for each method are averaged across all folds; see \Cref{sec:exp_noerror_gnomAD} (blue: our method, Bayesian nonparametric (BNP), Eq. \eqref{eq:naive_new}; green: \citet{ionita2009estimating}; red: \citet{zou2016quantifying}; orange: \citet{gravel2014predicting} (4th order), cyan: \citet{chakraborty2019somatic}). Shaded regions show one standard deviation across data folds. }
\label{fig:prediction_all_gnomAD}
\end{figure}

Various existing methods predict the number of new variants in a follow-up study under the assumption that experimental conditions remain constant between the pilot and follow-up. These approaches use, respectively, parametric Bayesian methods \citep{ionita2009estimating}, linear programming \citep{gravel2014predicting, zou2016quantifying}, a harmonic jackknife \citep{gravel2014predicting}, and a smoothed version of the classic Good-Toulmin estimator \citep{orlitsky2016optimal,chakraborty2019somatic}. We encountered numerical issues with the linear programming method of \citet{gravel2014predicting} and so, like \citet{zou2016quantifying}, we do not include it in our comparison. Even in their original paper, \citet{gravel2014predicting} did not report superior performance of their linear programming approach over their other method, the harmonic jackknife \citep{gravel2014predicting}, which we do include in our comparison. To assess prediction error in each case, we use an approach akin to cross validation. Namely, we treat each subpopulation as a dataset. We divide the subpopulation into 33 folds
of equal size. For a smaller number of folds, each fold represents a larger pilot study. All methods improve when the pilot study is increased substantially in size, i.e.\ when there is more information in the pilot. We find that the choice of 33 folds creates a challenging scenario with a small amount of pilot information. Nonetheless, both our method and the harmonic jackknife still perform well in these conditions. We consider each fold in turn as data from the pilot study and treat the remaining data (i.e., the data not in this fold) as the follow-up. We follow \citet{zou2016quantifying} to make a visual summary of our results; namely, we plot the mean number of variants across all folds as a function of dataset size in the pilot, and we plot the mean number of total predicted variants (across both pilot and follow-up) as a function of dataset size in the same plot. A vertical dashed line marks the pilot size. Shaded regions indicate one empirical standard deviation, measured across the folds. We include the exact values in the plot for comparison. \Cref{fig:prediction_all_gnomAD} demonstrates that our method matches the exact value more closely than the parametric Bayesian approach \citep{ionita2009estimating}, the linear programming approach \citep{zou2016quantifying} and the nonparametric smoothed Good-Toulmin method \citep{orlitsky2016optimal,chakraborty2019somatic}. And our method has roughly the same performance as the jackknife approach \citep{gravel2014predicting} when the pilot and follow-up have the same experimental conditions. We next explain the relative performance of all methods in more detail.

In \Cref{sec:app-exp_additional} we run both the Bayesian parametric approach and our method on data simulated under the parametric Bayesian model used by \citet{ionita2009estimating}. We also run both methods on data simulated under the 3-parameter beta process model we propose above; see \Cref{sec:app-exp_additional_bnp}. The approach of \citet{ionita2009estimating} provides excellent predictions when the data is generated under their assumed model, but deteriorates in performance under the simulated 3-parameter beta process data, as for real data. Therefore, we believe the parametric Bayesian method suffers on real data due to the assumed \citet{ionita2009estimating} model being ill adapted for real-life power laws. Similarly, we show in \Cref{sec:app-exp_additional_gt} that the smoothed Good-Toulmin predictor \citep{orlitsky2016optimal, chakraborty2019somatic} performs well for power laws with low exponent values, as expected. However, we also see that this estimator performs poorly for power laws with high exponent values. We verify in \Cref{sec:app-exp_additional_gt} that the gnomAD data exhibits high exponent values. This behavior explains the underperformance of the smoothed Good-Toulmin predictor in our real-data experiments.

\citet{zou2016quantifying} use a linear program to estimate \emph{rare} variant frequencies; they approximate frequencies of common variants with the empirical frequency. ``Rare'' is defined to be any frequency less than $\kappa/100$, for some user-defined threshold $\kappa \in (0,100)$, interpreted as a percent. In practice, we found that the output of the algorithm is very sensitive to the choice of $\kappa$ (see \Cref{sec:app-exp_additional_zou}). The authors suggest $\kappa = 1$ as a default setting, but we observed numerical instability and poor predictive performance for this choice. This observation holds especially when the pilot size $N$ is small, which we believe to be a particular case of interest in designing experiments for further data collection (i.e., for the follow-up study). For instance, we expect the small-$N$ case to arise frequently in the study of non-model organisms \citep{russell2017non}. In \Cref{fig:prediction_all_gnomAD}, we chose $\kappa = 20$, which led to convergence of the optimization algorithm in all cases. We explore other values of $\kappa$ in \Cref{sec:app-exp_additional_zou}. 

\subsection{Prediction under different experimental conditions}

\begin{figure}[t!]
      \centering \includegraphics[width=\textwidth,height=\textheight,keepaspectratio]{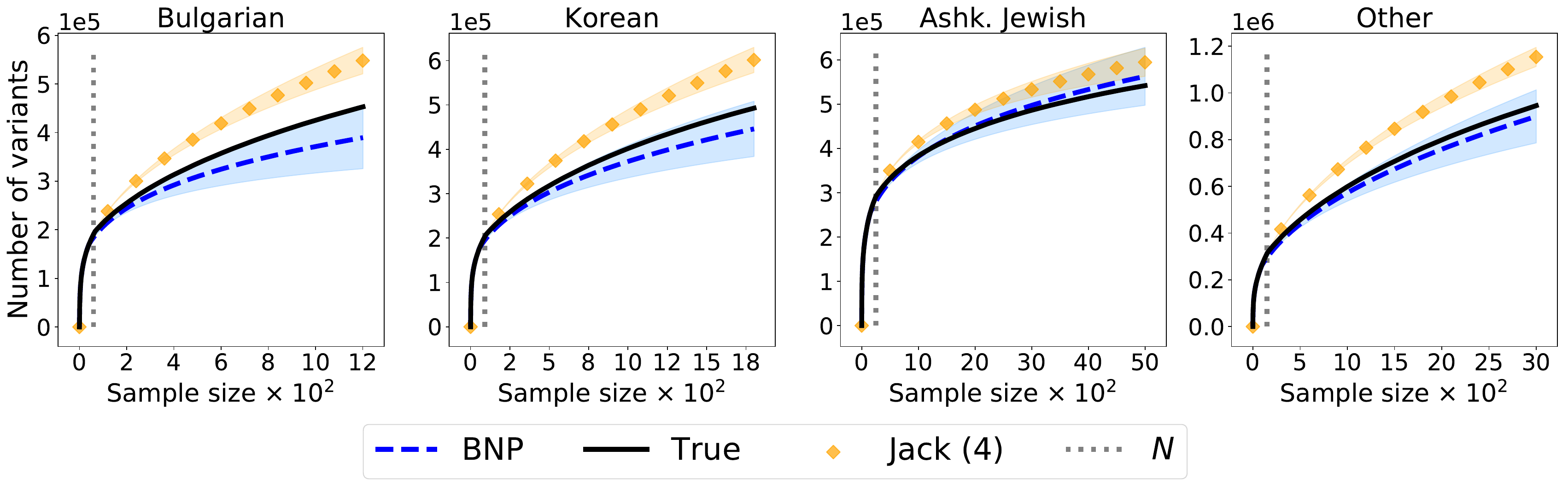}
\caption{Predicting the number of new variants under different experimental conditions between the pilot and follow-up. Same four subpopulations (gnomAD). Pilot sequencing quality is $\seqinit = 45$. Follow-up sequencing quality is $\seqfollowup=32$. Horizontal axis is the number of samples. Vertical axis is the number of total observed variants across both pilot and follow-up. The threshold is $T=30$.}\label{fig:pred_exp_d_gnomAD}
\end{figure}

We now turn to the case where there may be sequencing errors in the pilot, in the follow-up, or both. And the sequencing quality may differ between the pilot and the follow-up. No existing method works in this case. Since, like our method, the method of \citet{ionita2009estimating} is Bayesian, we believe it could be straightforwardly adapted using similar ideas to the ones we present here. But we have already seen that our Bayesian nonparametric method provides much more accurate predictions in the case of no sequencing errors, so we believe it is more fruitful to develop the Bayesian nonparametric approach. Similarly, we believe the linear programming approach \citep{zou2016quantifying} might be adapted to the special case where the follow-up is more error prone than the pilot. But even this development would still leave open the case where the follow-up might be made more accurate by increasing sequencing quality, and we have already observed that the Bayesian nonparametric approach provides better, more automatic predictions in the case of error-free observation. Finally, while the jackknife approach performs very well in the case with no sequencing errors, we do not think it will be as straightforward to adapt to the case where sequencing quality may change between the pilot and follow-up. 

In \Cref{fig:pred_exp_d_gnomAD} we see that there is indeed a noticeable difference in the number of observed variants when the experimental conditions change between the pilot and follow-up. In particular, we consider a pilot sequencing quality $\seqinit=45$ and a follow-up sequencing quality $\seqfollowup=32$. We use a fixed threshold $T=30$, a standard coverage value in human genomic experiments \citep{karczewski2019variation}. To represent this change between studies, we use the gnomAD data as in \Cref{sec:exp_noerror_gnomAD} but apply additional thinning to simulate imperfect observation due to sequencing depth; see \Cref{sec:exp_setup_gnomAD} for additional details. 
Since the jackknife is not able to use information about the changing sequencing depth, we expect our Bayesian nonparametric method should deliver superior predictive performance when sequencing quality changes. This behavior is exactly what we see in \Cref{fig:pred_exp_d_gnomAD}.

\subsection{Designing experiments to maximize the number of observed variants}

\begin{figure}[t!]
      \centering \includegraphics[width=\textwidth,height=\textheight,keepaspectratio]{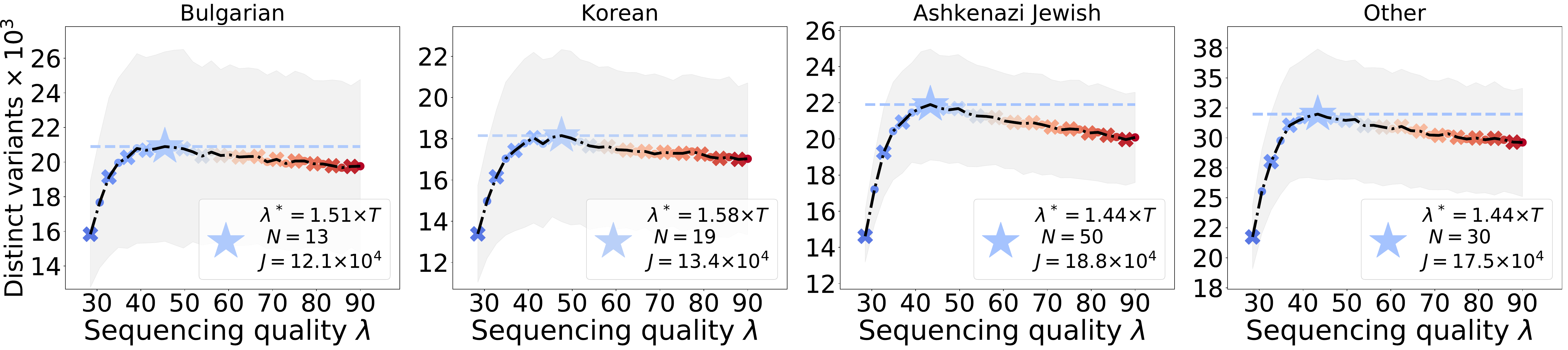}
\caption{Designing an experiment to maximize the number of new variants in a follow-up study. Same four subpopulations (gnomAD). Horizontal axis is the follow-up sequencing quality $\seqfollowup$. Vertical axis is the predicted number of observed variants in the follow-up by maximizing $M$ under the budget $\budget$ and quality $\seqfollowup$.}
\label{fig:opt_design_gnomAD}
\end{figure}

Finally, we demonstrate that our Bayesian nonparametric predictor can be used for experimental design in practice. Our procedure consists of three steps. (1) Given the pilot data and sequencing quality $\seqinit$, we minimize \Cref{eq:cost} to estimate the parameters $\conc, \discount, \mass$. (2) Next, we consider a range of values of the follow-up sequencing quality $\seqfollowup$. For each $\seqfollowup$, we choose the maximum follow-up size $M$ that stays within our budget $\budget$. And we use the learned values of the parameters $\conc, \discount, \mass$ to predict the number of new variants in each case. (3) We choose the settings of $\seqfollowup$ and $M$ that maximize the number of new variants.

We illustrate this procedure in \Cref{fig:opt_design_gnomAD}.

In our experiments, we set the cost function $\costfcn(M,\seqfollowup) = M \log (\seqfollowup)$ as in \citet{ionita2010optimal}, budget $D = 3000$, threshold $\threshold = 30$, error $p_{\text{err}} = 0.01$, and $\seqinit = 40$. We run the procedure over all folds, plot the empirical mean line, and plot the shaded region to illustrate one standard deviation. We see a trade-off in quality and quantity in \Cref{fig:opt_design_gnomAD}. Namely, maximizing quantity $M$ leads to very small values of $\seqfollowup$ to maintain the budget $\budget$. With sufficiently low quality, though, fewer variants are discovered. Conversely, when $\seqfollowup$ is set very high, we require a very small $M$ to maintain the budget $\budget$, and not many variants are discovered. Intermediate values of $\seqfollowup$ and $M$ serve to maximize the number of variants discovered under a fixed budget.

In \Cref{fig:ionita_gnomaAD_2}, \Cref{fig:unseen_gnomAD_2} and \Cref{fig:jack_gnomAD_2} we report results of the prediction of the number of new variants on some sub populations of the gnomAD dataset. We consider the Bulgarian, South Korean, Other East Asian, African and East Asian subpopulations. The $x$-axis displays the total number of samples collected. On the $y$-axis, we plot the number of distinct genomic variants. The solid black line displays the true number of distinct variants, the vertical grey line is placed in correspondence of the training sample size $N$ (left: $N\in\{42, 61, 228, 257, 291\}$). 

 \begin{figure}
      \centering \includegraphics[width=\textwidth,height=\textheight,keepaspectratio]{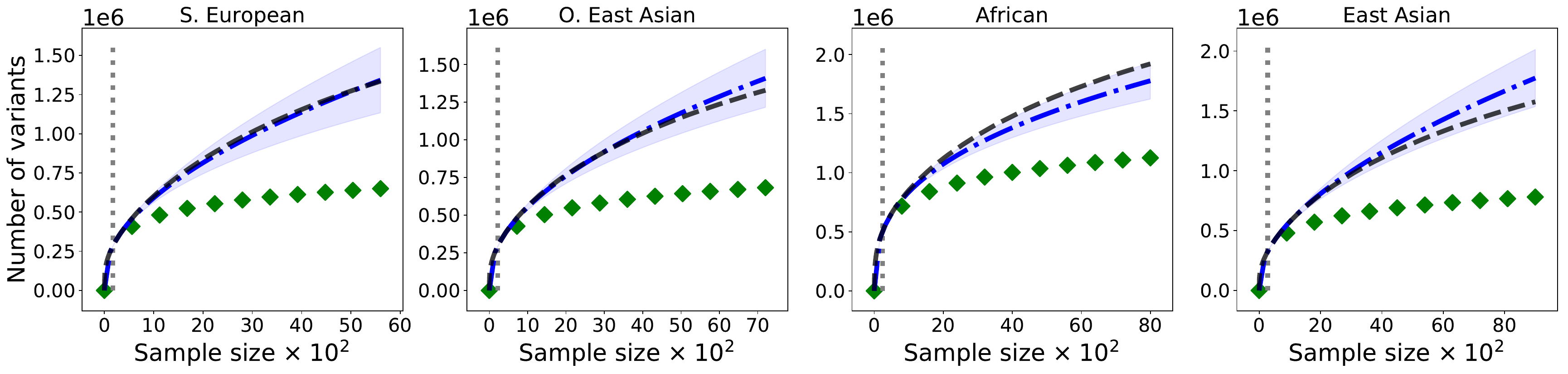}
      \caption{The blue dotted line is the posterior predictive mean of the number of distinct variants $U_N^{(m)}$ observed according to the Bayesian nonparametric predictor, averaged across 33 samples of size $N$. The green diamonds report the posterior predictive mean of  the Bayesian parametric estimator of \citet{ionita2009estimating}, averaged across the same subsets of the original data. The shaded blue and green regions report the prediction error by covering one standard empirical deviation for the two predictors.} 
\label{fig:ionita_gnomaAD_2}
\end{figure}

 \begin{figure}
      \centering \includegraphics[width=\textwidth,height=\textheight,keepaspectratio]{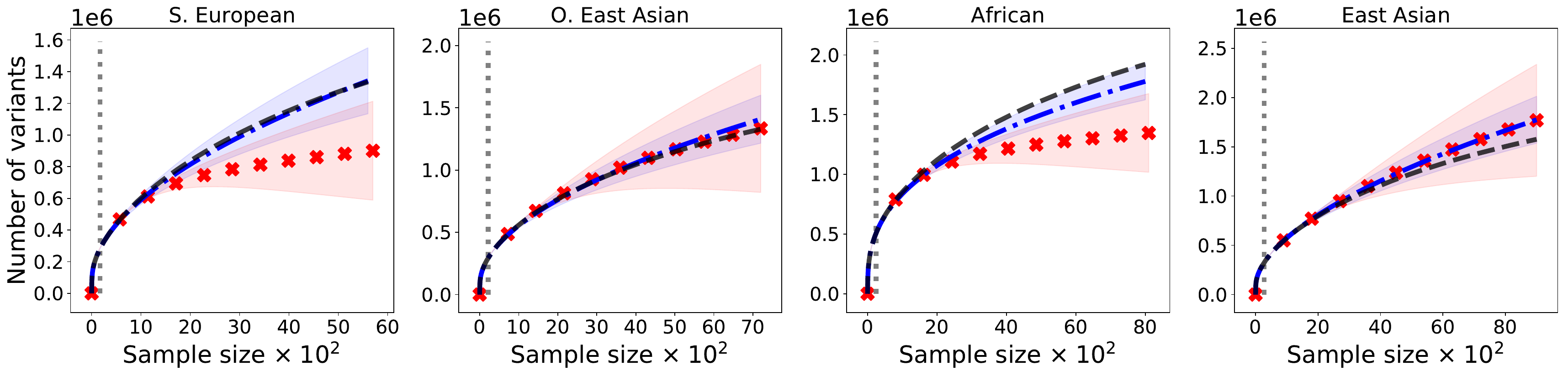}
      \caption{Results of the estimation of the number of new variants on some sub populations of the gnomAD dataset. The $x$-axis displays the total number of samples collected. On the $y$-axis, we plot the number of distinct genomic variants observed. The solid black line keeps track of the true number of distinct variants, the vertical grey line is placed in correspondence of the training sample size $N$. The blue dotted line is the posterior predictive mean of the number of distinct variants $U_N^{(m)}$ observed according to the Bayesian nonparametric predictor, averaged across 33 samples of size $N$. The dotted red line is the empirical mean of  the UnseenEST estimator of \citet{zou2016quantifying} across the same samples. The shaded blue and red regions report the prediction error by covering one standard empirical deviation for the two predictors. Here, we fix $\kappa = 1\%$, the value considered in \citet{zou2016quantifying}.} 
\label{fig:unseen_gnomAD_2}
\end{figure}

 \begin{figure}
      \centering \includegraphics[width=\textwidth,height=\textheight,keepaspectratio]{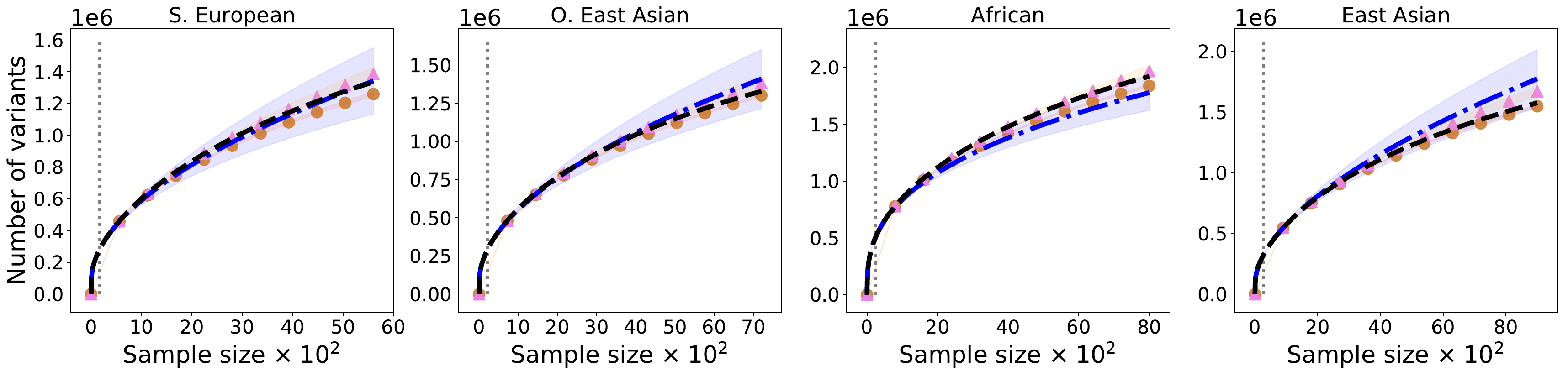}
      \caption{Again for the same sub populations considered  in \Cref{fig:ionita_gnomaAD_2} and \Cref{fig:unseen_gnomAD_2}, we compare  the Bayesian nonparametric estimator to the Jackknife estimator proposed in \citet{gravel2014predicting}, for the third and fourth orders. Lower order consistently underestimate the number of distinct variants} 
\label{fig:jack_gnomAD_2}
\end{figure}

 \begin{figure}
      \centering \includegraphics[width=\textwidth,height=\textheight,keepaspectratio]{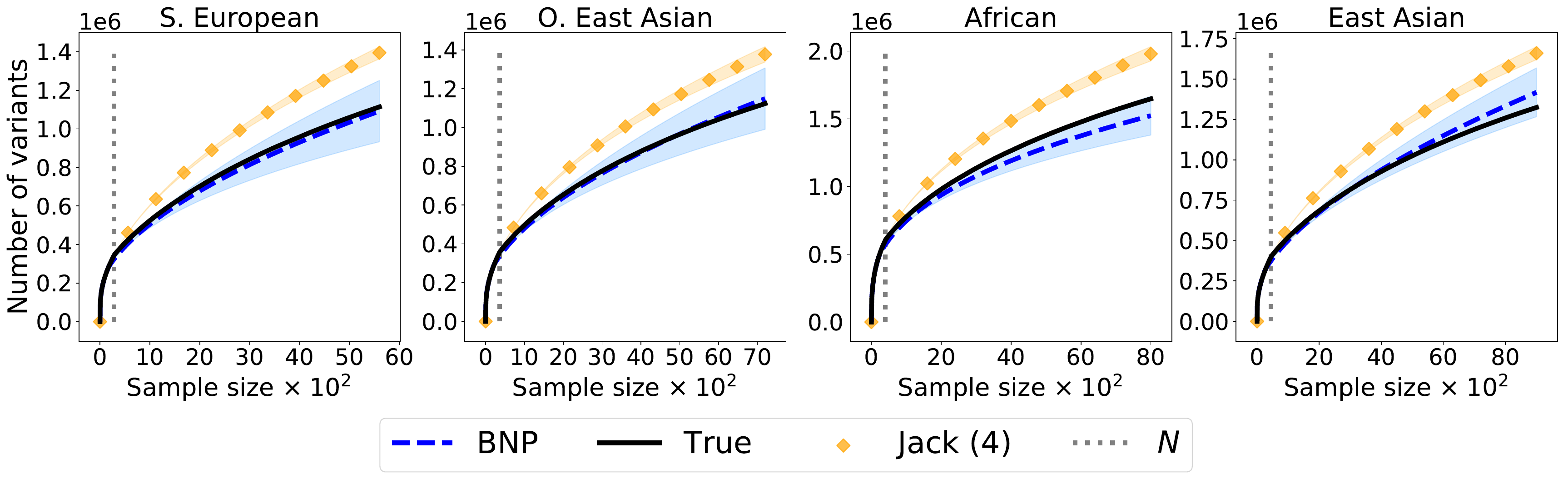}
      \caption{Prediction of the number of yet unseen variants for five subpopulations (gnomAD). For each subpopulation, we assume to have access to a small number of high quality genetic samples $N$ in which variation is observed at $J$ distinct loci, collected at initial sequencing depth $\seqinit = 45$ and threshold $\threshold = 30$. We imagine that the follow-up is performed at a different sequencing depth, $\seqfollowup=32$.} 
\label{fig:pred_exp_d_2}
\end{figure}


\renewcommand{\theequation}{H.\arabic{equation}}
\setcounter{equation}{0} 

\section{Additional experimental results: results on synthetic data} \label{sec:app-exp_additional}

\subsection{Synthetic data from the Indian buffet process} \label{sec:app-exp_additional_bnp}
In this section, we provide experimental results for data drawn from the three parameters Indian buffet process. When the data is drawn from the true model, we expect the Bayesian nonparametric estimators of \Cref{sec:prediction} to work particularly well. We test against a large collection of parameters $\alpha>0$, $\sigma \in [0,1)$ and $c>-\sigma$. We report here results for different configurations. In all cases, the optimization procedure outlined in \Cref{sec:exp} recovers the rate of growth of the distinct variants. Interestingly, in some instances, the optimization recovers parameters that differ from the true parameters that generate the process, but still have good predictive performance (see \Cref{fig:app_bnp_0}).

 \begin{figure}
      \centering \includegraphics[width=\textwidth,height=\textheight,keepaspectratio]{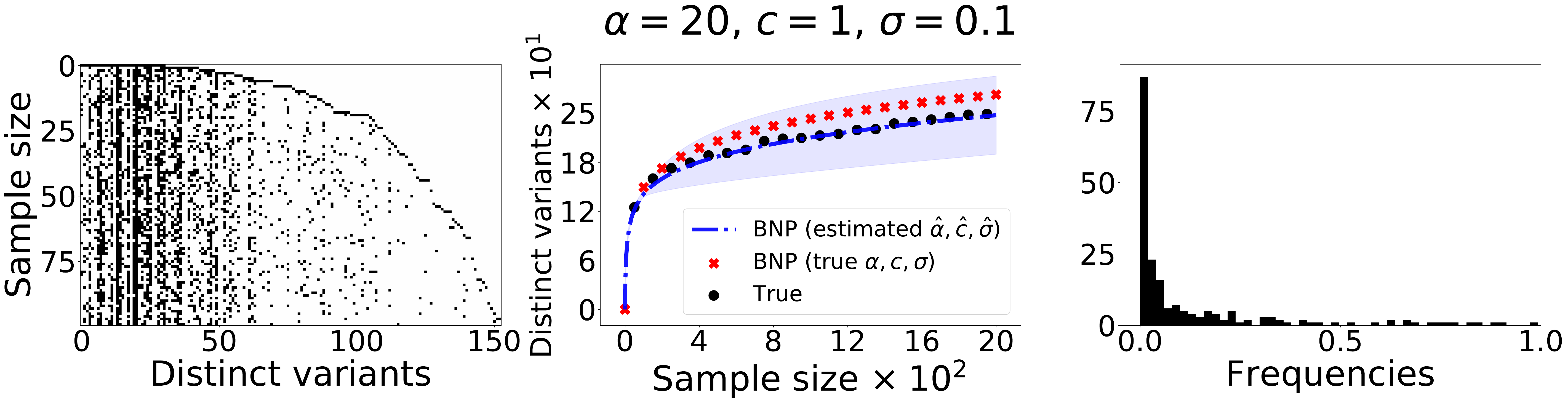}
      \caption{A draw from a three-parameter Indian buffet process. Here, $\alpha = 20$, $c=1$, $\sigma = 0.1$. In the left panel, we see the binary matrix $\bm{X}$ containing the first $N=100$ samples ($x$-axis) from the process, in its left-ordered-form (lof), i.e.\ variants ($y$-axis) are sorted by the order of appearance, so that as more points are added to the dataset, more columns contain nonzero entries. In the central panel, we plot the number of distinct variants ($y$-axis) as a function of the sample size ($x$-axis), extrapolating up to $M=1900$ additional samples. Last, on the right panel, we plot the empirical distribution of frequencies among the first $N$ samples.}
\label{fig:app_bnp_0}
\end{figure}

 \begin{figure}
      \centering \includegraphics[width=\textwidth,height=\textheight,keepaspectratio]{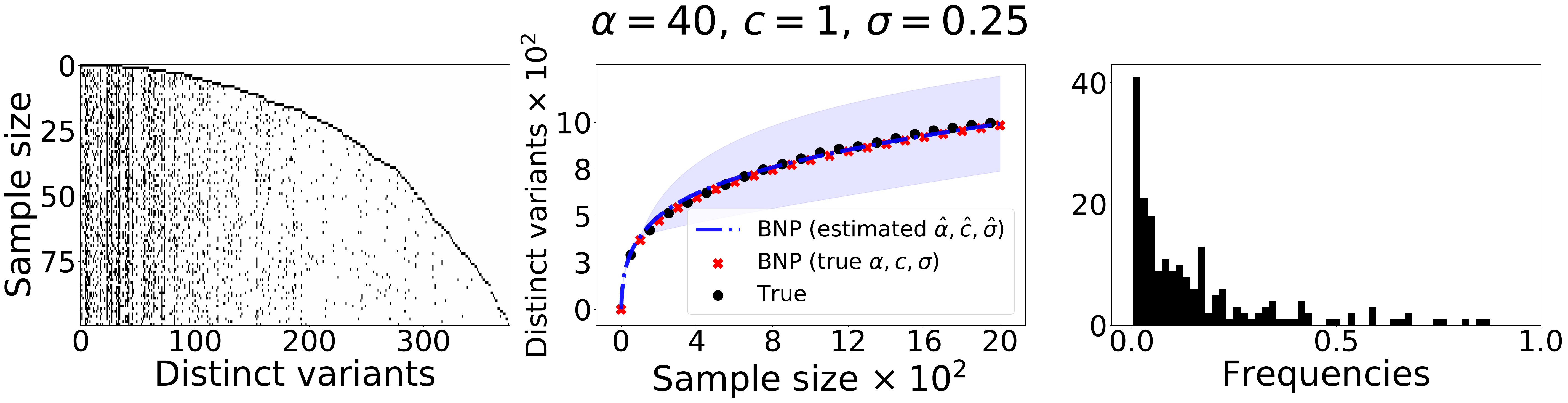}
      \caption{In this figure, we reproduce the visualizations explained in \Cref{fig:app_bnp_0} for a draw from a three-parameter Indian buffet process with parameters $\alpha = 40$, $c=1$ and $\sigma = 0.25$}
\label{fig:app_bnp_1}
\end{figure}

 \begin{figure}
      \centering \includegraphics[width=\textwidth,height=\textheight,keepaspectratio]{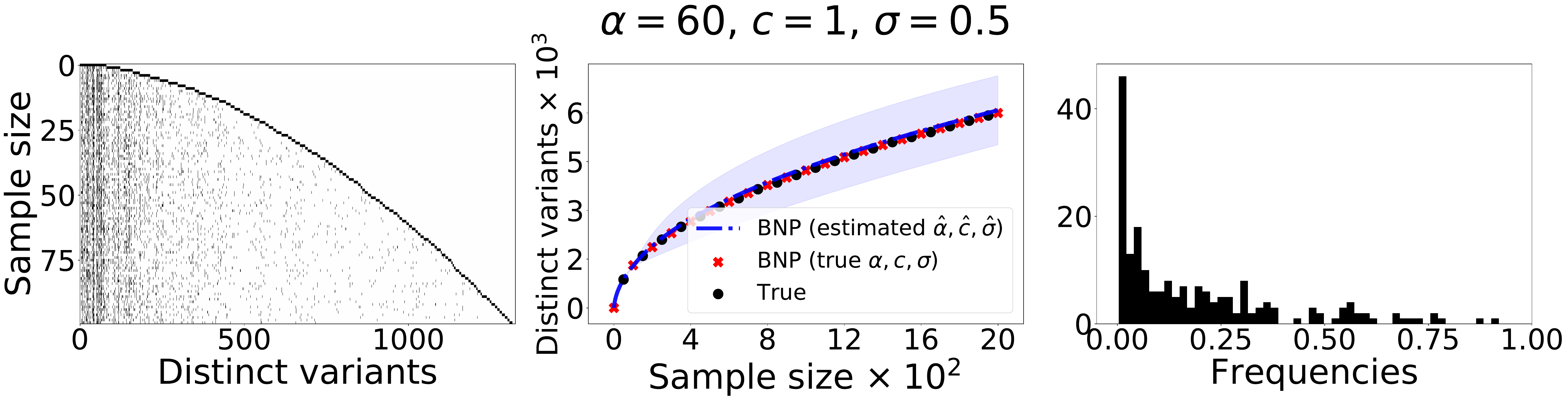}
      \caption{In this figure, we reproduce the visualizations explained in \Cref{fig:app_bnp_0} for a draw from a three-parameter Indian buffet process with parameters $\alpha = 60$, $c=1$ and $\sigma = 0.5$}
\label{fig:app_bnp_2}
\end{figure}

We also tested the performance of the predictor $\predtimes{N}{M}{r}$ for the number of new variants that are going to appear a given number $r$ of times as the initial sample of size $N$ is enlarged with $M$ additional observations. We found the performance of the predictor, in this case, to be very sensitive to the value of $\sigma$. In particular, while we expect the estimator to be exact as the extrapolation size $M$ diverges, we observe that when $\sigma$ is close to $0$, the performance degrades for small values of $r$. 
 \begin{figure}
      \centering \includegraphics[width=\textwidth,height=\textheight,keepaspectratio]{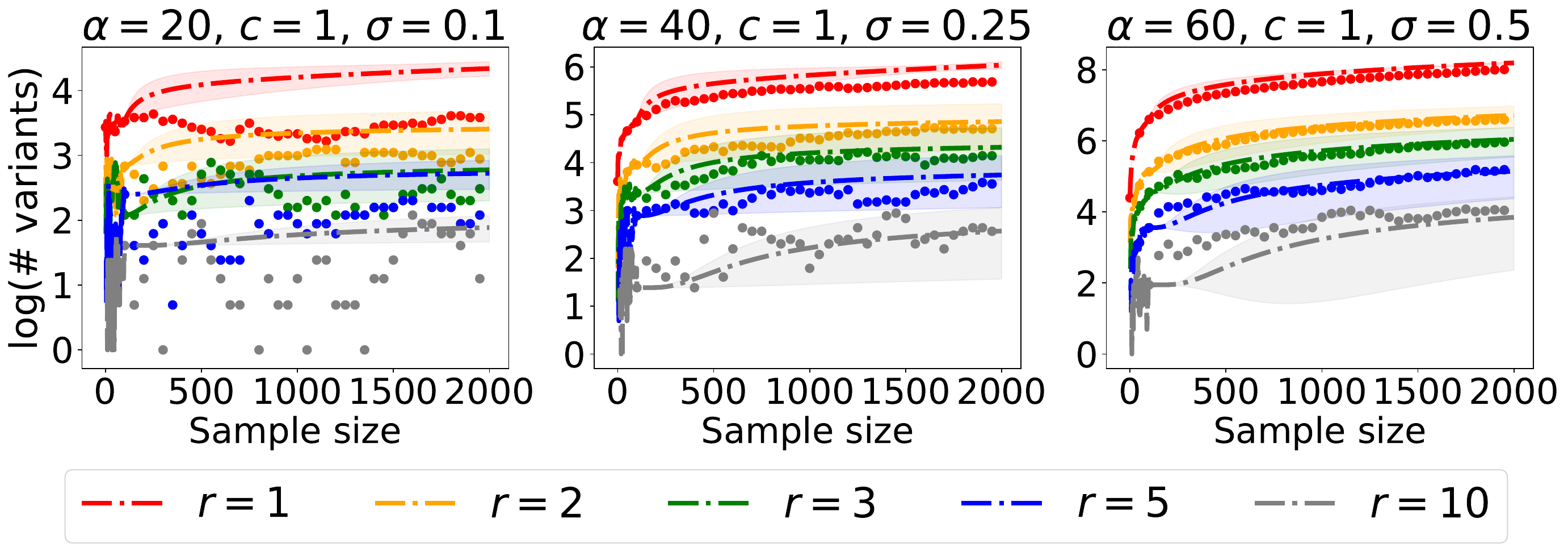}
      \caption{In this figure, for the three configurations of parameters $\alpha, c, \sigma$ considered in \Cref{fig:app_bnp_0}, \Cref{fig:app_bnp_1} and \Cref{fig:app_bnp_2}, we plot the performance of the estimators $\predtimes{N}{M}{r}$ for $M=1,\ldots,1900$ and $r= 1,2,3,5,10$. Dotted line show the performance of the estimators, while points show the true values of the process.}
\label{fig:app_bnp_3}
\end{figure}

\subsection{Synthetic data from the beta-Bernoulli model} \label{sec:app-exp_additional_ionita}

\textbf{\citet{ionita2009estimating} under the true model}: We first consider the case in which the variants frequencies $\theta_1,\ldots, \theta_{K}$ are independently and identically distributed\ draws from a beta distribution, i.e. for some parameters $\alpha>0$ and $\beta>0$, independently and identically distributed across $j=1,\ldots,K$ it holds
 \begin{align}\label{eq:beta}
     \theta_j \sim f(\theta) \propto \theta^{\alpha - 1} (1-\theta)^{\beta-1}\bm{1}_{(0,1)}(\theta).
 \end{align}
 Conditionally on the variants $\theta_1,\ldots,\theta_K$, each observation $X_n$ is a binary vector of independent Bernoulli random variables, $x_{n,j} \mid \theta_j \sim \Bern(\theta_j)$.
 This is exactly the model considered by \citet{ionita2009estimating}. Therefore we are not surprised to verify in \Cref{fig:ionita_0} that the predictor derived in \Cref{sec:app-ionita} outperforms the Bayesian nonparametric counterpart when the variants comes from the model of \Cref{eq:beta}.
 \begin{figure}
      \centering \includegraphics[width=\textwidth,height=\textheight,keepaspectratio]{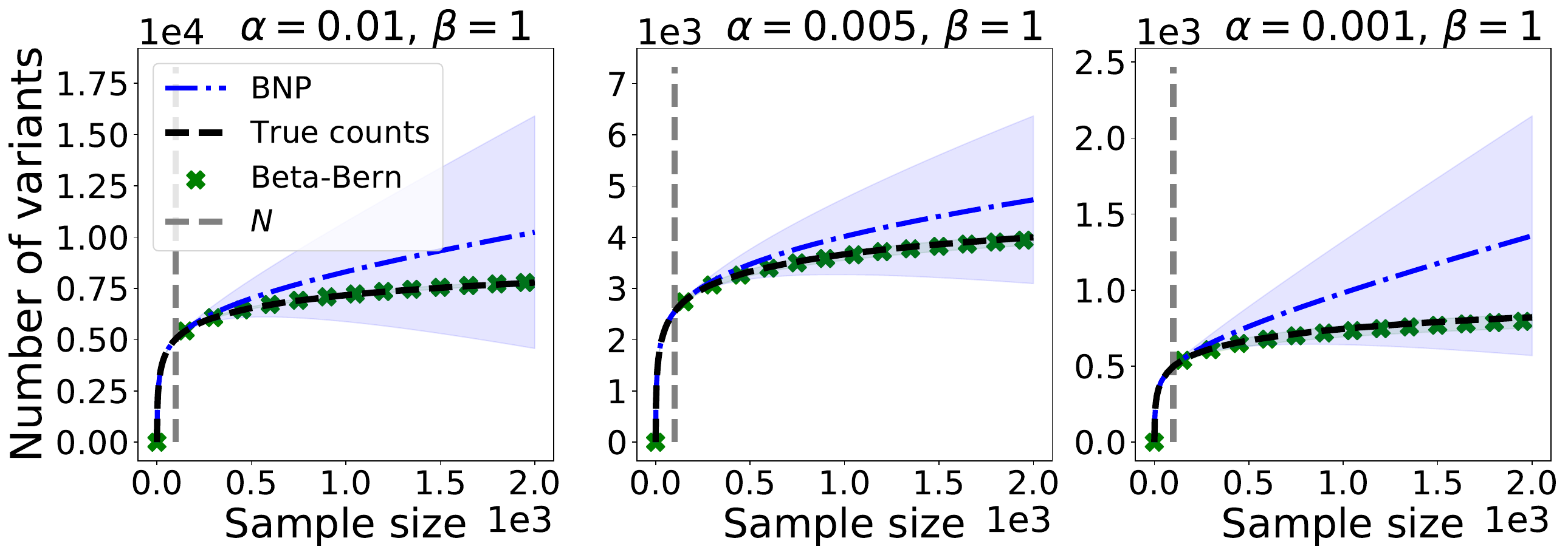}
      \caption{Performance of the beta-Bernoulli predictor (green crosses) proposed by \citet{ionita2009estimating} and of the nonparametric Bayesian predictor (dotted blue line) on three different datasets (each panel represents a different dataset). Each dataset is generated as follows: we first draw a random vector $\bm{\theta}$ of dimension $K = 10^4$. The $K$ coordinates are independently and identically distributed draws from a beta distribution. Conditionally on $\bm{\theta}$, we draw a random matrix $\bm{X}$ with $N=2000$ rows and $K$ columns. The $(n,j)$-th entry  $x_{n,j}$ is Bernoulli distributed with mean $\theta_j$, so that the columns of $\bm{X}$ are independently and identically distributed. We retain the first $N=200$ rows as training set and obtain the two estimators. We project up to $N+M=2000$ observations. We repeat the procedure over ten distinct folds of the data of the same size $N$ to produce estimates of the prediction error. This estimate of the error is displayed by plotting one empirical standard deviation across the ten predicted values across the different folds, for each extrapolation value $\ell = 201,\dots,2000$. From left to right, we vary the first shape parameter of the beta distribution $\alpha ~\in~ \{10^{-1}, 10^{-2}, 10^{-2}\times 2^{-1} \}$, driving the mean of the distribution to zero, while keeping the second parameter $\beta = 1$ fixed.}
\label{fig:ionita_0}
\end{figure}

\textbf{\citet{ionita2009estimating} under misspecification: the case of power laws}: Here we consider the case in which the variants frequencies $\theta_1,\ldots, \theta_{K}$ are independently and identically distributed\ draws from a power law distribution, i.e. for some tail exponent $\xi \geq 0$
 \begin{align}\label{eq:power_law}
     \theta_j \sim f(\theta) \propto \theta^{-\xi} \bm{1}_{(0,1)}(\theta).
 \end{align}
The parameter $\xi$ controls the left tail of the distribution: for $\xi = 0$, the distribution is uniform over the support $[0,1]$. The larger the value of $\xi$, the more mass we put over rare frequencies. Power laws arise in a vast number of natural phenomena, including ecology, biology, physical and social sciences \citep{clauset2009power}. Therefore, having an estimator that is effective when frequencies exhibit a power law behavior is desirable for virtually any applied scenario. In our experiments, the Bayesian parametric approach works well for moderate exponents, i.e.\ when the power law behavior is relatively mild. However, as soon as the exponent $\xi$ becomes large, the parametric model fails to deliver consistent results (see \Cref{fig:ionita_zipf}). Conversely, the Bayesian nonparametric estimator performs reasonably well.

 \begin{figure}
      \centering \includegraphics[width=\textwidth,height=\textheight,keepaspectratio]{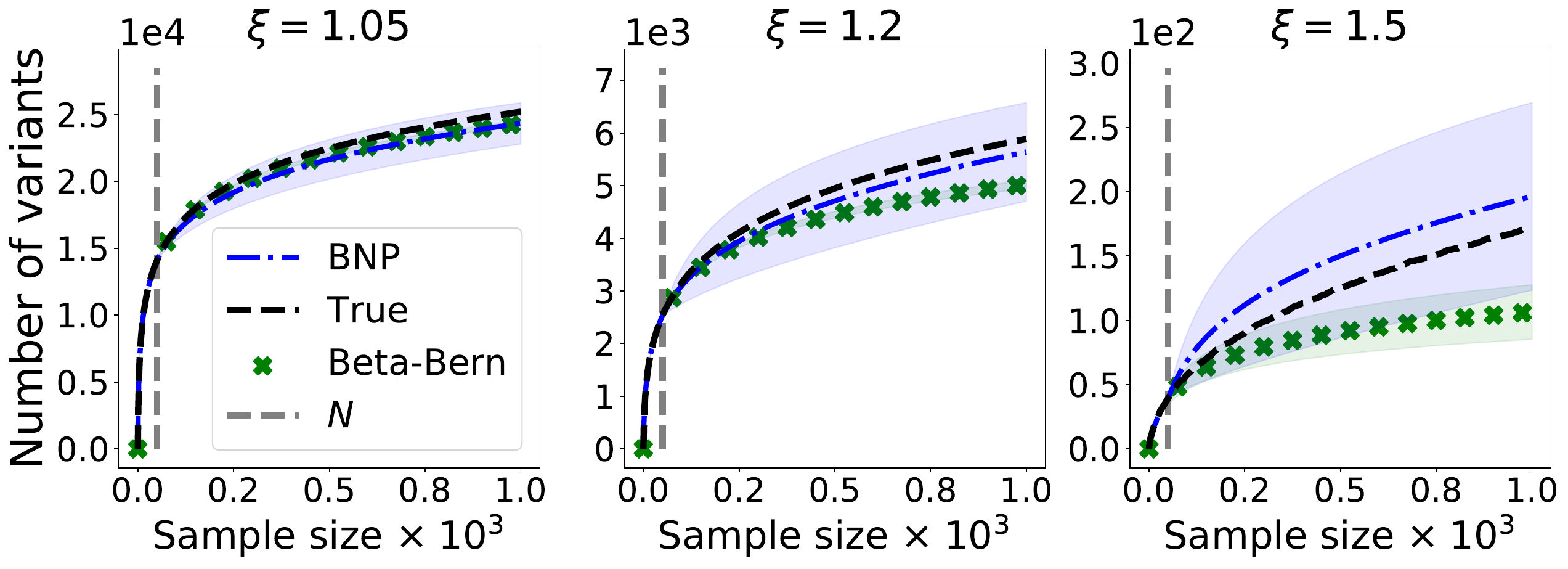}
      \caption{Performance of the beta-Bernoulli predictor (green solid line) proposed by \citet{ionita2009estimating} versus the nonparametric Bayesian predictor (dotted blue line) on three different datasets (each panel represents a different dataset). Each dataset is generated as follows: we first draw a random vector $\bm{\theta}$ of dimensions $K = 10^4$. The $K$ coordinates are independently and identically distributed draws from a power law distribution as described in \Cref{eq:power_law}. Conditionally on $\bm{\theta}$, we draw a random matrix $\bm{X}$ with $N=1000$ rows and $K$ columns. The $(n,j)$-th entry  $x_{n,j}$ is Bernoulli distributed with mean $\theta_j$, so that the columns of $\bm{X}$ are independently and identically distributed. We retain the first $N=50$ rows as training set and obtain the two predictors. We project up to $N+M=1000$ observations. We repeat the procedure over ten folds of the same data  of same size $N=50$. We repeat the procedure over ten folds of the same data to produce estimates of the prediction error. This estimate of the error is displayed by plotting one empirical standard deviation across the ten predicted values across the different folds, for each extrapolation value $\ell = 51,\dots,1000$. From left to right, we vary the exponent of the power law distribution (left, $\xi = 1.05$, center, $\xi = 1.2$, right $\xi = 1.5$).}
\label{fig:ionita_zipf}
\end{figure}

\subsection{The choice of the hyperparameter $\kappa$ for the frequentist nonparametric estimator proposed by \citet{zou2016quantifying}} \label{sec:app-exp_additional_zou}

Choosing the parameter $\kappa$ is particularly challenging when the sample size $N$ is small relative to the total number of frequencies - as in the genomics application we consider. As a general principle, in order to avoid numerical instability, the input size has to be sufficiently large. For example, given a sample of $N=100$ observations, if one sets $\kappa = 1$, the algorithm will take as an input only the number of variants which have been observed once. This will typically lead to numerical instability, which  will not arise for larger values of $\kappa$ (e.g.\ $\kappa \geq 10$). A general rule of thumb one could follow is to decrease $\kappa$ as a function of the training sample size $N$: the larger $N$, the smaller $\kappa$. While this intuition seems to work on some instances, we found cases in which unpredictable behaviors can affect the quality of the predictions (see \Cref{fig:zou_common} and \Cref{fig:zou_rare}).

 \begin{figure}
      \centering \includegraphics[width=\textwidth,height=\textheight,keepaspectratio]{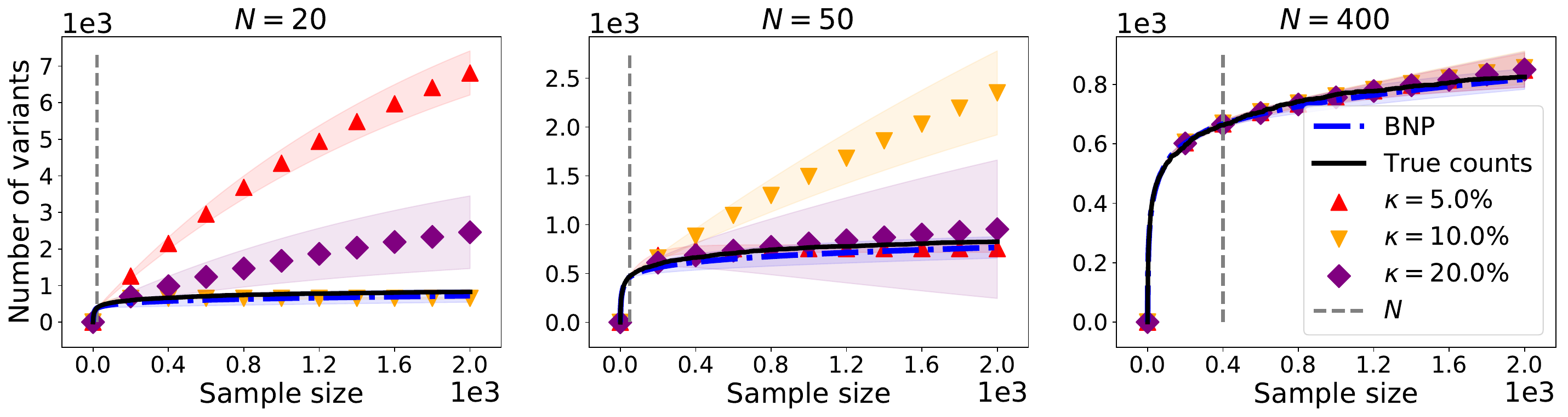}
  \caption{Comparison of the Bayesian nonparametric estimator (blue dotted line) to the frequentist nonparametric estimator proposed by \citet{zou2016quantifying}. We generate synthetic datasets as follows: we first draw a random vector $\bm{\theta}$ of $K=10^4$ independently and identically distributed beta random variables with parameters $\alpha = 0.001$ and $\beta = 1$. Conditionally on $\bm{\theta}$, we draw a random matrix $\bm{X}$ with $N=2000$ rows and $K$ columns. In each subplot, we retain a different fraction of rows of $\bm{X}$ to be used as training set (from left to right, $N \in \{20,50,400\}$). For each value of $N$, we compute the Bayesian nonparametric estimator, as well as the frequentist nonparametric estimator, varying the threshold parameter $\kappa \in \{5\%, 10\%, 20\%\}$ (red $(+)$, orange $(\star)$, purple $(\Diamond)$) respectively. We highlight how the performance of the frequentist nonparametric estimator, especially when $N$ is small, highly depends on the choice of $\kappa$, in an counterintuitve and somewhat unpredictable way. For example, when $N  =20$, choosing $\kappa = 10\%$ provides much better results than $\kappa = 5\%$ or $\kappa = 20\%$. However, for $N=50$, both $\kappa = 5\%$ and $\kappa = 20\%$ perform much better than $\kappa = 10\%$. As $N$ increases, the performance of the nonparametric frequentist estimator stabilizes, and becomes less sensitive to the choice of the parameter $\kappa$.} 
\label{fig:zou_common}
\end{figure}

 \begin{figure}
      \centering \includegraphics[width=\textwidth,height=\textheight,keepaspectratio]{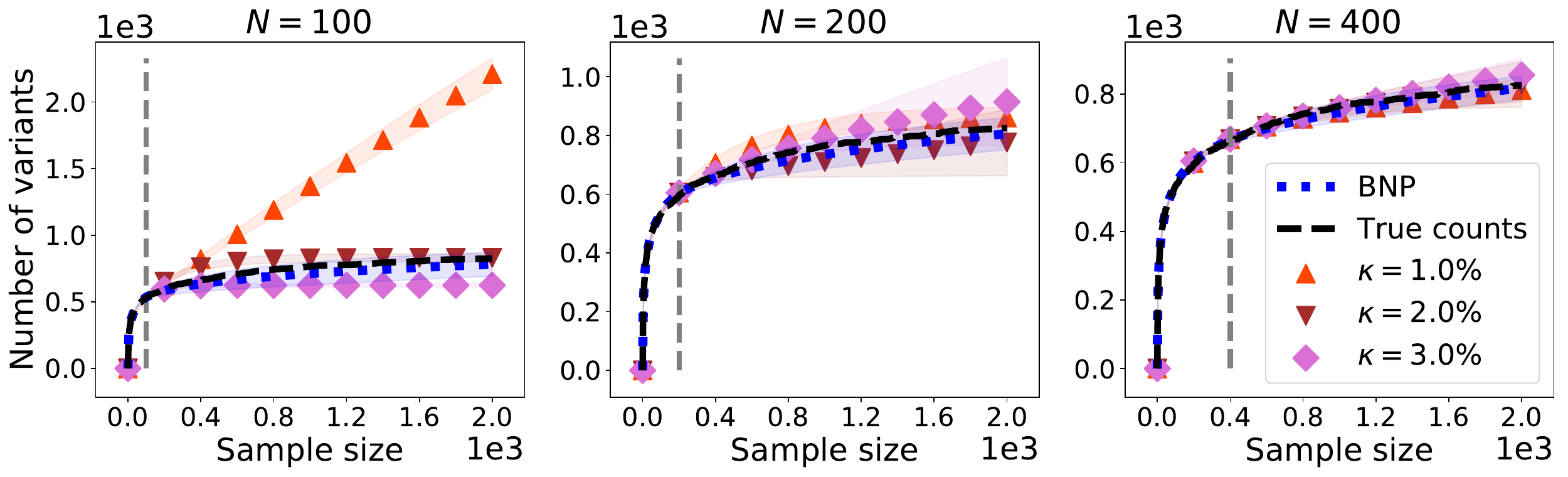}
    \caption{Comparison of the Bayesian nonparametric estimator (blue dotted line) to the frequentist nonparametric estimator of \citet{zou2016quantifying}. We use the same data showed in \Cref{fig:zou_common} but using much smaller values of $\kappa \in \{1\%, 2\%, 3\%\}$. Trying to run the linear program for these values of $\kappa$ and $N<100$ causes issues in the optimization routine, and therefore we only test it for $N$ sufficiently large. We notice that for both $N=100$ and $N=200$, the suggested value of $\kappa = 1\%$ provides worse results than choosing a larger value of $\kappa$, whereas for $N=400$, the performance of the estimator becomes less sensitive to the choice of $\kappa$.}
\label{fig:zou_rare}
\end{figure}

\subsection{Bias variance trade-off for the Jackknife estimator and optimal choice of the order $p$} \label{sec:app-exp_additional_jack}

As discussed  in \citet{burnham1978estimation, gravel2011demographic, gravel2014predicting}, and briefly in \Cref{sec:choice_jack}, the prediction quality of jackknife estimators  crucially depends on the \emph{order} chosen. Lower orders can suffer from large bias, but have small variance, while higher orders incur in less bias at the  cost of higher variance. On different datasets, the accuracy of different orders can vary dramatically.
In this section we provide  some experimental results (see \Cref{fig:jack_beta}) to illustrate this trade-off. 

In this section we provide  some experimental results (see \Cref{fig:jack_beta}) to illustrate this trade-off. 
 \begin{figure}
      \centering \includegraphics[width=\textwidth,height=\textheight,keepaspectratio]{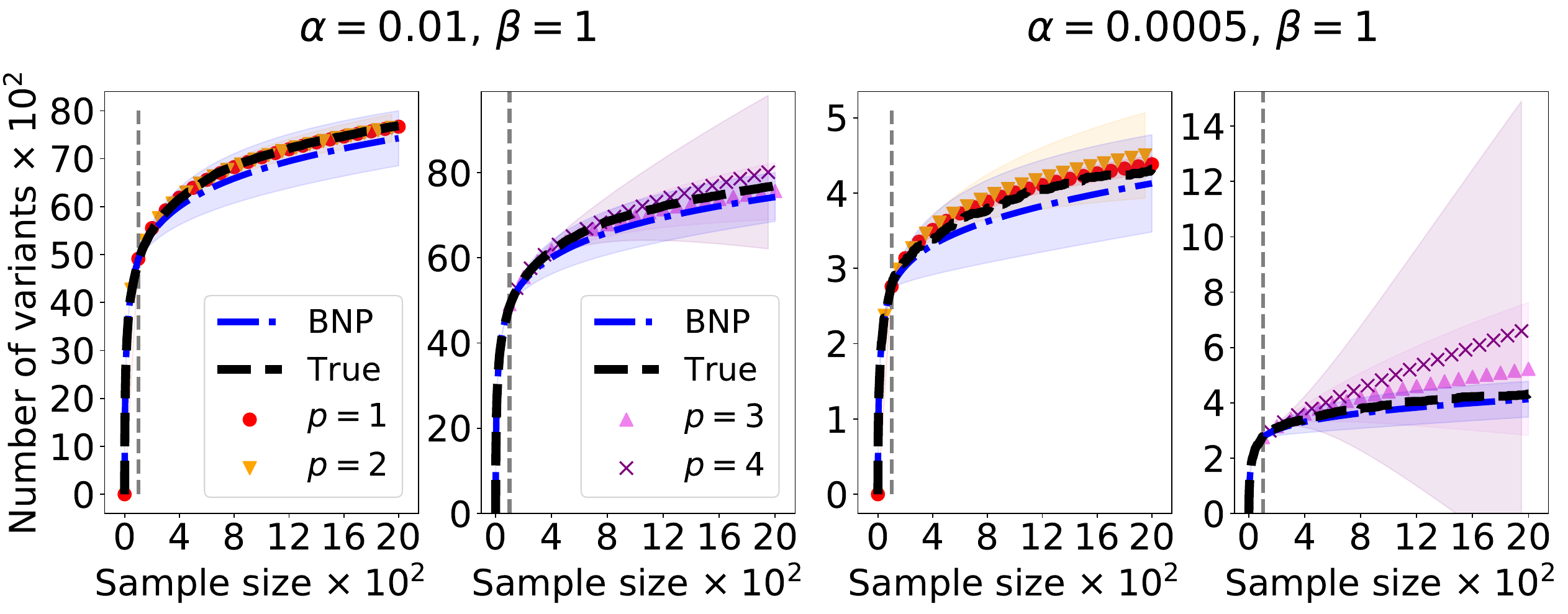}
    \caption{Comparison of the Bayesian nonparametric estimator (blue dotted line) to the jackknife estimator of \citet{gravel2014predicting} for different choices of the order $p$. We generate two datasets as follows: for $\alpha \in \{0.01, 0.0005\}$ and $\beta =1$, we generate two sets of $K=10^4$ independently and identically distributed beta distributed draws $\bm{\theta}$ with parameters $\alpha, \beta$. We then draw a random matrix $\bm{X}$ with $N = 2000$ rows, in which each entry $x_{n,j}$ is Bernoulli distributed with mean $\theta_j$. We retain $N=100$ rows for training. The two left panels show results for the dataset obtained when $\alpha = 0.01, \beta = 1$ across different choices of the jackknife order $p$. The two right panels show the same results for the dataset obtained when $\alpha = 0.0005$. Lower order jackknife estimators perform extremely well, and have little variance, while higher order jackknife estimators have worse performance, and higher variance. Such behavior worsens as $\alpha$ gets smaller, i.e.\ when the mean of the beta draws approach $0$.}
\label{fig:jack_beta}
\end{figure}

\subsection{Analysis of the Good-Toulmin estimator and the case of power laws} \label{sec:app-exp_additional_gt}

Last, we performed  synthetic experiments to understand the behavior of the smoothed Good-Toulmin estimator proposed by \citet{orlitsky2016optimal} and recently used by \citet{chakraborty2019somatic}. In our experiments, we considered differed regimes for the data generating process and found that  the estimator performs very well  when the distribution over variants' frequencies  does not put too much mass on very small variants. This holds true even for moderate and small sample size $N$ (see \Cref{fig:GT_1}) and well beyond the $M=N\log N$ extrapolation limit. However, when the vast majority of variants' are very rare, the estimator struggles to produce reliable results (see \Cref{fig:GT_2}). In all our experiments, we consider the two different smoothing choices suggested in \citet{chakraborty2019somatic} (GT 1 corresponds to \Cref{eq:GT_1} and GT 2 corresponds to \Cref{eq:GT_2} in \Cref{sec:app-gt}).

 \begin{figure}
      \centering \includegraphics[width=\textwidth,height=\textheight,keepaspectratio]{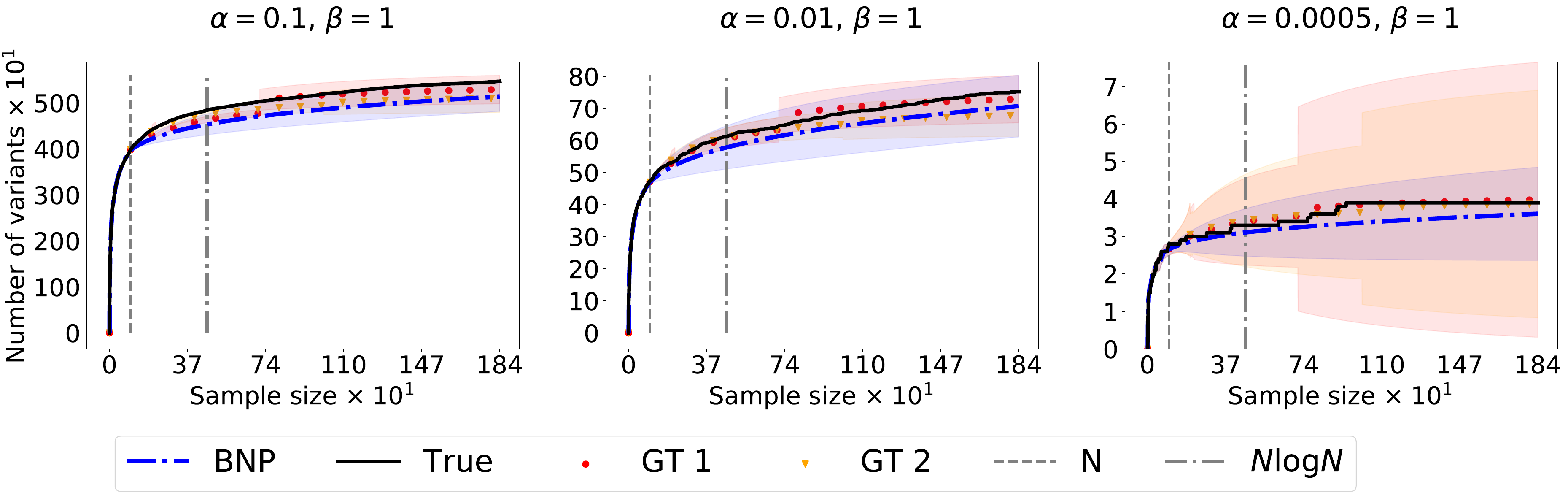}
      \caption{Prediction of the number of yet unseen variants for three synthetic datasets. Each dataset is generated as follows: we first draw a random vector $\bm{\theta}$ of dimension $K = 10^4$. The $K$ coordinates are independently and identically distributed draws from a beta distribution. Conditionally on $\bm{\theta}$, we draw a random matrix $\bm{X}$ with $N=2000$ rows and $K$ columns. The $(n,j)$-th entry  $x_{n,j}$ is Bernoulli distributed with mean $\theta_j$, so that the columns of $\bm{X}$ are independently and identically distributed. We retain the first $N=400$ rows as training set and obtain the two estimators. We project up to $N+M=2000$ observations. We repeat the procedure over ten distinct folds of the data of the same size $N$ to produce estimates of the prediction error. This estimate of the error is displayed by plotting one empirical standard deviation across the ten predicted values across the different folds, for each extrapolation value $\ell = 401,\dots,2000$. From left to right, we vary the first shape parameter of the beta distribution $\alpha ~\in~ \{10^{-1}, 10^{-2}, 10^{-2}\times 2^{-1} \}$, driving the mean of the distribution to zero, while keeping the second parameter $\beta = 1$ fixed.} 
\label{fig:GT_1}
\end{figure}

 \begin{figure}
      \centering \includegraphics[width=\textwidth,height=\textheight,keepaspectratio]{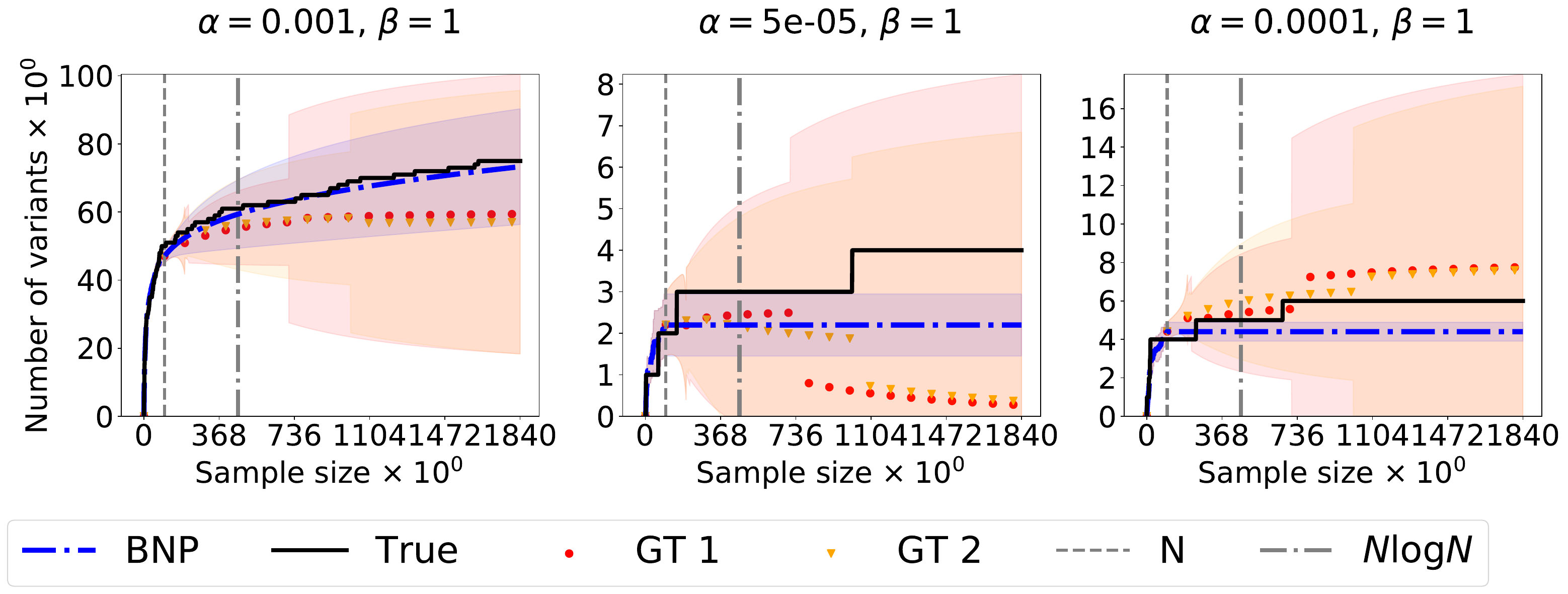}
      \caption{Prediction of the number of yet unseen variants for three synthetic datasets. Each dataset is generated as follows: we first draw a random vector $\bm{\theta}$ of dimension $K = 10^4$. The $K$ coordinates are independently and identically distributed draws from a beta distribution. Conditionally on $\bm{\theta}$, we draw a random matrix $\bm{X}$ with $N=2000$ rows and $K$ columns. The $(n,j)$-th entry  $x_{n,j}$ is Bernoulli distributed with mean $\theta_j$, so that the columns of $\bm{X}$ are independently and identically distributed. We retain the first $N=400$ rows as training set and obtain the two estimators. We project up to $N+M=2000$ observations. We repeat the procedure over ten distinct folds of the data of the same size $N$ to produce estimates of the prediction error. This estimate of the error is displayed by plotting one empirical standard deviation across the ten predicted values across the different folds, for each extrapolation value $\ell = 401,\dots,2000$. From left to right, we vary the first shape parameter of the beta distribution $\alpha ~\in~ \{10^{-3}, 5\times 10^{-5}, 10^{-4}\}$, driving the mean of the distribution to zero, while keeping the second parameter $\beta = 1$ fixed.} 
\label{fig:GT_2}
\end{figure}

We also considered the case in which the variants follow a power law distribution. In this case, we find that when the exponent of the power law is not too large (in absolute  value), then the estimator performs well. However, when the absolute value of the exponent satisfies $|\xi| > 1$, the estimator systematically underestimates the number of new variants to be seen, in the same way observed for real data (see \Cref{fig:GT_3}). Importantly, we verify that the empirical distribution of variants' frequencies across the datasets considered is indeed explained by power laws with exponent $|\xi|>1$.

 \begin{figure}
      \centering \includegraphics[width=\textwidth,height=\textheight,keepaspectratio]{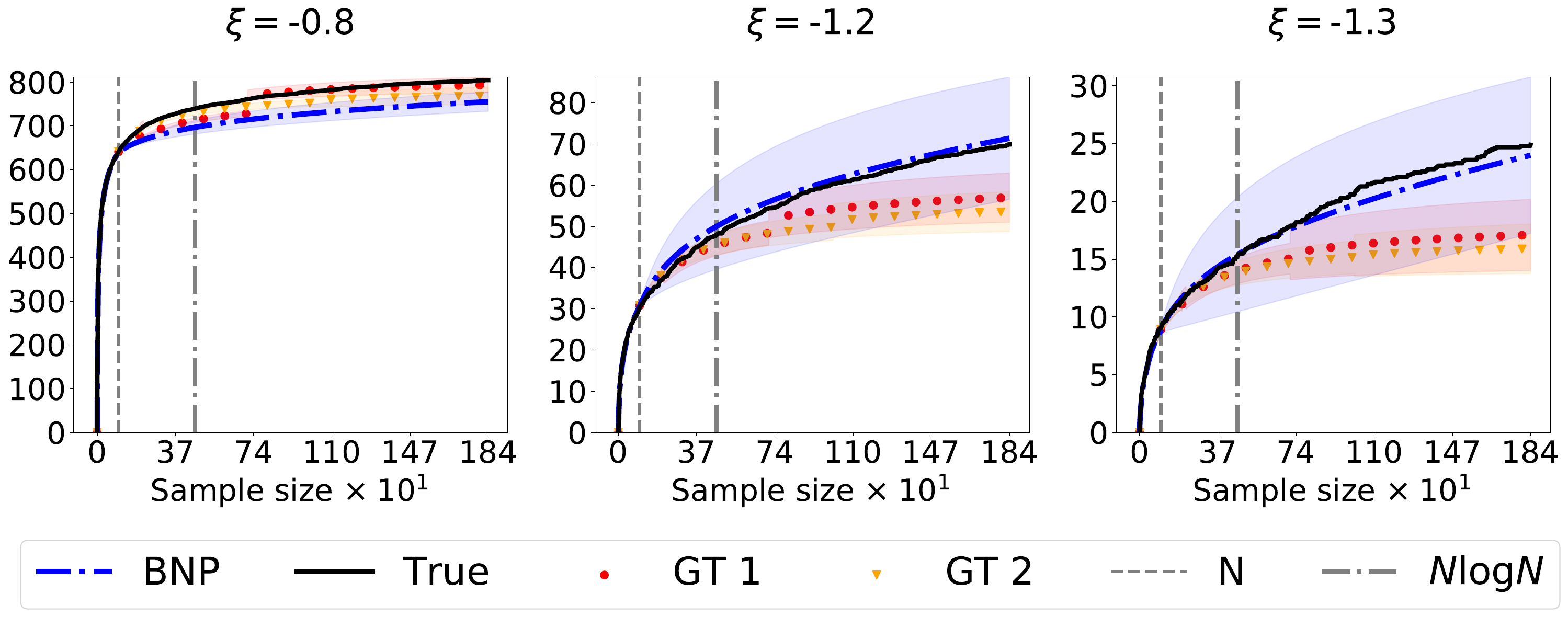}
      \caption{Performance of the smoothed Good-Toulmin predictor proposed by \citet{orlitsky2016optimal} versus the nonparametric Bayesian predictor (dotted blue line) on three different datasets (each panel represents a different dataset). Each dataset is generated as follows: we first draw a random vector $\bm{\theta}$ of dimensions $K = 10^4$. The $K$ coordinates are independently and identically distributed drawn from a power law distribution as described in \Cref{eq:power_law}. Conditionally on $\bm{\theta}$, we draw a random matrix $\bm{X}$ with $N=2000$ rows and $K$ columns. The $(n,j)$-th entry  $x_{n,j}$ is Bernoulli distributed with mean $\theta_j$, so that the columns of $\bm{X}$ are independently and identically distributed. We retain the first $N=400$ rows as training set and obtain the two predictors. We project up to $N+M=2000$ observations. We repeat the procedure over ten folds of the same data  of same size $N=400$. We repeat the procedure over ten folds of the data to produce estimates of the prediction error. This estimate of the error is displayed by plotting one empirical standard deviation across the ten predicted values across the different folds, for each extrapolation value $\ell = 401,\dots,2000$. From left to right, we vary the exponent of the power law distribution (left, $\xi = 1.05$, center, $\xi = 1.2$, right $\xi = 1.3$).} 
\label{fig:GT_3}
\end{figure}

 \begin{figure}
      \centering \includegraphics[width=\textwidth,height=\textheight,keepaspectratio]{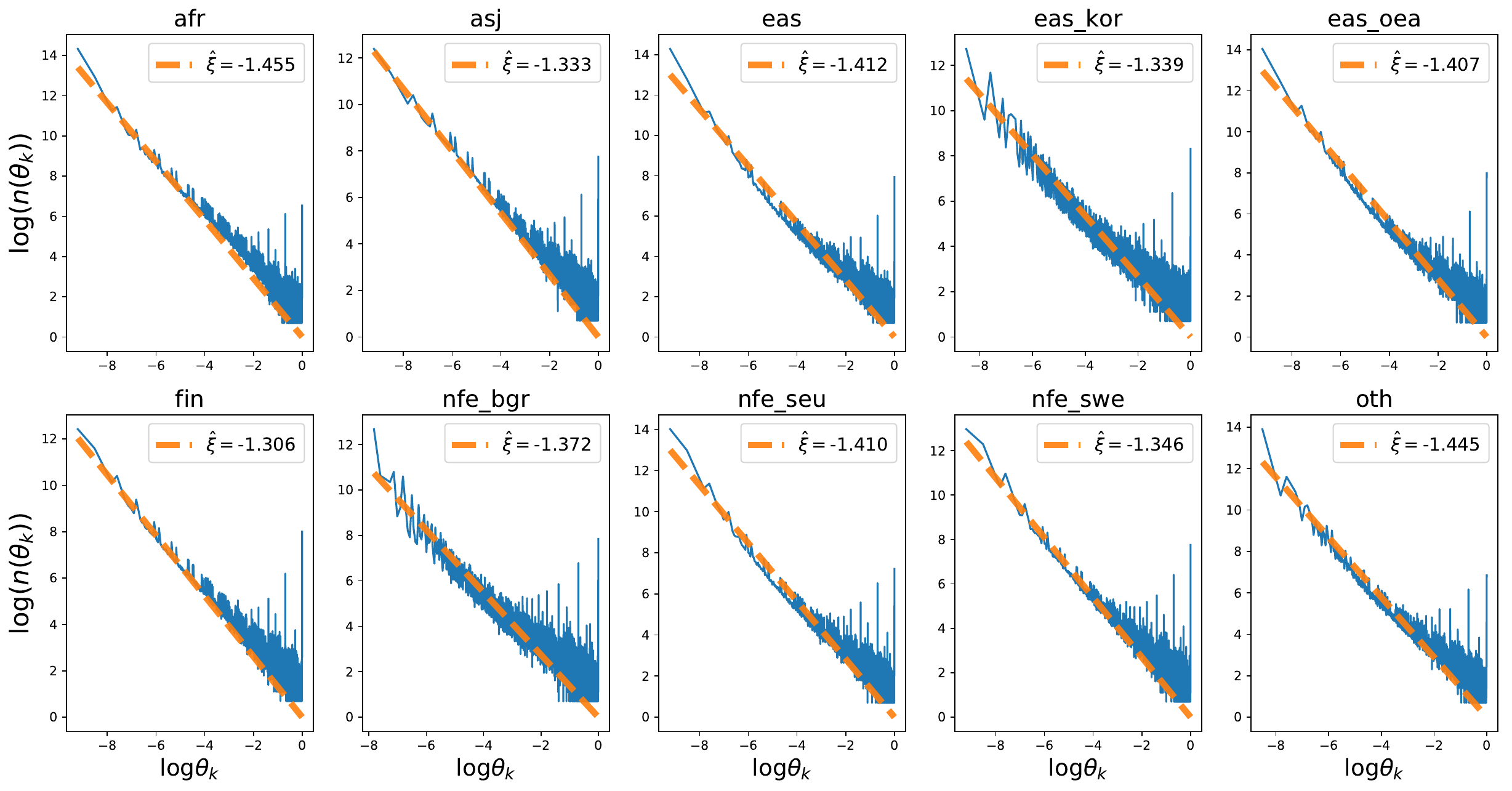}
      \caption{Fitting a power-law distribution to the empirical variants' frequencies of the gnomAD dataset. Following the method outlined in \citet{goldstein2004problems}, we fit a linear regression to the log-log plot of the binned empirical variants' frequencies to determine the exponent of the power law distribution. We only consider the 20 rarest frequencies, as  suggested in \citet{goldstein2004problems}. In all the datasets considered, we find that the estimate of  $\xi$ is larger than $1$, the  regime in which the smoothed Good-Toulmin provides systematic underestimation even on synthetic data.} 
\label{fig:GT_4}
\end{figure}


\end{document}